%% file: Modules_PQ-tree.tex
\documentclass[11pt,letter,leqno]{amsart}
\usepackage{amsmath,amsfonts,amssymb,amsthm,url,fdsymbol}
\usepackage{graphicx, color}
\usepackage{enumitem}
\usepackage{tikz}
\usetikzlibrary{shapes.geometric}
\usepackage[left=1in,right=1in,top=1in,bottom=1in,foot=0.5in]{geometry}
\usepackage{hyperref}
\usepackage{xcolor}
\hypersetup{
    colorlinks,
    linkcolor={red!70!black},
    citecolor={blue!50!black},
    urlcolor={blue!80!black}
  }
\usepackage{algorithmicx}
\usepackage{algpseudocodex}
\usepackage[capitalize,nameinlink]{cleveref}
\usepackage[final]{showlabels} % remplacer inline par final pour supprimer les labels de la sortie
\usepackage{newfloat}
\DeclareFloatingEnvironment{algorithm}

\newtheorem{theorem}{Theorem}[section]
\newtheorem{lemma}[theorem]{Lemma}
\newtheorem{claim}[theorem]{Claim}
\newtheorem{proposition}[theorem]{Proposition}
\newtheorem{corollary}[theorem]{Corollary}

\theoremstyle{definition}
\newtheorem{example}[theorem]{Example}
\newtheorem{remark}[theorem]{Remark}
\newtheorem{definition}[theorem]{Definition}

\Crefname{claim}{Claim}{Claims}

\newcommand{\op}{\ensuremath{{\rm op}}}

\newcommand{\M}{\mathcal{M}}
\newcommand{\Mmax}{{\mathcal{M}_{\mathrm{max}}}}
\newcommand{\PP}{\mathcal{P}}
\newcommand{\C}{\mathcal{C}}
\newcommand{\Tt}{\mathcal{T}}

\newcommand{\TPQ}{\mathcal{T}_{PQ}}
\newcommand{\TM}{\mathcal{T}_{M}}
\newcommand{\TD}{\mathcal{T}_{D}}
\newcommand{\hd}{\widehat{d}}
\newcommand{\Td}{\mathcal{T}_{\widehat{d}}}

\newcommand{\HH}{\mathcal{H}}

\newcommand{\co}[1]{\overline{#1}}
\newcommand{\Gdelta}{G_{\overline{\delta}}} %%%%%\bar
\newcommand{\Gdeltas}{G_{\delta^\ast}}
\newcommand{\deltas}{\rho}%{{\delta^\ast}}
\newcommand{\rhos}{\deltas}
\newcommand{\Grho}{G_{\overline{\rho}}} %%%%%\bar

\newcommand{\MST}{T}
\newcommand{\<}{<}%{\prec}
\newcommand{\cleq}{\leq}%{\preccurlyeq}

\newcommand{\rev}[1]{<_{\overleftarrow{#1}}}
\newcommand{\revS}{\rev{S}}

\newcommand\restr[2]{{% we make the whole thing an ordinary symbol
 \left.\kern-\nulldelimiterspace % automatically resize the bar with \right
 #1 % the function
 \vphantom{\big|} % pretend it's a little taller at normal size
 \right|_{#2} % this is the delimiter
 }}

\DeclareMathOperator{\diam}{diam}

\newcommand{\call}[1]{\ensuremath{\mathsf{#1}}}
\newcommand{\variable}[1]{\ensuremath{\mathit{#1}}}
\newcommand{\Let}{\State {\bfseries let} }

\newcommand{\Yield}{\State {\bfseries yield} }
\newcommand{\YieldAll}{\State {\bfseries yield from} }
\renewcommand{\Return}{\State {\bfseries return} }

\newcommand{\When}{\textbf{when}}
\makeatletter
\algnewcommand\algorithmicmatch{\textbf{match}}
\algnewcommand\algorithmicwith{\textbf{with}}
\algnewcommand\algorithmiccase{\textbf{case}}
\algnewcommand\algorithmiccolon{\textbf{:}}
\algnewcommand\algorithmicmatchend{\textbf{end}}
\algdef{SE}[MATCH]{Match}{EndMatch}[1]{%
	\algpx@startIndent\algpx@startCodeCommand\algorithmicmatch\ #1\ \algorithmicwith%
}{%
	\algpx@endIndent\algpx@startEndBlockCommand\algorithmicend\ \algorithmicmatch%
}
\algdef{C}[MATCH]{MATCH}{Case}[1]{%
	\algpx@endIndent[1]\algpx@startCodeCommand\algorithmiccase\ #1\algorithmiccolon%
}
\algtext*{EndMatch}
\apptocmd{\EndMatch}{\algpx@endIndent}{}{}
\pretocmd{\Case}{\algpx@endCodeCommand}{}{}
\pretocmd{\Match}{\algpx@endCodeCommand}{}{}
\pretocmd{\EndMatch}{\algpx@endCodeCommand[1]}{}{}
\makeatother

\algrenewcommand\algorithmicrequire{\textbf{Input:}}
\algrenewcommand\algorithmicensure{\textbf{Output:}}

\newcommand{\dtopq}{\call{deltaPqTree}}
\newcommand{\pqtom}{\call{mmodTree}}
\newcommand{\mtopq}{\call{pqTree}}
\newcommand{\findBip}{\call{findBipartition}}
\newcommand{\refine}{\call{refine}}
\newcommand{\drefine}{\call{pivotTree}}
\newcommand{\partRefine}{\call{refinePart}}
\newcommand{\refineDendro}{\call{refineTree}}
\newcommand{\stablePart}{\call{stablePartition}}
\newcommand{\stableDendro}{\call{stableTrees}}

\newcommand{\insertDendo}{\call{insert}}
\newcommand{\dendo}{\call{dendrogram}}
\newcommand{\lca}{\call{lca}}
\newcommand{\propagate}{\call{propagate}}
\newcommand{\mmodTree}{\call{mmoduleTree}}
\newcommand{\insertDico}{\call{insert}}
\newcommand{\containsKey}{\call{containsKey}}
\newcommand{\values}{\call{values}}
\newcommand{\getDico}{\call{get}}
\newcommand{\concatenate}{\call{concatenate}}
\newcommand{\join}{\call{join}}
\newcommand{\cptopqt}{\call{pqTree2}}
\newcommand{\cptopqtC}{\call{copointsToPqTree}}
\newcommand{\sortSiblings}{\call{nextFrontier}}

\newcommand{\Leaf}{\mathrm{Leaf}}
\newcommand{\Node}{\mathrm{Node}}

\newcommand\mdoubleplus{\ensuremath{\mathbin{+\mkern-7mu+}}}
\newcommand*{\append}{\mdoubleplus}
\newcommand*{\cons}{\cdot}
\newcommand*{\map}{\call{map}}

\begin{document}

\centerline{\bf\Large Modules and PQ-trees in Robinson
  Spaces %\footnote{This research was supported in part by ANR project
%    DISTANCIA (ANR-17-CE40-0015) and has received funding from
%    Excellence Initiative of Aix-Marseille - A*MIDEX (Archimedes
%    Institute AMX-19-IET-009), a French ``Investissements d’Avenir”
%    Programme.}
    }

\bigskip\bigskip \centerline{\sc \large M. Carmona$^{a,b}$, V.
Chepoi$^{b}$, G. Naves$^{b}$, and P.
Pr\'{e}a$^{a,b}$}

\bigskip\bigskip \centerline{$^{a}$LIS, École Centrale Marseille,
Marseille, France}

\medskip \centerline{$^{b}$LIS, Aix-Marseille Université, CNRS, and
Université de Toulon} \centerline{Marseille, France}

\medskip \centerline{ {\sf
\{mikhael.carmona, victor.chepoi, guyslain.naves,
pascal.prea\}@lis-lab.fr} }

\date{}

\begin{center} {\today} \end{center}

\bigskip\bigskip {\footnotesize {\bf Abstract.} A Robinson space is a
  dissimilarity space $(X,d)$ on $n$ points for which there exists a
  compatible order, {\it i.e.} a total order $\<$ on $X$ such that $x\<y\<z$
  implies that $d(x,y)\le d(x,z)$ and $d(y,z)\leq d(x,z)$. Recognizing
  if a dissimilarity space is Robinson has numerous applications in
  seriation and classification. A PQ-tree is a classical data
  structure introduced by Booth and Lueker to compactly represent a
  set of related permutations on a set $X$. In particular, the set of
  all compatible orders of a Robinson space are encoded by a PQ-tree.
  An mmodule is a subset $M$ of $X$ which is not distinguishable from
  the outside of $M$, {\it i.e.} the distances from any point of
  $X\setminus M$ to all points of $M$ are the same. Mmodules define
  the mmodule-tree of a dissimilarity space $(X,d)$. Given $p\in X$, a
  $p$-copoint is a maximal mmodule not containing $p$. The
  $p$-copoints form a partition of $X\setminus \{p\}$. There exist two
  algorithms recognizing Robinson spaces in optimal $O(n^2)$ time. One
  uses PQ-trees and one uses a copoint partition of $(X, d)$.

  In this paper, we establish correspondences between the PQ-trees and
  the mmodule-trees of Robinson spaces. More precisely, we show how to
  construct the mmodule-tree of a Robinson dissimilarity from its
  PQ-tree and how to construct the PQ-tree from the mmodule-tree. To
  establish this translation, additionally to the previous notions, we
  introduce the notions of $\delta$-graph $\Gdelta$ of a Robinson
  space and of $\delta$-mmodules, the connected components of
  $\Gdelta$. We also use the dendrogram of the subdominant ultrametric
  of $d$. All these results also lead to optimal $O(n^2)$ time
  algorithms for constructing the PQ-tree and the mmodule tree of
  Robinson spaces.

 }

\section{Introduction}\label{sec:introduction}

\input{sec_1_introduction}

\section{Preliminaries}\label{preliminaries}
\input{sec_2_preliminaries}
\section{Represented orders and nodes of a PQ-tree}\label{SUB_SECTION_PQ_tree_To_Mmodules}
\input{sec_3_ro_nodes}

\section{The graph \texorpdfstring{$\Gdelta$}{Gd} and the construction of the PQ-tree}\label{SECTION_Gdelta_PQ}
\input{sec_4_Gdelta_pq}

\section{Translation between PQ-trees and mmodule trees}\label{SECTION_translations}
\input{sec_5_translations}

\section{Construction of the mmodule tree using partition refinement}\label{SECTION_mmodtree}
\input{sec_6_mmodtree}

\section{Translation between copoint partitions and PQ-trees and mmodule trees}\label{SECTION_copoints_2_PQ_tree}
\input{sec_7_copoints_to_pq}

%\section{Balls, mmodules, and PQ-trees}\label{SECTION_balls_modules_PQ}
%\input{sec_8_balls}

\section{Conclusion}\label{SECTION_conclusion}
\input{sec_conclusion}

%\newpage

%\vspace{-1cm}

{
\bibliographystyle{siamplain}
\bibliography{citations1}
}

%\appendix
%\newpage
\input{appendix}

\end{document}

%% file: sec_1_introduction.tex
The classical seriation problem asks to find a simultaneous ordering
(or permutation) of the rows and the columns of the distance matrix
$D$ of a dissimilarity space $(X,d)$ with the objective that small
values should be concentrated around the main diagonal as closely as
possible, whereas large values should fall as far from it as possible.
This goal is best achieved by considering the so-called Robinson
property~\cite{Robinson}: a distance matrix $D$ is said to have the
Robinson property if the values of $D$ increase monotonically in the
rows and the columns when moving away from the main diagonal in both
directions. In case of $(0,1)$-matrices, the Robinson property is best
known as the Consecutive One Property. A Robinson space is a
dissimilarity space whose distance matrix can be transformed by
permuting its rows and columns to a distance matrix having the
Robinson property. The permutation which leads to a matrix with the
Robinson property is called a compatible order.

A PQ-tree is a classical data structure introduced by Booth and Lueker
\cite{BoothLueker} to efficiently encode a set of related permutations
on a finite set $X$. All compatible orders of a Robinson space
$(X, d)$ can be encoded by a PQ-tree. This fact was used by several
recognition algorithms for the Robinson spaces
\cite{AtBoHe,Prea,SestonThese}; the algorithm of~\cite{Prea} was the
first algorithm which recognizes Robinson spaces on $n$ points in
optimal $O(n^2)$ time. Even if optimal, the algorithm of~\cite{Prea}
is far from being simple. Recently, in~\cite{CaCheNaPre} we designed a
simple and practical divide-and-conquer algorithm for recognition of
Robinson spaces in optimal $O(n^2)$ time. This algorithm is based on
the notions of mmodules and copoint partitions of dissimilarity
spaces. An mmodule of a dissimilarity space $(X, d)$ (generalizing the
notion of a module in graph theory) is a subset $M$ of $X$ which is
not distinguishable from the outside of $M$, i.e., the distances from
any point of $X\setminus M$ to all points of $M$ are the same.
Mmodules define the mmodule-tree of a dissimilarity space $(X,d)$. If
$p$ is any point of $X$, then $p$ and the maximal by inclusion
mmodules of $(X,d)$ not containing $p$ define a partition of $X$,
which is called the copoint partition.

In this paper, we establish correspondences between the PQ-trees and
the mmodule-trees of Robinson spaces $(X, d)$. Namely, we show how to
derive the mmodule-tree from the PQ-tree of $(X,d)$ and, vice-versa,
how to construct the PQ-tree from the mmodule-tree of $(X,d)$. We also
show how to derive the branches of a PQ-tree from the copoint
partitions of $(X,d)$. We also describe optimal $O(n^2)$ algorithms
for constructing the PQ-tree and the mmodule-tree of a Robinson space
$(X,d)$. To establish the cryptomorphism between PQ-trees and
mmodules-trees, additionally to the previous notions, we introduce the
notion of $\delta$-graph $\Gdelta$ of a Robinson space $(X,d)$. We
prove that either $\Gdelta$ is connected for all $\delta>0$ or there
exists a unique value of $\delta$ for which $\Gdelta$ is not
connected. In the later case, the connected components of $\Gdelta$
are called $\delta$-mmodules. The dichotomy between the connectivity
for all $\delta$ and the non-connectivity for some $\delta$ of
$\Gdelta$ and the $\delta$-mmodules in the second case are crucial in
the construction of the PQ-tree and the translations between PQ-trees
and mmodule trees. The dendrogram $\Td$ of the ultrametric subdominant
$\hd$ of $(X,d)$ is yet another important ingredient, used in the
algorithm for the construction of the mmodule tree. Notice also that
the algorithm for the construction of the PQ-tree from the mmodule
tree uses the recognition of flat Robinson spaces
from~\cite{CaCheNaPre} as a subroutine. On the other hand, we present
an optimal $O(n^2)$ algorithm for constructing the PQ-tree, using the
$p$-copoints partition and the concept of $p$-proximity order,
also introduced in~\cite{CaCheNaPre}. Since this paper is a follow-up
of~\cite{CaCheNaPre}, we refer to the paper~\cite{CaCheNaPre} for a
complete bibliography on Robinson spaces and on their recognition
algorithms (in this paper, we cite only the recognition algorithms
using PQ-trees).

The rest of the paper is organized as follows. In
\Cref{preliminaries}, we present classical notions that we use:
Robinson dissimilarities, PQ-trees and mmodules. We end this section
with an illustrative example. In
\Cref{SUB_SECTION_PQ_tree_To_Mmodules} we characterize the sets of
total orders on $X$ that are representable by a PQ-tree and the
subsets of $X$ that correspond to nodes of the PQ-tree of a Robinson
space $(X,d)$. These results are used in the following two sections.
In \Cref{SECTION_Gdelta_PQ} we introduce the notion of a
$\delta$-graph $\Gdelta$ of a Robinson space and investigate the
properties of its $\delta$-mmodules. Using them, we show how to
construct the PQ-tree of a Robinson space. In
\Cref{SECTION_translations} we show how to construct for a Robinson
space its mmodule tree from its PQ-tree and, vice versa, its PQ-tree
from its mmodule tree. In \Cref{SECTION_mmodtree} we show how to build
the mmodule tree using partition refinement and the subdominant
ultrametric. In \Cref{SECTION_copoints_2_PQ_tree} we show how to build
the PQ-tree from any point $p$ and the copoint partition of $p$ and
the $p$-proximity order introduced in~\cite{CaCheNaPre}.
%We show and investigate
%the structure of its connected components, which we call $\delta$-mmodules. The graph $\Gdelta$ and its connected
%\Cref{SUB_SECTION_PQ_tree_To_Mmodules,SECTION_mmodules_2_PQ_tree,SECTION_copoints_2_PQ_tree}
%we show, given a Robinson space, how to construct, respectively, its
%mmodule-tree from its PQ-tree, its PQ-tree from its mmodule-tree and
%its PQ-tree from a copoint partition. In
%\Cref{SECTION_balls_modules_PQ}, we detail the links between three
%equivalence relations on Robinson spaces.

%% file: sec_2_preliminaries.tex
\subsection{Robinson dissimilarities}\label{ssec21}

Let $X=\{ p_1,\ldots, p_n\}$ be a set of $n$ elements, called
\emph{points}. A {\it dissimilarity} on $X$ is a symmetric function
$d$ from $X^2$ to the nonnegative real numbers such that $d(x,y) = 0$
if $x=y$. Then $d(x,y)$ is called the {\it distance} between $x,y$ and
$(X,d)$ is called a \emph{dissimilarity space}. A partial order on $X$
is called \emph{total} if any two elements of $X$ are
comparable. Since we will mainly deal with total orders, we
abbreviately call them \emph{orders}.

\begin{definition}[Compatible order]\label{def:compatorder}
  Given a dissimilarity space $(X,d)$, an order $\<$ on $X$ is {\it
    compatible} if $x\<y\<z$ implies that
  $d(x,z)\geq \max\{ d(x,y),d(y,z)\}.$ We denote by $\Pi(X,d)$
  the set of compatible orders of $(X,d)$. If $\<$ is a compatible
  order, then so is the order $\<^{\op}$ opposite to $\<$.
\end{definition}

\begin{definition}[Robinson space]\label{def:robinspace}
  A dissimilarity space $(X,d)$ is a \emph{Robinson space} if it
  admits a compatible order, i.e., if $\Pi(X,d)\ne\varnothing$. Then
  $(X,d)$ is said to be \emph{Robinson}.
\end{definition}
Equivalently, $(X,d)$ is Robinson if its distance matrix $D =
(d(p_i,p_j))$ can be symmetrically permuted so that its elements do
not decrease when moving away from the main diagonal along any row or
column. Such a dissimilarity matrix $D$ is said to have the {\it
  Robinson property}~\cite{CrFi,Di,DuFi,Robinson}.

If $Y \subseteq X$, we denote by $(Y,d)$ the dissimilarity space
obtained by restricting $d$ to $Y$; we call $(Y,d)$ a \emph{subspace}
of $(X,d)$. If $(X,d)$ is a Robinson space, then any subspace $(Y,d)$
of $(X,d)$ is also Robinson and the restriction of any compatible
order $\<$ of $X$ to $Y$ is compatible.

\begin{definition}[Block]\label{def:block}
Let $(X,d)$ be a Robinson space. A set $Y\subset X$ is called a {\em
    block} if $Y$ is an interval in any compatible order of $(X, d)$.
\end{definition}

The \emph{ball} of radius $r\ge 0$ centered at $x \in X$ is the set
$B_r(x)=\{y\in X: d(x,y)\leq r\}$. From the definition of compatible
orders it follows all balls of a Robinson space are blocks.  The
\emph{diameter} of a set $Y\subseteq X$ is
$\diam(Y)=\max\{ d(x,y): x,y\in Y\}$ and a pair $x,y\in Y$ such that
$d(x,y)=\diam(Y)$ is called a \emph{diametral pair} of $Y$.

Basic examples of Robinson dissimilarities are
the ultrametrics and line-distances. A \emph{line-distance} is provided
by the standard distance $d(p_i,p_j)=|p_i-p_j|$  between $n$ points
$p_1<\ldots <p_n$ of ${\mathbb R}$. Notice that any line-distance has
exactly two compatible orders: the order $p_1\<\ldots\<p_n$ defined by
the coordinates of the points and its opposite. This leads to the following
notion:

\begin{definition}[Flat Robinson space]\label{def:flat}
  A Robinson space is \emph{flat} is it has exactly two compatible
  orders, reverse of each other. All line-distances are flat but the
  converse is not true.
\end{definition}

\subsection{\texorpdfstring{$X$-trees}{X-trees}}\label{ssec22}
Given a finite set $X$, an \emph{$X$-tree} is an ordered rooted tree
$\Tt$ in which the exists a bijection between $X$ and the set of
leaves of $\Tt$ and any inner node of $\Tt$ has degree at least 2. We
will use Greek letters as variables over nodes of trees but for
convenience we will give the same name to the elements of $X$ and the
corresponding leaves of $\Tt$. As usually, we say that a node $\alpha$
is an ancestor of a node $\beta$ if $\alpha$ belongs to a unique path
of $\Tt$ between $\beta$ and the root. For two nodes $\alpha$ and
$\beta$, we denote by $\lca(\alpha, \beta)$ the \emph{lowest common
  ancestor} of $\alpha$ and $\beta$. Given a node $\alpha$ of $\Tt$,
we denote $X(\alpha) \subseteq X$ the set of leaves having $\alpha$ as
an ancestor. We say that two $X$-trees $\Tt$ and $\Tt'$ are
\emph{isomorphic} if there exists an isomorphism $f$ between the trees
$\Tt$ and $\Tt'$ such that $f(u)=u$ for any $u\in X$ and
$X(f(\alpha))=X(\alpha)$ for any inner node $\alpha$ of $\Tt$.

\begin{definition}[Pertinent node]\label{def:pertinent}
  Given $S \subset X$ and an X-tree $\Tt$ on $X$, we will say that a
  node $\alpha$ is $S$-{\em pertinent} or the {\em pertinent node} of
  $S$ if $\alpha$ is the lowest node such that $S \subseteq X(\alpha)$.
\end{definition}

In this paper we consider three types of $X$-trees: mmodule trees for
arbitrary dissimilarity spaces $(X,d)$, PQ-trees for Robinson spaces
$(X,d)$, and dendrograms for ultrametric spaces $(X,d)$.

\subsection{Ultrametrics}\label{ssec23}
Recall, that a dissimilarity space $(X,d)$ is an {\it ultrametric
  space} if it satisfies the three-point condition
$d(x,y)\le \mbox{max}\{ d(x,z),d(y,z)\}$ for all $x,y,z\in X$. Equivalently,
 the two largest distances among $d(x,y),d(x,z),$ and
$d(y,z)$ are equal. The ultrametrics are thoroughly used in phylogeny
and hierarchical clustering, because they can be represented by
particular $X$-trees, called dendrograms. A \emph{dendrogram} for an
ultrametric space $(X,d)$ is an $X$-tree $\TD$ with weighted edges
such that each inner node $\alpha$ has the same distance to all leaves
in the subtree rooted at $\alpha$. The sets $X(\alpha)$ for nodes
$\alpha$ of $\TD$ are called \emph{clusters} and the set system $\HH$
consisting of all clusters is called a \emph{hierarchy}. Any two
clusters of $\HH$ are either disjoint or one is included in the
another one. $\HH(X)$ is unique, and $\TD(X)$ is unique up to
reordering children of each node. The dendrogram $\TD$ and the weights
of its edges are constructed by the well-known \emph{single-linkage
  clustering algorithm} \cite{Ev,GoRo}. Then for any two points
$x,y\in X$, $d(x,y)$ equals the length of the unique path between $x$
and $y$ in $\TD$ or, equivalently, to twice the height of $x$ and $y$
in the subtree of $\TD$ rooted at the lowest common ancestor
$\lca(x,y)$ of $x$ and $y$.

We consider the dendrogram $\TD$ of an ultrametric space as
unweighted, but we weight each inner node $\alpha$ of $\TD$ by the
distance $d(x,y)$ between any two leaves $x,y\in X(\alpha)$ such that
$\lca(x,y)=\alpha$; we denote the weight of $\alpha$ by $p(\alpha)$.
Notice then when moving from any leaf to the root of $\TD$, the
weights of the nodes occurring on this path are strictly increasing.
The representation of an ultrametric by a dendrogram in which the
edges are weighted is called an \emph{equidistant representation} and
the representation in which the nodes are weighted is called a
\emph{vertex representation} \cite{SeSt}. The two representations are
equivalent: node weights correspond to potential, and the weight of an
edge is half the difference of potential between its two extremities.
It is straightforward to compute the potentials from the differences
of potential (setting the potential of any leaf to be $0$), and
\emph{vice versa}. Notice that, in the vertex representation, the
distance between two leaves $x$ and $y$ is equal to $p(\lca(x,y))$. In
particular, for any inner node $\alpha$ and $x,y \in X(\alpha)$ such
that $x$ and $y$ belong to different children of $\alpha$, we have
$d(x,y) = p(\alpha) = \diam(X(\alpha))$.

One fundamental property of ultrametrics is the existence for each
dissimilarity space $(X,d)$ of the \emph{subdominant ultrametric}
$\widehat{d}$ on $X$: for any ultrametric $d'$ on $X$ such that
$d' \le d$ (i.e., $d'(x,y)\le d(x,y)$ for all $x,y\in X$) we have
$d' \le \widehat{d}\le d$. The subdominant ultrametric $\widehat{d}$
can be defined as follows: for all $x,y \in X$, $\hd(x,y)$ is the minimum over
all paths $P$ between $x$ and $y$ of the maximum weight of an edge on that
path $P$:
$$ \hd(x,y) = \min \{ \max_{uv \in P} d(u,v) : \textrm{ $P$ is a $(x,y)$-path}\}.$$ 
In combinatorial optimization, $\hd(x,y)$  is called the \emph{bottleneck distance} between $x$ and $y$ and 
a $(x,y)$-path providing the minimum in the previous formula is called the 
\emph{bottleneck shortest path}~\cite{Schrijver}. 
The subdominant ultrametric $\widehat{d}$ can be constructed in the
following elegant way. Let $T$ be a minimum spanning tree of the
complete graph on $X$ weighted by the values of the dissimilarity
function $d$. Then for any $x,y\in X$, $\widehat{d}(x,y)$ is the
weight of the heaviest edge of the unique path of $T$ between $x$ and
$y$ \cite{GoRo}. We denote by $\Td$ the dendrogram of the ultrametric
space $(X,\widehat{d})$; sometimes we will say that $\Td$ is the
dendrogram of the dissimilarity space $(X,d)$. We can construct $\Td$
in the following iterative way (which seems us to be new): when using
Prim's algorithm to compute the minimum spanning tree $T$ of $(X,d)$,
one can build $\Td$ by inserting each vertex in $\Td$ at the moment
when it is visited, leading to \Cref{algo:dendo} presented in the
Appendix. This algorithm has complexity $O(n^2)$.

In this paper, we use ultrametrics as an illustrative example. In
fact, for ultrametrics their dendrograms have the same shape as their
PQ-trees and their mmodule trees. Additionally, we use the dendrogram
$\Td$ of the ultrametric subdominant $(X,\widehat{d})$ of a Robinson
space $(X,d)$ in our construction of the mmodule tree. The idea of
constructing the subdominant ultrametric via the minimum spanning tree
also occurs in our techniques.

\subsection{PQ-trees}\label{SUB_SECTION_PQ_tree}

A PQ-tree is a tree-based data structure introduced by Booth and
Lueker~\cite{BoothLueker} in 1976 to efficiently encode a family of
permutations on a set $X$ in which various subsets of $X$ occur
consecutively.

\begin{definition}[PQ-tree]
  A \emph{PQ-tree} over a set $X$ is an $X$-tree $\mathcal{T}$ whose
  internal nodes are distinguished as either P-nodes or Q-nodes.  Two
  PQ-trees are said to be \emph{equivalent} if one can be transformed
  into the other by applying a sequence of the following two
  equivalence transformations.
  \begin{itemize}
  \item[(1)] Arbitrarily permute the children of a P-node.
  \item[(2)] Reverse the children of a Q-node.
  \end{itemize}
  The nodes of arity 2 can be equally viewed as P-nodes and as
  Q-nodes; in our notations, they will be considered as P-nodes.
\end{definition}

We use the convention that P-nodes are represented by circles or
ellipses and Q-nodes are represented by rectangles. For PQ-trees
$\beta_1, \ldots, \beta_k$, we denote $P(\beta_1, \ldots,\beta_k)$ and
$Q(\beta_1, \ldots,\beta_k)$ the PQ-trees whose root are a P-node,
respectively a Q-node, with children $\beta_1,\ldots,\beta_k$ in that
order.

For convenience, we will abusively identify permutations and (total)
orders.  The \emph{canonical order} of a PQ-tree $\mathcal{T}$ over
$X$ is the permutation on $X$ obtained by a left-to-right traversal of
$\mathcal{T}$.

\begin{definition}[Represented permutations of a PQ-tree]
  The set of represented permutations of a PQ-tree $\mathcal{T}$ is
  the set of canonical orders of all PQ-trees equivalent to
  $\mathcal{T}$, and is denoted $\Pi(\mathcal{T})$. A set of
  orders/permutations $\Pi$ is called \emph{representable by a
    PQ-tree} if there exists a PQ-tree $\mathcal T$ such that
  $\Pi({\mathcal T})=\Pi$.
\end{definition}

Represented permutations may also be defined using composition of
orders, as defined below.

\begin{definition}[Composition of orders]
  Let $X$ be a set, $\mathcal{P}$ a partition of $X$,
  $\<_{\mathcal{P}}$ an order on $\mathcal{P}$, and for each part
  $S \in \mathcal{P}$, $\<_S$ an order on $S$. Then the
  \emph{composition} of $\<_{\mathcal{P}}$ and
  $(\<_S)_{S \in \mathcal{P}}$ is the order $\<$ defined by $x \< y$ when
  \begin{enumerate}[label=(\roman*)]
  \item either there is $S \in \mathcal{P}$ with $x, y \in S$, and $x \<_S y$,
  \item or there are $S, S' \in \mathcal{P}$ distinct, with $x \in S$,
    $y \in S'$ and $S \<_{\mathcal{P}} S'$.
  \end{enumerate}
\end{definition}

Then the permutations represented by a PQ-tree matches with orders
obtained by composing an order on the children of the root with a
choice of orders on each children of the root node, where the order on
the children is
\begin{itemize}[label=-]
\item in the case of a P-node: an arbitrary order,
\item in the case of a Q-node: either the order of the children in the
  PQ-tree or its reverse.
\end{itemize}

\begin{example} The PQ-tree $\mathcal T$ of \Cref{PQ-TREE-1} has one Q-node (the
  root) and one P-node $\alpha$ with $X({\alpha})=\{ 1,2,3\}$. The
  subtree rooted at $\alpha$ represents all the permutations of the
  elements $1,2,3$. Consequently, the PQ-tree $\mathcal T$ represents
  the 12 permutations of the form $(\pi,4,5,6,7)$ and
  $(7,6,5,4,\pi)$, where $\pi$ is any permutation on $\{ 1,2,3\}$.

 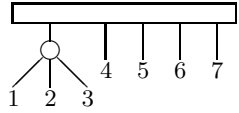
\begin{figure}[h!]
 \begin{center}
  \begin{picture}(110, 33) \setlength{\unitlength}{.7pt}
	\put(10,35){\framebox(120,10){}}
  \put(30,20){\circle{10}}
  \put(30,25){\line(0,1){10}} \put(60,15){\line(0,1){20}}
  \put(80,15){\line(0,1){20}} \put(100,15){\line(0,1){20}}
  \put(120,15){\line(0,1){20}} \put(10,0){\line(1,1){16}}
  \put(34,16){\line(1,-1){16}} \put(30,15){\line(0,-1){15}}
  \put(7,-10){\footnotesize{1}} \put(27,-10){\footnotesize{2}}
  \put(47,-10){\footnotesize{3}} \put(57,5){\footnotesize{4}}
  \put(77,5){\footnotesize{5}} \put(97,5){\footnotesize{6}} \put(117,5){\footnotesize{7}}
  \end{picture}
 \end{center}
 \caption{\label{PQ-TREE-1}\small{A PQ-tree}}
 \end{figure}
\end{example}

Préa and Fortin~\cite{Prea} used PQ-trees to encode the compatible
orderings of a Robinson dissimilarity space $(X,d)$, because this set
of orders is represented by a PQ-tree. We recall this correspondence.
A $(0,1)$-matrix $A$ has the {\it Consecutive Ones Property} ({\it
  C1P}) if its columns can be permuted in such a way that in all rows
the $1$s appear consecutively. Such an order is called {\it
  compatible}. If $A$ is a C1P-matrix, then the sets of all its
compatible permutations can be represented by a PQ-tree; Booth and
Lueker designed an iterative algorithm, using PQ-trees, which
determines if a matrix $M$ has the C1P~\cite{BoothLueker}. Let
${\mathbf B}$ denote the set of all distinct balls of a dissimilarity
space $(X,d)$. Let $M_{\mathbf B}$ be the $\{0,1\}$-matrix whose
columns are indexed by the points of $X$ and rows by the balls of
${\mathbf B}$: for $x\in X$ and $B\in {\mathbf B}$ we define
$M_{\mathbf B}(B,x)=1$ if $x\in B$ and $M_{\mathbf B}(B, x)=0$
otherwise. The following simple result of Mirkin and Rodin~\cite{MiRo}
links Robinson dissimilarities with C1P-matrices:

\begin{proposition}\cite{MiRo}\label{THEOREM_Robinson_Boules}
  A dissimilarity space $(X,d)$ is Robinson if and only if its matrix
  $M_{\mathbf B}$ satisfies the C1P. There exists a bijection between
  the set $\Pi(X,d)$ of orders compatible with $d$ and the set of permutations
  compatible with $M_{\mathbf B}$.
\end{proposition}

Since the sets of all compatible permutations of a C1P-matrix can be
represented by a PQ-tree~\cite{BoothLueker}, from
\Cref{THEOREM_Robinson_Boules} we obtain:

\begin{corollary}
  The set $\Pi(X,d)$ of all compatible orders of a Robinson space $(X,d)$ can be
  represented by a PQ-tree.
\end{corollary}

For a Robinson space $(X,d)$, we denote by $\TPQ(X,d)$, or $\TPQ$ for
short, its PQ-tree (unique up to equivalence).

\subsection{Mmodules and copoints} \label{mmodulesgeneral}

In this subsection, we recall the basic facts about the mmodules and
copoints in dissimilarity spaces from our paper \cite{CaCheNaPre}. Let
$(X,d)$ be a dissimilarity space.

\begin{definition}[Mmodule]
  A set $M \subseteq X$ is called an \emph{mmodule} (a \emph{metric
    module} or a \emph{matrix module}) if $M$ cannot be distinguished
  from outside, {\it i.e.}, for any $z \in X \setminus M$ and all
  $x,y \in M$ we have $d(z,x) = d(z,y)$.
\end{definition}

In graph theory, the subgraphs indistinguishable from the outside are
called modules, explaining our choice of the term ``mmodule''. Denote
by $\M(X,d)$ (or $\M$ fo short) the set of all mmodules of
$(X,d)$. Trivially, $\varnothing, X,$ and $\{ p\}, p\in X$ are
mmodules; we call them \emph{trivial mmodules}. An mmodule $M$ is
called \emph{maximal} if $M$ is a maximal by inclusion mmodule
different from $X$. Denote by $\Mmax(X,d)$ (or $\Mmax$ for short) the
set of all maximal mmodules of $(X,d)$. We continue with the basic
properties of mmodules.

\begin{proposition}\cite[Proposition 3.1]{CaCheNaPre}\label{mmodules}
Let $(X, d)$ be a dissimilarity space.
  The set $\M = \M(X,d)$ has the following properties:
  \begin{enumerate}[label=(\roman*)]
  \item\label{item:mmod1} $M_1, M_2 \in \M$ implies that $M_1 \cap M_2\in \M$;
  \item\label{item:mmod2} if $M \in \M$ and $M'\subset M$, then $M'\in \M$ if and
    only if $M'$ is an mmodule of $(M,d|_M)$;
  \item\label{item:mmod3} if $M_1, M_2 \in \M, $ $M_1 \cap M_2 \ne \varnothing$,
    then $M_1 \cup M_2 \in \M$, and if additionally
    $M_1 \setminus M_2 \ne \varnothing, M_2 \setminus M_1 \ne \varnothing$,
    then $M_1 \cup M_2, M_1 \setminus M_2, M_2 \setminus M_1,$ and
    $M_1 \Delta M_2$ are mmodules;
  \item\label{item:mmod4} the union $M_1 \cup M_2$ of two intersecting maximal
    mmodules $M_1, M_2\in \M$ is $X$;
  \item\label{item:mmod5} if $M_1$ and $M_2$ are two disjoint maximal mmodules and
    $M$ is a nontrivial mmodule contained in $M_1 \cup M_2$, then either
    $M \subset M_1$ or $M \subset M_2$;
  \item\label{item:mmod6} if $M_1,M_2 \in \M$ and $M_1 \cap M_2 = \varnothing$, then
    $d(u,v) = d(u',v')$ for any (not necessarily distinct) points
    $u, u' \in M_1$ and $v, v' \in M_2$.
  \end{enumerate}
\end{proposition}

The next results show that the mmodules of a dissimilarity can be
organized within a tree-structure. We say that a family of subsets
$\{M_1,\ldots,M_k\}$ of $X$ is a \emph{copartition} of $X$ if
$\{X \setminus M_1,\ldots, X \setminus M_k\}$ is a partition of $X$.
For a set $M \subseteq X$, we denote its complement by
$\co{M} = X \setminus M$.

\begin{lemma}\cite[Lemma 3.7]{CaCheNaPre}\label{lemma:partition-copartition}
  Let $(X,d)$ be a dissimilarity space.  Then $\Mmax(X,d)$ is either a
  partition or a copartition of $X$.
\end{lemma}

% \begin{lemma}\cite[Lemma 3.8]{CaCheNaPre}\label{lemma:mmodules-in-copartition}
%   Let $(X, d)$ be a dissimilarity space whose set
%   $\M= \{M_1,\ldots,M_p\}$ of maximal mmodules is a copartition of
%   $X$. Then for any mmodule $M$, either there is
%   $J \subseteq \{1,2,\ldots,p\}$ such that
%   $M = \bigcup_{j \in J} \co{M_j}$, or there is
%   $i \in \{1,2,\ldots,p\}$ such that $M \subseteq \co{M_i}$.
% \end{lemma}

\begin{proposition}\cite[Proposition 3.9]{CaCheNaPre}\label{mmodule-tree}
  Let $(X,d)$ be a dissimilarity space. There is a unique $X$-tree (up to
  children reorderings) with inner nodes labelled by
  $\cup$ or $\cap$, such that:
  \begin{enumerate}[label=(\roman*)]
  \item if a node $\alpha$ is a $\cup$-node, then its arity is at
    least three and for any child $\beta$ of $\alpha$, $X(\beta)$ is
    an mmodule,
  \item if a node $\alpha$ is a $\cap$-node, then its arity is at
    least two, and for any proper subset $\{\beta_1,\ldots,\beta_i\}$
    of children of $\alpha$, $X(\beta_1) \cup \ldots X(\beta_i)$ is an
    mmodule.
  \item any proper mmodule appears exactly once as in (i) and (ii).
  \end{enumerate}
\end{proposition}

If $\Mmax$ is a partition but not a bipartition, then the root is a
$\cup$-node, while if $\Mmax$ is a copartition (possibly a
bipartition), it is a $\cap$-node. For a dissimilarity space $(X,d)$,
we call the $X$-tree defined in \Cref{mmodule-tree} the {\em
  mmodule-tree} of $(X,d)$ and we denote it by $\TM(X,d)$ (or $\TM$
for short). It would be tempting to believe that, for Robinson spaces, $\cup$-nodes
correspond to $Q$-nodes of $\TPQ$ and $\cap$-nodes to $P$-nodes. We
will see later that this is not always the case, and we will describe
precisely the relationship between $\TPQ$ and $\TM$.

\begin{definition}[Copoint]
  A \emph{copoint at a point $p$} (or a \emph{$p$-copoint}) is any
  maximal by inclusion mmodule $C$ not containing $p$; the point $p$
  is the \emph{attaching point } of $C$.
\end{definition}

The copoints of $\M$ minimally generate $\M$, in the sense that each
mmodule $M$ is the intersection of all copoints containing $M$.
Maximal by inclusion mmodules are copoints but the converse is not
true. Denote by $\C_p$ the set of all copoints at $p$ plus the trivial
mmodule $\{ p\}$.

\begin{lemma}\cite[Lemma 3.4]{CaCheNaPre}\label{copoint-partition}
  For any point $p$ of a dissimilarity space $(X,d)$, the copoints of
  $\C_p$ are pairwise disjoint and define a partition of the set $X$.
\end{lemma}

We call $\C_p=\{ C_0=\{ p\}, C_1,\ldots,C_k\}$ a \emph{copoint
  partition} of $(X,d)$ with attaching point $p$. The copoint
partition $\C_p$ is called \emph{trivial} if $\C_p$ consists only of
the points of $X$, i.e., $\C_p=\{ \{ x\}: x\in X\}$, and
\emph{co-trivial} if $\C_p=\{ \{ p\}, X\setminus \{ p\}\}$, i.e., all
points of $X\setminus \{ p\}$ have the same distance to $p$.

We conclude this subsection with the definition of a quotient space of
a dissimilarity space $(X,d)$. In~\cite{CaCheNaPre} we defined and
used it in case of copoint partitions $\C_p$. Now, we will define it
for arbitrary partitions of $X$ into mmodules.

\begin{definition}[Quotient space]\label{def:quotient}
  Let $\M'=\{ M_1,\ldots,M_k\}$ be a partition of a dissimilarity
  space $(X,d)$ into mmodules. The \emph{quotient space}
  $(\M',\widehat{d})$ of $(X,d)$ has the mmodules of $\M'$ as points
  and for $M_i,M_j, i\ne j$ of $\M'$ we set
  $\widehat{d}(M_i,M_j):=d(u,v)$ for an arbitrary pair
  $u\in M_i,v\in M_j$.
\end{definition}
From the definition of mmodules and since $\M'$ is a partition, the
notion of a quotient space is well-defined because $d(u,v)$ is the
same for any choice of the points $u\in M_i$ and $v\in M_j$.

\subsection{Examples}\label{SUBSECTION_example}

In order to illustrate the notions introduced in this section, we give the following example.
In \Cref{TABLE_gros_example} we present a Robinson
  space $(X,d)$ on 12 points. In \Cref{FIGURE_gros_example} we provide
  its PQ-tree and its mmodule-tree. In \Cref{fig:copoints} we list all
  its non-trivial copoints, together with the point to which they are
  attached.

\begin{figure}[htpb]
  {\scriptsize{
      $$
      \begin{array}{c}
        \vspace{0.1cm} D\\
        1 \\ 2 \\ 3 \\ 4 \\ 5 \\ 6 \\ 7 \\ 8 \\ 9 \\ 10 \\ 11 \\  12
      \end{array}
      \,
      \begin{array}{ccccccccccccc}
        \vspace{0.1cm}1 & 2 & 3 & 4 & 5 & 6 & 7 & 8 & 9 & 10 & 11 & 12 \\
        % \\
        0 & 2 & 2 & 3 & 5 & 5 & 5 & 8 & 8 & 8 & 8 & 8 \\
          & 0 & 1 & 2 & 5 & 5 & 5 & 8 & 8 & 8 & 8 & 8 \\
          &   & 0 & 2 & 5 & 5 & 5 & 8 & 8 & 8 & 8 & 8 \\
          &   &   & 0 & 5 & 5 & 5 & 8 & 8 & 8 & 8 & 8 \\
          &   &   &   & 0 & 1 & 1 & 6 & 6 & 6 & 6 & 6 \\
          &   &   &   &   & 0 & 1 & 6 & 6 & 6 & 6 & 6 \\
          &   &   &   &   &   & 0 & 6 & 6 & 6 & 6 & 6 \\
          &   &   &   &   &   &   & 0 & 1 & 2 & 2 & 3 \\
          &   &   &   &   &   &   &   & 0 & 2 & 2 & 2 \\
          &   &   &   &   &   &   &   &   & 0 & 2 & 2 \\
          &   &   &   &   &   &   &   &   &   & 0 & 2 \\
          &   &   &   &   &   &   &   &   &   &   & 0
      \end{array}
      $$
    }}
  \medskip
  \caption{\small{The distance matrix $D$ of a Robinson space $(X,d)$
      with $X = \{1,\ldots, 12\}$.}
    \label{TABLE_gros_example}}
\end{figure}

\begin{figure}[h!]
  \begin{tikzpicture}[scale=.45]
    \begin{scope}
      \draw(0, 2.5)-- (0,0) [below] node{\scriptsize{$1$}} ;
      \draw(1.2, .9)-- (1,0) [below] node{\scriptsize{$2$}} ;
      \draw(1.8, .9)-- (2,0) [below] node{\scriptsize{$3$}} ;
      \draw(1.5, 1.4) circle(.6) ; \draw(1.5, 2.5) -- (1.5, 2) ;
      \draw(1.5, 1.4) node {\scriptsize{$\gamma_1$}};

      \draw(3, 2.5)-- (3,0) [below] node{\scriptsize{$4$}} ;
      \draw(-.3, 2.5) -- (3.3, 2.5) -- (3.3, 3.3) -- (-.3, 3.3) -- (-.3, 2.5) ;
      \draw(1.5, 2.9) node {\scriptsize{$\beta_1$}} ;
      \draw(1.5, 4.1) -- (1.5, 3.3) ;
      \draw(4.55, .95) -- (4, 0) [below] node{\scriptsize{$5$}} ;
      \draw(5, .8) -- (5, 0) [below] node{\scriptsize{$6$}} ;
      \draw(5.45, .95) -- (6, 0) [below] node{\scriptsize{$7$}} ;
      \draw(5, 1.4) circle(.6) ; \draw(5, 1.4) node {\scriptsize{$\beta_2$}};
      \draw(5, 4.1)-- (5, 2) ;
      \draw(7, 2.5)-- (7,0) [below] node{\scriptsize{$8$}} ;
      \draw(8, 2.5)-- (8,0) [below] node{\scriptsize{$9$}} ;
      \draw(9.2, .9) -- (9,0) [below] node{\scriptsize{$10$}} ;
      \draw(9.8, .9) -- (10,0) [below] node{\scriptsize{$11$}} ;
      \draw(9.5, 1.4) circle(.6) ; \draw(9.5, 2.5) -- (9.5, 2) ;
      \draw(9.5, 1.4) node {\scriptsize{$\gamma_2$}};
      \draw(11, 2.5) -- (11, 0) [below] node{\scriptsize{$12$}} ;
      \draw(6.7, 2.5) -- (11.3, 2.5) -- (11.3, 3.3) -- (6.7, 3.3) -- (6.7, 2.5) ;
      \draw(9, 2.9) node {\scriptsize{$\beta_3$}} ;
      \draw(9, 4.1) -- (9, 3.3) ;
      \draw(-.3, 4.1) -- (11.3,4.1) -- (11.3, 4.9) -- (-.3, 4.9) -- (-.3, 4.1) ;
      \draw(6, 4.5) node{\scriptsize{$\alpha$}} ;
      \draw(5.5, -1.5) node{(A)} ;
    \end{scope}

    \begin{scope}[xshift= 15cm]
      \draw(-.7, .85) -- (-1, 0) [below] node{\scriptsize{$1$}} ;
      \draw(-.3, .85) -- (0, 0) [below] node{\scriptsize{$4$}} ;
      \draw(-.5, 1.4) circle(.6) ; \draw(-1.5, 1.4) node{\small{$\eta_1$}} ;
      \draw(-.5, 1.35) node{\scriptsize{$\cap$}} ;
      \draw(1.3, .85) -- (1, 0) [below] node {\scriptsize{$2$}} ;
      \draw(1.7, .85) -- (2, 0) [below] node{\scriptsize{$3$}} ;
      \draw(1.5, 1.4) circle(.6) ; \draw(1.5, 1.35) node{\scriptsize{$\cap$}} ;
      \draw(2.6, 1.4) node{\small{$\eta_2$}} ;
      \draw(0.1, 2.6) -- (-.5, 2) ;
      \draw(.9, 2.6) -- (1.5, 2) ;
      \draw(.5, 3) circle(.6) ; \draw(.5, 2.95) node{\scriptsize{$\cap$}} ;
      \draw(-.5, 3) node{\small{$\xi_1$}} ;
      \draw(3.55, .95) -- (3, 0) [below] node{\scriptsize{$5$}} ;
      \draw(4, .8) -- (4, 0) [below] node{\scriptsize{$6$}} ;
      \draw(4.45, .95) -- (5, 0) [below] node{\scriptsize{$7$}} ;
      \draw(4, 1.4) circle(.6) ;
      \draw(4, 1.4) node {\scriptsize{$\cap$}};
      \draw(5.1, 1.4) node{\small{$\xi_2$}} ;
      \draw(6.7, 2.5) -- (6, 0) [below] node{\scriptsize{$10$}} ;
      \draw(7, 2.4) -- (7,0) [below] node{\scriptsize{$11$}} ;
      \draw(8.55, .95) -- (8, 0) [below] node{\scriptsize{$8$}} ;
      \draw(9, .8) -- (9, 0) [below] node{\scriptsize{$9$}} ;
      \draw(9.45, .95) -- (10, 0) [below] node{\scriptsize{$12$}} ;
      \draw(9, 1.4) circle(.6) ;
      \draw(9, 1.4) node {\scriptsize{$\cup$}};
      \draw(10.1, 1.4) node{\small{$\eta_3$}} ;
      \draw(7.3, 2.5) -- (8.5, 1.7) ;
      \draw(7, 3) circle(.6) ; \draw(7, 2.95) node{\scriptsize{$\cap$}} ;
      \draw(8.1, 3) node{\small{$\xi_3$}} ;
      \draw(4,3.9) -- (4, 2) ; \draw(3.6, 4)--(1, 3.3) ; \draw(4.4, 4)--(6.5, 3.3) ;
      \draw(4, 4.5 ) circle(.6) ; \draw(4, 4.45) node {\scriptsize{$\cup$}};
      \draw(5, 4.5) node{\small{$\zeta$}} ;
      \draw(4.5, -1.5) node{(B)} ;
    \end{scope}
  \end{tikzpicture}
  \caption{\small{ (A) The PQ-tree $\TPQ$ and (B) the mmodule-tree
      $\TM$ of the dissimilarity space from
      \Cref{TABLE_gros_example}.}\label{FIGURE_gros_example}}
\end{figure}

\begin{figure}[htpb]
  \begin{center}
        {\small{
    \begin{tabular}{cc}
      non-trivial mmodule & copoint for \\ \hline
      \{2, 3\} & 1, 4 \\
      \{1, 4\} & 2, 3 \\
      \{1, 2, 3, 4\} & 5 to 12 \\
      \{5, 6\} & 7 \\
      \{5, 7\} & 6 \\
      \{6, 7\} & 5
    \end{tabular}
    }}
    \hspace{1cm}
    {\small{
    \begin{tabular}{cc}
      non-trivial mmodule & copoint for \\ \hline

      \{5, 6, 7\} & 1 to 4 and 8 to 12 \\
      \{10, 11\} & 8, 9, 12 \\
      \{8, 9, 10, 12\} & 11 \\
      \{8, 9, 11, 12\} & 10 \\
      \{8, 9, 10, 11, 12\} & 1 to 7 \\
      &
    \end{tabular}
    }}
  \end{center}
  \caption{\small{The non-trivial mmodules of the dissimilarity space from
    \Cref{TABLE_gros_example} that are copoints, together with the points to which
    they are attached.}}
  \label{fig:copoints}
\end{figure}

%% file: sec_3_ro_nodes.tex
In this section, first we characterize the set of orders that are
representable by a PQ-tree and then we characterize the subsets of a
Robinson space $(X,d)$ that are the nodes of the PQ-tree of
$(X,d)$. These auxiliary results will be used in the next two
sections.

\subsection{Represented orders}
In this subsection, we characterize the set $\Pi$ of orders of $X$
that are represented by PQ-trees.

\begin{definition}[$\Pi$-block]
  For a set $\Pi$ of orders on $X$, a set $Y \subsetneq X$ is called a
  \emph{$\Pi$-block} if $Y$ is an interval in each order of $\Pi$.
  Then the blocks from Definition \ref{def:block} are precisely the
  $\Pi(d)$-blocks.
\end{definition}

\begin{lemma}\label{lemma:block}
  Let $\TPQ$ be a PQ-tree over a set $X$ and $\Pi(\TPQ)$ be the set of
  represented permutations of $\TPQ$. Then a set $S \subseteq X$ is a
  $\Pi(\TPQ)$-block if and only if
  \begin{enumerate}[label=(\roman*)]
  \item\label{item:block1} either there is a node $\alpha$ in $\TPQ$
    with $X(\alpha) = S$,
  \item\label{item:block2} or there is a Q-node $\alpha$ in $\TPQ$
    with children $\beta_1,\ldots, \beta_k$, and $j, j'$ with
    $1 \leq j < j' \leq k$ such that $\{j,j'\} \neq \{1,k\}$ and
    $\bigcup_{i \in \{j,\ldots,j'\}} X(\beta_i) = S$.
  \end{enumerate}
\end{lemma}

\begin{proof}
  Let $S \subseteq X$ be a $\Pi(\TPQ)$-block, and $\alpha$ be the
  pertinent node of $S$. If $S = X(\alpha)$, case~\ref{item:block1}
  holds. Otherwise, $S \subsetneq X(\alpha)$.

  For the sake of contradiction, suppose that $\alpha$ is a P-node.
  If there is a child $\beta$ of $\alpha$ and $x \in S \cap X(\beta)$,
  $y \in X(\beta) \setminus S$, then let
  $z \in S \cap (X(\alpha) \setminus X(\beta)$ be in another child. As
  $S$ is a block, there is a compatible order $\<$ with $x \< z \< y$
  or $y \< x \< z$ (say the former). By definition of the compatible
  orders of a PQ-tree, reversing the order on $X(\beta)$ provides
  another compatible order, with $z \< x \< y$, in contradiction with
  the fact that $S$ is a block. Otherwise, $S$ is the union of the
  sets induced by some children of $\alpha$. Let
  $x, y, z \in X(\alpha)$ be in three distinct children of $\alpha$,
  with $x,z \in S$, $y \notin S$. By property of P-nodes, there is a
  compatible order with $x \< y \< z$, contradicting the fact that $S$
  is a block.

  Hence $\alpha$ is a $Q$-node. Let $\beta_1,\ldots,\beta_k$ be its
  children. Then
  $\{i \in \{1,\ldots,k\} : \beta_i \cap S \neq \varnothing\}$ is an
  interval $\{j,\ldots,j'\}$ because $S$ is a $\Pi(\TPQ)$-block and by
  property of Q-nodes, with $j < j'$ as $\alpha$ is
  $S$-pertinent. Finally, if there is
  $z \in \bigcup_{k \in \{j,\ldots,j'\}} \beta_k \setminus S$, let
  $x \in \beta_j \cap S$ and $y \in \beta_{j'} \cap S$, there is a
  compatible order with $x \< z \< y$, contradicting of the fact
  that $S$ is a $\Pi(\TPQ)$-block. Hence~\ref{item:block2} holds. The
  reverse direction is immediate.
\end{proof}

Consider a set $\Pi$ of orders on $X$, and say that two points
$x,y \in X$ are \emph{equivalent} if every maximal $\Pi$-block
contains either none or both of them. This is clearly an equivalence
relation denoted $\approx_\Pi$. We can use it to characterize when a
set of orders is represented by a PQ-tree. To this end, let
$\mathcal{B}_\Pi$ denote the set of equivalence classes of $\approx_\Pi$.

\begin{proposition}\label{prop:charac-PQtree}
  A set of orders $\Pi$ on $X$ is represented by a PQ-tree if and
  only if the restriction of $\Pi$ on each $S \in \mathcal{B}_\Pi$
   is represented by a PQ-tree $\mathcal{T}_S$, and
  \begin{enumerate}[label=(\roman*)]
  \item\label{item:pcase} either $\Pi$ is the set of orders obtained
    by composing any order on $\mathcal{B}_\Pi$ with orders from each
    $\mathcal{T}_S, S \in \mathcal{B}_\Pi$,
  \item\label{item:qcase} or there are at least three equivalence
    classes and there is an order $\<_{\mathcal{B}}$ on
    $\mathcal{B}_\Pi$ such that $\Pi$ is the set of orders obtained by
    composing $\<_{\mathcal{B}}$ or its reverse with orders from each
    $\mathcal{T}_S, S \in \mathcal{B}_\Pi$.
  \end{enumerate}
  In that case, the root of the PQ-tree is a P-node in
  case~\ref{item:pcase} and a Q-node in case~\ref{item:qcase}, and its
  children are the PQ-trees $\mathcal{T}_S$ of each equivalence class
  $S$ of $\approx_\Pi$, in arbitrary order in case~\ref{item:pcase}
  and in order $\<_{\mathcal{B}}$ (or its reverse) in case~\ref{item:qcase}.
\end{proposition}

\begin{proof}
  Suppose first that $\Pi$ is represented by a PQ-tree
  $\mathcal{T}$. Let $\beta_1,\ldots,\beta_k$ be the children of the
  root. From \Cref{lemma:block}, when the root is a P-node, the
  maximal blocks are the sets $X(\beta_i)$ and thus are also the
  equivalence classes of $\approx_\Pi$, whereas when the root is a
  Q-node, the maximal blocks are induced by intervals
  $\bigcup \{X(\beta_i) : i \in \{j,\ldots,j'\}\}$, implying that the
  equivalence classes of $\approx_\Pi$ are also the set $X(\beta_i)$.
  The result then follows, with~\ref{item:pcase} corresponding to a
  P-node root, and~\ref{item:qcase} corresponding to a Q-node root
  ($\<_{\mathcal{B}}$ being the order on the children of the root).

  Conversely, suppose that each equivalence class of $\approx_\Pi$ is
  represented, and either~\ref{item:pcase} or~\ref{item:qcase}
  holds. Then it can be readily checked that the PQ-tree defined in
  the statement of the result does indeed represent $\Pi$.
\end{proof}

%
% The goal of this section is to show how the mmodule-tree $\TM$ of a
% Robinson space can be derived from its PQ-tree $\TPQ$. In this
% section, $(X, d)$ is a Robinson space and $\TPQ$ is its PQ-tree. Given
% an mmodule $M$ of $(X, d)$, we say that a node $\alpha$ of the PQ-tree
% $\TPQ$ is {\em $M$-full} if $X(\alpha) \cap M = X(\alpha)$, {\em
%   $M$-empty} if $X(\alpha) \cap M = \varnothing$, and {\em
%   $M$-partial} if $X(\alpha) \cap M$ is a nonempty proper subset of
% $M$.
%
%
% A node $\alpha$ is the {\em pertinent} node of an mmodule $M$, or
% {\em $M$-pertinent}, if $\alpha$ is the lowest node of $\TPQ$ such
% that $M \subseteq X(\alpha)$.
%
\subsection{Nodes of a PQ-tree}
Sets that are simultaneously mmodules and blocks of a Robinson space
$(X,d)$ play a special role because in any compatible order their
points can be ordered independently of the rest of the points. We
prove that they correspond exactly with the nodes of the PQ-tree $\TPQ$ of
$(X,d)$.

%Given a set $S \subset X$, a PQ-tree $\mathcal{T}$ on $X$ and a node
%$\alpha$ of $\mathcal{T}$, we say that $\alpha$ is
%
%$S${\em -empty} if $S \cap X(\alpha) = \varnothing$,
%
%$S${\em -full} if  $X(\alpha) \subseteq S$,
%
%$S$-{\em partial} if $S \cap X(\alpha)  \neq \varnothing, X(\alpha)$,
%
%$S$-{\em pertinent} (or the {\em pertinent node} of $S$) if $\alpha$
%is the lowest node such that $S \subseteq X(\alpha)$.

\begin{lemma}\label{lemma:mmod-block-orders}
  Let $(X,d)$ be a Robinson space with a compatible order $\<'$. Let
  $S$ be an mmodule that is an interval in $\<'$ and $\<_S$ be any
  compatible order on $S$.  Let $\<$ be the total order on $X$ defined
  by setting $x \< y$ when either $\{x, y\} \subset S$ and
  $x \<_ S y$, or $\{x, y\} \nsubset S$ and $x \<' y$ (equivalently,
  $\<$ is obtained from $\<'$ by reordering the elements of $S$
  according to $\<_S$). Then $\<$ is a compatible order of $(X,d)$.
\end{lemma}

\begin{proof}
  Pick any $x \< y \< z$ and we check that
  $d(x,z) \geq \max \{d(x,y), d(y,z)\}$. If $x,y,z \in S$, then
  $x \<_S y \<_S z$ and the result follows. If
  $|\{x,y,z\} \setminus S| \in \{2,3\}$, then $x \<' y \<' z$ and we
  are done again. If $x \notin S$, $y,z \in S$, then $d(x,y) = d(x,z)$
  because $S$ is an mmodule, and
  $d(y,z) \leq \max \{ d(x,y), d(x,z)\}$ because either
  $x \<' y \<' z$ or $x \<' z \<' y'$ holds. Thus
  $d(x,z) \geq \max \{d(x,y), d(y,z)\}$. The case $x,y \in S$,
  $z \notin S$ is symmetric, while the case $x,z \in S$, $y \notin S$
  cannot happen as $S$ is an interval in $\<'$.
\end{proof}

Let $\<$ be a compatible order for $(X,d)$ for which $S$ is an
interval. Let $\revS$ be the order defined by setting $x \revS y$ if
either $x \< y$ and $\{x, y\} \nsubseteq S$, or $y \< x$ and
$\{x,y\} \subseteq S$, that is $\revS$ is the order obtained from $\<$
by reversing the interval $S$.  We can easily generalize $\revS$ to a
subset $S$ which is not an interval and we have the following
elementary result:

% \begin{lemma} \label{lemma:mmodule_reversal} Let $(X,d)$ be a Robinson
%   space with a compatible order $\<$ and $S$ be an mmodule of
%   $X$. Then $\revS$ is a compatible order of $(X,d)$.
% \end{lemma}

\begin{lemma}\label{lemma:mmodule_interval}
  Let $(X,d)$ be a Robinson space and $S \subseteq X$ an interval of a
  compatible order $\<$. Then $\revS$ is a compatible order if and
  only if $S$ is an mmodule.
\end{lemma}

\begin{proof}
  If $S$ is an mmodule, then by \Cref{lemma:mmod-block-orders} $\revS$
  is compatible. Conversely, let $\revS$ be a compatible order. Let
  $y,z \in S$ with $y \< z$, and let $x \notin S$, and say (by
  symmetry) that $x \< y \< z$. Then $d(x,y) \leq d(x,z)$. Since
  $x \revS z \revS y$ and $\revS$ is compatible, we obtain
  $d(x,z) \leq d(x,y)$. Consequently, $d(x,y) = d(x,z)$, yielding that
  $S$ is an mmodule.
\end{proof}

More generally, one can reorder the elements of a block as long as the
order of the block itself remains compatible.

\begin{lemma}\label{THEOREM_R_module_simple}
  For any node $\alpha$ of the PQ-tree $\TPQ$ of a Robinson space
  $(X, d)$, the set $X(\alpha)$ is an mmodule of $(X,d)$.
\end{lemma}

\begin{proof}
  By definition of PQ-trees, for any compatible order $\<$ of $\TPQ$,
  $S := X(\alpha)$ is an interval and $\revS$ is compatible. By
  \Cref{lemma:mmodule_interval}, $X(\alpha)$ is an mmodule.
\end{proof}

Using \Cref{mmodules}\ref{item:mmod6}, this justifies the notation
$d(\alpha,\beta)$ for any two nodes $\alpha$ or $\beta$ of $\TPQ$ or
$\TM$ with $X(\alpha)$ and $X(\beta)$ disjoint, where
$d(\alpha,\beta) = d(x,y)$ for any $x \in X(\alpha)$,
$y \in X(\beta)$. Combining the last results, we get the following
characterization of the subsets of $X$ corresponding to the nodes of
the PQ-tree of $(X,d)$:

\begin{theorem}\label{LEMMA_block_node}
  Let $(X, d)$ be a Robinson space and $M$ be a subset of $X$. Then
  $X(\alpha)=M$ for some node $\alpha$ of $\TPQ$ if and
  only if $M$ is a block and an mmodule of $(X,d)$.
\end{theorem}

\begin{proof}
  From \Cref{THEOREM_R_module_simple,lemma:block}, for any node
  $\alpha$ of the PQ-tree, $X(\alpha)$ is an mmodule and a
  block. Conversely, consider a set $M \subseteq X$ that is both an
  mmodule and a block. By way of contradiction, suppose that there is
  no node $\alpha$ such that $M = X(\alpha)$. By \Cref{lemma:block},
  there is a Q-node $\alpha$ in $\TPQ$ with children
  $\beta_1,\ldots,\beta_k$, and $j, j'$ with $1 \leq j < j' \leq k$
  and $\{j,j'\} \neq \{1,k\}$ such that
  $M = X(\beta_j) \cup \ldots \cup X(\beta_{j'})$. Let $\<$ be a
  compatible order with $X(\beta_1) \< \ldots \< X(\beta_k)$. Then by
  \Cref{lemma:mmodule_interval}, $\revS$ is also a compatible order,
  with
  $X(\beta_1) \revS X(\beta_j') \revS X(\beta_j) \revS X(\beta_k)$. But
  that order is not compatible with $\TPQ$, a contradiction.
\end{proof}

This leads to the following result, which involves flat Robinson spaces.

%This leads to a characterization of flat Robinson spaces: Robinson
%spaces that have only two compatible orders reverse to each other.

\begin{corollary}\label{thm:flat}
  Let $(X,d)$ be a Robinson space in which all mmodules are trivial
  and $|X| \geq 3$. Then the following holds:
  \begin{enumerate}[label=(\roman*)]
  \item\label{item:flat1} $(X,d)$ is flat,
  \item\label{item:flat2} $\TPQ$ has a single non-leaf node, that is a Q-node.
  \end{enumerate}
\end{corollary}

\begin{proof}
   Since
  each mmodule is trivial, by \Cref{THEOREM_R_module_simple}, $\TPQ$
  has a single non-leaf node $\alpha$.
  Suppose for sake of contradiction that $\alpha$ is a P-node. Let
  $x,y,z \in X$ be distinct points, that are children of
  $\alpha$. Then there exist compatible orders $\<_1$, $\<_2$, and
  $\<_3$ with $x \<_1 y \<_1 z$, $y \<_2 z \<_2 x$, and
  $z \<_3 x \<_3 y$. Hence
  $d(x,z) \geq d(x,y) \geq d(y,z) \geq d(x,z)$, meaning that these
  three distances are equals. Consequently, all pairwise distances
  between the points of $X$ are equal, thus any subset of points is an
  mmodule, a contraction. Thus $\alpha$ is a Q-node and~\ref{item:flat2} 
  is verified.
  
  Clearly, \ref{item:flat2}  implies~\ref{item:flat1}.
\end{proof}

\begin{remark}
It is worth observing that the converse of \Cref{thm:flat} does not
hold, and dealing with this limitation is arguably one of the main
technical difficulties that this paper addresses. This is illustrated
by the simple example given in \Cref{fig:flat-non-trivial}. Later we will
show that this happens when some maximal mmodule has larger
diameter than its distance to other points; here $\{a,c\}$ has
diameter 2, which is higher than the distance between the
mmodules $\{a,c\}$ and $\{b\}$.

\begin{figure}[htbp]
  \begin{center}
    \begin{tabular}{cp{1cm}cp{1cm}c}
      \begin{minipage}{0.2\textwidth}
        \begin{flushright}
        {\small{
          \begin{tabular}{cccc}
            D   & $a$ & $b$ & $c$ \\
            $a$ &  0  &  1  &  2  \\
            $b$ &     &  0  &  1  \\
            $c$ &     &     &  0  \\
          \end{tabular}
          }}
        \end{flushright}
      \end{minipage}
      &&
         \begin{minipage}{0.2\textwidth}
           \begin{tikzpicture}[x=.8cm,y=.9cm]
             \node[draw,rectangle,minimum width=2.2cm,minimum height=0.4cm] (q) at (0,-1) {};
             \node (a) at (-1,-2) {$a$};
             \node (b) at (0,-2) {$b$};
             \node (c) at (1,-2) {$c$};
             \draw (a) |- (a |- q.south);
             \draw (b) |- (b |- q.south);
             \draw (c) |- (c |- q.south);
           \end{tikzpicture}
         \end{minipage}
      &&
         \begin{minipage}{0.2\textwidth}
           \begin{tikzpicture}[x=.8cm,y=.8cm]
             \node[draw,circle] (r) at (0,0) {$\cap$};
             \node[draw,circle] (t) at (-1,-1) {$\cap$};
             \node (a) at (-2,-2) {$a$};
             \node (c) at (0,-2) {$c$};
             \node (b) at (1,-1) {$b$};
             \foreach \i/\j in {r/t,r/b,t/a,t/c} {
               \draw (\i) -- (\j);
             }
           \end{tikzpicture}
        \end{minipage}
    \end{tabular}
  \end{center}
  \caption{A flat Robinson space, together with its PQ-tree $\TPQ$ and
    its mmodule tree $\TM$.}
  \label{fig:flat-non-trivial}
\end{figure}
\end{remark}

\subsection{Distances and the PQ-tree}

In this subsection, we present some results about the values of the
dissimilarity relative to the nodes of the PQ-tree (and the mmodule
tree). The next two lemmas motivate the notions and the results of the
next section, and also relates to the properties of clusters and of
the weights of nodes in dendrograms of ultrametrics.

\begin{lemma}\label{lemma:coP}
  Let $(X,d)$ be a Robinson space with PQ-tree $\TPQ$ and $\alpha$ a
  P-node in $\TPQ$. Let $\beta_1,\dots,\beta_k$ be some of the
  children of $\alpha$, then $S := X(\beta_1) \cup \ldots \cup X(\beta_k)$
  is an mmodule.
\end{lemma}

\begin{proof}
  If $k=1$ or $S = X(\alpha)$, this follows from
  \Cref{LEMMA_block_node}. Otherwise, let $x,y \in S$ be in two
  distinct children of $\alpha$, and $z \in X(\alpha) \setminus S$. By
  property of the P-nodes, there are compatible orders $\<$ and $\<'$
  with $x \< y \< z$ and $y \<' x \<' z$. Thus
  $d(y,z) \leq d(x,z) \leq d(y,z)$, implying that $d(x,z) = d(y,z)$; so
  $S$ is an mmodule.
\end{proof}

\begin{lemma}\label{lemma:uniform-pnode}
  Let $(X,d)$ be a Robinson space with PQ-tree $\TPQ$ and mmodule tree
  $\TM$. Let $\alpha$ be a P-node of $\TPQ$ or a $\cap$-node of
  $\TM$. Then there exists $\delta > 0$ such that for any $x,y$
  appearing in two distinct children of $\alpha$, we have $d(x,y)=\delta$.
\end{lemma}

\begin{proof}
  Let $\beta$ be a child of $\alpha$. Then
  $X(\alpha) \setminus X(\beta)$ is an mmodule, by \Cref{mmodule-tree}
  when $\alpha$ is a $\cap$-node and by \Cref{lemma:coP} when $\alpha$
  is a P-node.  Hence by \Cref{mmodules}\ref{item:mmod6} the distances
  are uniform between $X(\beta)$ and $X(\alpha) \setminus X(\beta)$,
  and thus uniform between children of $\alpha$.
\end{proof}

As a consequence, for any $\alpha$ P-node of $\TPQ$ or $\cap$-node of
$\TM$, we denote $\rhos(\alpha)$ the value $\delta$ from
\Cref{lemma:uniform-pnode}, that is the distance between children
of $\alpha$. Another observation is that, similarly to dendrograms,
the diameters of nodes are almost strictly monotone in the PQ-tree.

\begin{lemma}\label{lemma:monotone-diameter}
  Let $(X, d)$ be a Robinson space with PQ-tree $\TPQ$. Let $\alpha$
  be an internal node of $\TPQ$ and $\beta$ a non-leaf child of $\alpha$. Then
  \begin{enumerate}[label=(\roman*)]
  \item\label{item:mono1} if $\alpha$ is a P-node,
    $\diam(X(\beta)) \leq \rhos(\alpha) = \diam(X(\alpha))$;
  \item\label{item:mono2} if $\alpha$ is a Q-node,
    $\diam(X(\beta)) \leq \min \{d(\beta,\beta') : \beta' \textrm{
      child of } \alpha, \beta \neq \beta'\} \leq \diam(X(\alpha))$;
  \item\label{item:mono3} if $\diam(X(\alpha)) = \diam(X(\beta))$,
    then $\alpha$ is a P-node, $\beta$ is a Q-node, and for any
    non-leaf child $\gamma$ of $\beta$,
    $\diam(X(\gamma)) < \diam(X(\alpha))$.
  \end{enumerate}
\end{lemma}

\begin{proof}
  Let $x,y \in X(\beta)$ and $z \in X(\alpha) \setminus X(\beta)$.  By
  \Cref{LEMMA_block_node}, $X(\beta)$ is a block and mmodule, hence
  for any compatible order $<$ with $x < y$, either $z < x$ or
  $y < z$. By symmetry we may assume $x < y < z$. Then
  $d(x,y) \leq d(x,z) = d(y,z)$. If $\alpha$ is a P-node, then by
  \Cref{lemma:uniform-pnode} $d(x,z) = \rhos(\alpha)$ and
  thus~\ref{item:mono1} holds. If $\alpha$ is a Q-node, by choosing
  $z \in \beta'$, we get~\ref{item:mono2}.

  Suppose that $\delta := \diam(X(\alpha)) = \diam(X(\beta))$. We
  claim that $\alpha$ is a P-node. Assume that $\alpha$ has arity at
  least three (as otherwise it is a P-node by definition). Then for
  any $\beta'$ child of $\alpha$ distinct from $\beta$,
  $d(\beta,\beta') = \delta$. Hence
  $S' := X(\alpha) \setminus X(\beta)$ is an mmodule.  As $\alpha$ has
  at least three children, $S'$ is not a set induced by a node of
  $\TPQ$, thus by \Cref{LEMMA_block_node} is not a block. Because
  $X(\alpha)$ is a block, it means that there is a compatible order
  $<$, $x,z \in S'$ and $y \in X(\beta)$ with $x < y < z$. Let
  $\beta_x$ and $\beta_z$ be the children of $\alpha$ containing $x$
  and $z$ respectively. Then, for each $x' \in S'$ with $x' < y$ and
  each $z' \in S'$ with $y < z'$,
  $\delta = d(x',y) \leq d(x',z') \leq \diam(X(\alpha)) = \delta$,
  proving that $S := X(\beta) \cup \{x \in S : x < y\}$ is an
  mmodule. Then, by \Cref{lemma:mmodule_interval}, $\rev{S}$ is a
  compatible order. Hence $\beta_x, \beta, \beta_z$ may be ordered in
  more than two ways, $\alpha$ is a P-node.

  For the sake of contradiction, suppose furthermore that $\beta$ is a
  P-node. Then $\rhos(\alpha) = \rhos(\beta) = \delta$. Let $\gamma$
  be a child of $\beta$. Then by \Cref{mmodule-tree},
  $X(\beta) \setminus X(\gamma)$ is an mmodule and
  $d(X(\gamma),X(\beta \setminus X(\gamma)) = \delta$. Thus
  $S := X(\alpha) \setminus X(\gamma)$ is also an mmodule. Let $<$ be a
  compatible order with $x < y < z$ for any $x \in X(\gamma)$,
  $y \in X(\beta) \setminus X(\gamma)$ and
  $z \in X(\alpha) \setminus X(\beta)$ (such an order exists by the
  structure of $\TPQ$). Then by \Cref{lemma:mmodule_interval}, $\revS$
  is also a compatible order, contradiction because $X(\beta)$ is a
  block by \Cref{LEMMA_block_node}. Thus $\beta$ is a
  Q-node. Consequently, for any non-leaf child $\gamma$ of $\beta$,
  $\diam(X(\gamma)) < \diam(X(\beta)) = \diam(X(\alpha))$. Thus~\ref{item:mono3}
  is proved.
\end{proof}

\subsection{Dendrograms and PQ-trees for ultrametrics}
As an application of the results of this section, we prove that for ultrametrics,
the X-trees $\TPQ$ and $\TD$ are isomorphic. Furthermore, we consider the
following characterizations of ultrametric spaces by their PQ-trees
(these results probably can be considered as half-folkloric, however
we give their full proofs):

\begin{proposition}\label{prop:ultrametric_pq_tree}
  A Robinson space $(X, d)$ is ultrametric if and only all internal
  nodes of $\TPQ$ are P-nodes.
\end{proposition}

\begin{proof}
  Suppose first that $(X,d)$ is ultrametric and, by a way of
  contradiction, that its PQ-tree $\TPQ$ has a Q-node
  $\alpha = Q(\beta_1, \beta_2, \ldots, \beta_k)$ with $k \geq 3$. We
  set $\delta = \diam(X(\alpha)) = d(\beta_1, \beta_k)$. Let
  $i \in \{1,\ldots, k\}$ be minimum such that
  $d(\beta_1, \beta_i) = \delta$. If $i = 2$, then, for all
  $j \in \{2,\ldots, k\}$, $d(\beta_1, \beta_j) = \delta$, implying
  that $X(\beta_2) \cup \ldots \cup X(\beta_k)$ is an mmodule. By
  \Cref{lemma:block}, it is also a block, and thus by
  \Cref{LEMMA_block_node} $k = 2$, contradiction.  If $i > 2$, since
  $d(\beta_1,\beta_{i-1}) < \delta$ and
  $d(\beta_{i-1},\beta_i) \leq d(\beta_1,\beta_i) = \delta$, by
  definition of ultrametric we have
  $d(\beta_{i-1}, \beta_i) = \delta$. Consequently, for all
  $1 \leq j < i \leq j' \leq k$, we have
  $d(\beta_j, \beta_{j'}) = \delta$. Thus
  $X(\beta_1) \cup \ldots \cup X(\beta_{i-1})$ and
  $X(\beta_i) \cup \ldots \cup X(\beta_k)$ are blocks and mmodules, by
  \Cref{LEMMA_block_node} $i = k = 2$, contradiction.  Conversely,
  suppose that all internal nodes of $\TPQ$ of arity more than two are
  P-nodes and let $x,y,z \in X$.  Let $\alpha$ be the lowest common
  ancestor of $x, y, z$. If two of them, say $x$ and $y$, have a child
  of $\alpha$ as common ancestor, then we have
  $d(x,y) \leq d(x,z) = d(y, z)$. Otherwise,
  $\lca(x,y) = \lca(x,z) = \lca(y,z) = \alpha$. Since $\alpha$ is a
  P-node of arity  at least three, we have
  $d(x,y) = d(y, z) = d(x, z)$.
\end{proof}

\begin{proposition}\label{prop:dendrogram_pq_tree}
  Let $(X, d)$ be an ultrametric space. Then the X-trees $\TPQ$ and
  $\TD$ are isomorphic.
\end{proposition}

\begin{proof}
  Let $\beta$ be an internal node of $\TD$. Then for any
  $x \notin X(\beta)$, let $\alpha$ be the least common ancestor of
  $x$ and $\beta$ in $\TD$, we have $d(x,\beta) = \rhos(\alpha)$. Thus
  $X(\beta)$ is an mmodule. Also, by definition of $\TD$, $X(\beta)$
  is a ball of radius $\rhos(\beta)$ centered at any point in
  $X(\beta)$. From the definition of Robinson space, any ball is a
  block. By \Cref{LEMMA_block_node}, as $X(\beta)$ is a block and an
  mmodule, there is a node $\beta'$ in $\TPQ$ with $\beta = \beta'$.

  Conversely, for any node $\beta'$ in $\TPQ$ and any
  $x \notin X(\beta')$, let $\alpha'$ be the least common ancestor of
  $\beta'$ and $x$. By \Cref{prop:ultrametric_pq_tree}, $\alpha'$ and
  $\beta'$ are P-nodes, and thus $d(x,\beta') = \rhos(\alpha')$. By
  \Cref{lemma:monotone-diameter}, $\rhos(\alpha') > \rhos(\beta')$,
  hence $d(x,X(\beta')) > \diam(X(\beta'))$. Thus there is a node
  $\beta$ in $\TD$ with $X(\beta) = X(\beta')$.
\end{proof}

%% file: sec_4_Gdelta_pq.tex
%\section{The graph $\Gdelta$ and the construction of the PQ-tree}

In this section we introduce the graph $\Gdelta$ and the
$\delta$-mmodules, which are the most important notions defined in
this paper. We use them together with the maximal mmodules to
construct the PQ-tree of a Robinson space.

\subsection{The graph \texorpdfstring{$\Gdelta$}{Gd} and \texorpdfstring{$\delta$}{d}-mmodules}

In this subsection, we define and give properties of  the graph $\Gdelta$ and $\delta$-modules. Some of the properties are valid for any dissimilarity space (first part of \Cref{lemma:delta-mmodules}, \Cref{lemma:copartition-delta,lemma:find-rho}), the others are only valid for Robinson spaces.

\begin{definition}[Graphs $\Gdelta$, $G_{<\delta}$,$G_{\le \delta}$,  and $\delta$-mmodules]
  Let $(X,d)$ be a dissimilarity space and let $\delta > 0$. Then
  $\Gdelta$ is the graph with $X$ as the set of vertices and edges
  $\{xy : x,y \in X, d(x,y) \neq \delta\}$. Let also
  $G_{<\delta}:=(X,\{xy : x \neq y, d(x,y) < \delta\})$ and
  $G_{\le \delta}:=(X,\{xy : x \neq y, d(x,y)\le \delta\})$. For
  $S\subseteq X$, we denote by $\Gdelta(S)$ the subgraph of $\Gdelta$
  induced by $S$. The connected components of the graph $\Gdelta$ are
  called the \emph{$\delta$-mmodules} of $(X,d)$, and their number is
  denoted $c_\delta$.
\end{definition}

% We will show below, as a consequence of
% \Cref{lemma:copartition-delta}, that there is at most one value of
% $\delta$ such that $\Gdelta$ is not connected.

% \begin{lemma}\label{lemma:block-diam}
%   Let $(X,d)$ be a Robinson space and $S \subseteq X$ be a block and
%   mmodule. Then for each $x \in X \setminus S$, $d(x,S) \geq \diam S$.
% \end{lemma}

% \begin{proof}
%   Let $<$ be a compatible order. Let $y,z$ be a diametral pair of $S$,
%   we may assume $x < y < z$. Then
%   $d(x,S) = d(x,z) \geq d(y,z) = \diam S$.
% \end{proof}

\begin{lemma}\label{lemma:delta-mmodules}
  Let $(X,d)$ be a dissimilarity space and $\delta > 0$ such that the graph
  $\Gdelta$ is not connected. Then each $\delta$-mmodule is an
  mmodule.
  Moreover, if $(X, d)$ is Robinson, then at most one $\delta$-mmodule is not a block, and each
  $\delta$-mmodule that is a block has diameter at most $\delta$.
\end{lemma}

\begin{proof}
  Let $M$ be a $\delta$-mmodule. Then for any $x,y \in M$,
  $z \notin M$, by definition of $\Gdelta$ we have $d(x,z) = \delta = d(z,y)$,
  hence $M$ is an mmodule.

  Suppose now that $(X,d)$ is Robinson and let $\<$ be any compatible order, with minimum $x_\star$ and maximum
  $x^\star$. Suppose that $M$ does not contain $x_\star$ (or
  symmetrically $x^\star$). Then for any $y,z \in M$ with $y<z$, we have
  $d(y,z) \leq d(x_\star,z) = \delta$, hence
  $\diam(M) \leq \delta$. Since at most one connected component of
  $\Gdelta$ may contain both points $x_\star,x^\star$, we conclude
  that any other component has diameter at most $\delta$.

  Let $M$ be a connected component of $\Gdelta$ with
  $\diam(M) \leq \delta$. We assert that $M$ is an interval of
  $\<$. Let $x \< y \< z$ with $x, z \in M$, $y \notin M$. Consider a
  path $P$ from $x$ to $z$ in $\Gdelta$. Necessarily $P$ contains an
  edge $x'z'$ with $x',z' \in M$ and $x'\< y \< z'$. Then by
  definition of edges of $\Gdelta$ and since $\diam(M) \leq \delta$,
  we have $d(x',z') < \delta$. Consequently,
  $d(x',y) \leq d(x',z') < \delta$, proving that $y \in M$, contrary
  to our choice of $y$. Thus $M$ is an interval of $\<$.

  If there is a connected component $M_0$ of $\Gdelta$ with
  $\diam(M_0) > \delta$, then all other connected components have
  diameter at most $\delta$, and $x_\star,x^\star$ is a diametral pair
  of $M_0$. Hence for every compatible order, each other connected
  component $M$ of $\Gdelta$ is an interval, hence $M$ is a block.

  Otherwise, if a connected component of $\Gdelta$ with diameter
  larger than $\delta$ does not exist, then we prove that there is no
  component $M$ with $x_\star, x^\star \in M$ for any compatible order
  $\<$. By way of contradiction assume there is one, and let
  $y \notin M$. For any $x, z \in M$ with $x \< y \< z$, we have
  $d(x,z) \geq \max \{d(x,y), d(y,z) \} = \delta$. Since
  $\diam(M)\le \delta$, this yields $d(x,z) = \delta$. But then the
  points $x$ such that $x\< y$ and $z$ such that $y\< z$ cannot belong
  to the same connected component of $\Gdelta$. This contradicts the
  fact that $x_\star$ and $x^\star$ belong to $M$. Consequently, in
  this case each connected component is an interval in any compatible
  order, hence is a block.
\end{proof}

The following result shows that either the graph $\Gdelta$ is
connected for all values of $\delta>0$ or there exists a unique
positive value of $\delta$ such that $\Gdelta$ is not
connected. Furthermore, we characterize the dissimilarity % Robinson—>
spaces for which
the second option occurs.

\begin{lemma}\label{lemma:copartition-delta}
  Let $(X,d)$ be a dissimilarity space such that $\Mmax$ is a copartition of
  $X$. Then there exists a unique $\delta > 0$ such that for any maximal
  mmodule $M$, for any $x \in M$ and $y \in X \setminus M$, we have
  $d(x,y) = \delta$. Consequently, the graph $\Gdelta$ is not connected and
  each connected component of $\Gdelta$ is the complement of a maximal
  mmodule. Conversely, if there exists $\delta > 0$ such that $\Gdelta$ is
  not connected, then $\Mmax$ is the copartition of $X$ consisting of
  the complements of the $\delta$-mmodules.
\end{lemma}

\begin{proof}
  First, let $\Mmax$ be a copartition. Let $M, M'$ be two maximal
  mmodules, and let $x \in M$, $x' \in X \setminus M$, $y' \in M'$ and
  $y \in X \setminus M'$. We assert that $d(x,x') = d(y,y')$, allowing
  us to set $\delta=d(x,x')$. As $\Mmax$ is a copartition,
  $X \setminus M' \subset M$ and $X \setminus M \subset M'$, hence
  $x' \in M'$ and $y \in M$.  Since $x, y \in M, x' \in X \setminus M$
  and $M$ is an mmodule, we deduce that
  $d(x,x') = d(y,x')$. Analogously, since $x',y' \in M'$,
  $y \in X \setminus M'$ and $M'$ is an mmodule, we get
  $d(x',y) = d(y',y)$, and thus $d(x,x') = d(y,y')$. Then $\Gdelta$ is
  not connected as each maximal mmodule and its complement define a
  cut of $\Gdelta$. The connected components of $\Gdelta$ are the
  subsets of the partition defined by the complements of maximal
  mmodules.

  Conversely, let $\delta>0$ be such that the graph $\Gdelta$ is not
  connected. Then any arbitrary union of connected components of
  $\Gdelta$ is an mmodule. Thus, for each component $C$ of $\Gdelta$,
  its complement $M = X \setminus C$ is an mmodule. We assert that $M$
  belongs to $\Mmax$.  Let $M'$ be an mmodule containing $M$. Then
  $C' := X \setminus M'$ is a subset of $C$. Since $M'$ is an mmodule,
  for any $x \in C', y \in M'$, $d(x,y) = \delta$ holds. This implies
  that $M'$ and $C'$ define a cut of $\Gdelta$, thus $C'$ is a union
  of connected components of $\Gdelta$. Since $C' \subseteq C$ and $C$
  is a connected component of $\Gdelta$, we conclude that $C' = C$ and
  $M' = M$, establishing the maximality of $M$.
\end{proof}

We can actually give a simple characterization of the unique value of
$\delta$ such that the graph $\Gdelta$ is not connected.

\begin{lemma}\label{lemma:find-rho}
  Let $(X,d)$ be a dissimilarity space such that $\Mmax$ is a copartition
  of $X$, and $\delta > 0$ be the unique value such that $\Gdelta$ is
  not connected. Let $\MST$ be a minimum spanning tree on the complete
  graph on vertex set $X$, with weights given by $d$.
  Then
  $$\delta = \max \{d(x,y) : xy \in T \}
  = \min \{\delta' > 0 : \textrm{the graph } (X, \{ xy : x \neq y, d(x,y) \leq \delta'\})\textrm{ is connected}\}.$$
\end{lemma}

\begin{proof}
  Let $\MST \subseteq \binom{X}{2}$ be a minimum-weight spanning tree on
  $X$. Since $\Gdelta$ is not connected, $\MST$ must contain an edge with
  weight $\delta$, hence
  $\max \{ d(x,y) : xy \in \MST\} \geq \delta$.

  Let $e=xy$ be any edge of $\MST$. If $e$ joins two components of
  $\Gdelta$, then its weight is $\delta$.  Else, let $M$ be the
  connected component of $\Gdelta$ containing $x$ and $y$. Pick any
  $z \notin M$ (such $z$ exists since $\Gdelta$ is not connected).  We
  may assume that $y$ lies between $x$ and $z$ in the tree $T$.  Then
  $\MST \setminus \{e\} \cup \{xz\}$ is a spanning tree. By minimality of
  $\MST$, $d(x,y)\le d(x,z)$. Since $x \in M, z \notin M$, we also have
  $d(x,z) = \delta$. Consequently, $d(x,y) \le \delta$ for any edge
  of $\MST$. This proves the first inequality.

  As $\Gdelta$ is not connected, the graph
  $G_{< \delta} := (X, \{ xy : d(x,y) < \delta\})$ is not
  connected. Moreover
  $G_{\leq \delta} := (X, \{xy : d(x,y) \leq \delta\})$ is a
  supergraph of the complement of $\Gdelta$. As the complement of a
  not-connected graph is connected, $G_{\leq \delta}$ is connected, proving
  the last inequality.
\end{proof}

\begin{definition}[Connected and non-connected dissimilarity %Robinson —>
spaces, large $\deltas$-mmodules]
  Let $\rhos$ be the minimum value such that the  graph
  $G_{\leq\rho}=(X,\{xy : x \neq y, d(x,y) \leq \rhos\})$ is connected. If the graph $\Grho$ is
  not connected, then we say that the dissimilarity space $(X,d)$ is
  \emph{non-connected}. Otherwise, all graphs $\Gdelta$ are connected
  for $\delta > 0$, and we say that the dissimilarity space $(X,d)$ is
  \emph{connected}. 
  If $(X,d)$ is a non-connected Robinson space and $\Grho$
  contains a connected component with diameter larger that $\rhos$
  (see \Cref{lemma:delta-mmodules}), then we call this component a
  \emph{large $\rhos$-mmodule}.
\end{definition}

Using the relationship between the minimum spanning tree and the
dendrogram  $\Td$ of the subdominant ultrametric $\hd$ of $d$, we get that $\rho$
is the weight of the root in the vertex representation of $\Td$. 
Observe also that $\rhos$ coincides with $\rhos(\alpha)$ when $\alpha$ is a
P-node root of $\TPQ$ or a $\cap$-node root of $\TM$. The next result
establishes some properties of large $\rhos$-mmodules.

\begin{lemma}\label{lemma:large_block}
  Let $(X,d)$ be a non-connected Robinson space containing a large
  $\rhos$-mmodule $S_0$. Then $S_0$ has the following properties:
  \begin{enumerate}[label=(\roman*)]
  \item\label{item:s01} there is a unique bipartition
    $S_0 = S_\star \cup S^\star$ of $S_0$ into two blocks,
  \item\label{item:s02} $\diam(S_\star) \leq \rhos$ and
    $\diam(S^\star) \leq \rhos$,
  \item\label{item:s03} $d(x,y) \geq \rhos$ for any
    $x \in S_\star, y \in S^\star$,
  \item\label{item:s04} if there exists $\delta > 0$ such that
    ${\Gdelta}(S_0)$ is not connected, then $\delta > \rho$ and
    the connected components of ${\Gdelta}(S_0)$ are $S_\star$ and
    $S^\star$.
  \end{enumerate}
\end{lemma}

\begin{proof}
  Let $y \in X \setminus S_0$, and $\<$ be any compatible order for
  $(X,d)$. Since $\diam(S_0)> \rhos$, $S_0$ contains at least two
  points.  Let $xz$ be any edge of $\Grho(S_0)$ and suppose that
  $x \< z$.  If $d(x,z) < \rhos$, then, as  $d(x,y) = d(y,z) = \rhos$,
   either $x \< z \< y$ or $y \< x \< z$ holds. Otherwise, by
  definition of the edges of $\Grho$ we have $d(x,z) > \rhos$, and
  thus $x \< y \< z$ holds. Setting
  $E_{> \rhos} := \{x'z' \in E(\Grho): d(x',z') > \rhos\}$, it
  implies that for any $xz \in E(\Grho)$ and
  $y \in X \setminus S_0$, $y$ is between $x$ and $z$ in any
  compatible order if and only if $xz \in E_{> \rhos}$.

  We can extend the previous argument to any pair of vertices $x,z$ of
  $S_0$. Since $S_0$ is connected in $\Grho$, there exists at least
  one $(x,z)$-path and let
  $P=(x=x_0x_1,x_1x_2,\ldots,x_{\ell-1}x_\ell=z)$ be an arbitrary
  $(x,z)$-path.  Then the parity of the number of edges $x_ix_{i+1}$
  of $P$ such that $y$ is between $x_i$ and $x_{i+1}$ is odd if and
  only if $y$ is between $x=x_0$ and $z=x_\ell$. In particular, this
  parity is the same for all paths with extremities $x$ and
  $z$. Fixing $x \in S_0$, we can thus define the following
  bipartition of $S_0$:
  $$S^\star := \{z \in S_0: \textrm{for each $(x,z)$-path $P$ of $\Gdeltas$, $|P \cap E_{> \rhos}|$ is even}\},$$
  $$S_\star := \{z \in S_0: \textrm{for each $(x,z)$-path $P$ of $\Gdeltas$, $|P \cap E_{> \rhos}|$ is odd}\}.$$

  Then for any $y \in X \setminus S_0$, $y$ is between any pair
  $x \in S_\star$, $z \in S^\star$, hence $X \setminus S_0$, $S_\star$
  and $S^\star$ are blocks, and this bipartition is unique,
  establishing~\ref{item:s01}.

  By construction, the bipartition $(S_\star,S^\star)$ has the
  properties~\ref{item:s02} and~\ref{item:s03}. For~\ref{item:s04},
  suppose there exists $\delta > 0$ such that ${\Gdelta}(S_0)$ is
  not connected. Then $\delta \neq \rhos$, as $S_0$ is a connected
  component of $\Grho$. The inequality $\delta < \rhos$ is also
  impossible, since the complete bipartite graph on $(S_\star,S^\star)$
  would be a subgraph of ${\Gdelta}(S_0)$.  Hence
  $\delta > \rhos$. By~\ref{item:s02}, each of the sets $S_\star$
  and $S^\star$ induces a connected subgraph in ${\Gdelta}(S_0)$,
  hence $S_\star$ and $S^\star$ are the connected components of
  ${\Gdelta}(S_0)$.
\end{proof}

The existence of large $\rhos$-mmodules will induce
some irregularities in both tree representations of Robinson
dissimilarities. The next definitions capture those irregularities,
and will be useful when comparing PQ-trees and mmodule trees.

\begin{definition}[$\delta$-special, large, $\delta$-conical, apex,
  split and standard nodes]
  Let $(X,d)$ be a Robinson space with mmodule tree $\TM$ and PQ-tree
  $\TPQ$. Let $\delta > 0$.

  A $\cap$-node $\alpha = \cap(\beta_1,\ldots,\beta_k)$ of $\TM$ is
  called \emph{$\delta$-special} if for all distinct
  $j, j' \in \{1,\ldots,k\}$, we have
  $d(X(\beta_j),X(\beta_{j'})) = \delta$ and there is
  $i \in \{1,\ldots,k\}$ such that $\diam(X(\beta_i)) > \delta$
  holds. Then $\beta_i$ is unique and called the \emph{large} child of
  $\alpha$. A $\cap$-node is \emph{special} if it is $\delta$-special
  for some $\delta > 0$.

  A Q-node $\alpha = Q(\beta_1,\ldots,\beta_k)$ of $\TPQ$ is called
  \emph{$\delta$-conical} if there is a child $\beta_i$ such that for
  all $j \in \{1,\ldots,k\} \setminus \{i\}$, we have
  $d(X(\beta_i),X(\beta_j)) = \delta$. Then $\beta_i$ is unique and
  called the \emph{apex} child of $\alpha$. A Q-node is \emph{conical}
  if it is $\delta$-conical for some $\delta > 0$.

  If $\alpha$ is the apex child of a $\delta$-conical Q-node and
  $\Gdelta(X(\alpha))$ is not connected, then $\alpha$ is called a
  {\em split node}. If a node $\alpha$ of $\TPQ$ is not split, then it
  is \emph{standard}.
\end{definition}

The uniqueness of $\beta_i$ in both definitions can be readily
checked:
\begin{itemize}
\item The uniqueness of the large child of a special $\cap$-node
  derives from \Cref{lemma:delta-mmodules}.
  % \Cref{lemma:copartition-delta,lemma:large_block}.
\item Suppose that a conical node
  $\alpha = Q(\beta_1, \ldots, \beta_k)$ has two apex children
  $\beta_i$ and $\beta_j$. First notice that the distance between
  $\beta_i$ and the other children is the same that the distance
  between $\beta_j$ and the other children (its value is
  $d(X(\beta_i),X(\beta_j))$). It would be possible to exchange the
  nodes $\beta_i$ and $\beta_j$ in the list of the children of
  $\alpha$, a contradiction.

\end{itemize}

\subsection{Construction of the PQ-tree}
The three next propositions describe how to build the nodes
of the PQ-tree of a Robinson dissimilarity through the analysis of its
$\rhos$-mmodules.

\begin{proposition}\label{lemma:no-d-cut}
  Let $(X,d)$ be a connected Robinson space.
  Then the following assertions hold:
  \begin{enumerate}[label=(\roman*)]
  \item\label{item:ndc1} $\Mmax$ is a partition of $X$ with $|\Mmax| \geq 3$,
  \item\label{item:ndc2} each maximal mmodule is a block,
  \item\label{item:ndc3} the quotient space $(X/\Mmax,\widehat{d})$ is flat,
  \item\label{item:ndc4} the compatible orders on $X$ are exactly the
    composition of each of the two compatible orders of
    $(X/\Mmax,\widehat{d})$ with the compatible orders of each maximal
    mmodule,
  \item\label{item:ndc5} the root of $\TPQ$ is a Q-node, whose
    children are the PQ-trees of the maximal mmodules, sorted by the
    compatible orders of $(X/\Mmax,\widehat{d})$.
  \end{enumerate}
\end{proposition}

\begin{proof}
  If $\Mmax = \{M_1, \ldots, M_k\}$ is a copartition of $X$, then by
  \Cref{lemma:copartition-delta}, $\Grho$ is not connected,
  contradicting the hypothesis. Hence~\ref{item:ndc1} holds.

  Consider the quotient space $(X / \Mmax,\widehat{d})$. By
  construction, its mmodules are all trivial. Hence by
  \Cref{thm:flat}, $(X / \Mmax,\widehat{d})$ is flat and has a unique
  compatible order $\<_\Mmax$ up to reversal, that is,~\ref{item:ndc3}
  holds. We may assume that $M_1\<_\Mmax \ldots \<_\Mmax M_k$. Then
  for each compatible order $\<$ on $X$, for each choice of
  $x_i \in M_i$ for $i \in \{1,\ldots,k\}$ the restriction of $\<$ to
  $\{x_1,\ldots,x_k\}$ coincides with $\<_\Mmax$ or its reversal, that
  is either $x_1 \< \ldots \< x_k$ or $x_k \< \ldots \< x_1$. This
  implies that each $M_i$ is an interval for $\<$,
  hence~\ref{item:ndc2} holds.

  By \Cref{LEMMA_block_node}, each $M_i$ is a node $\beta_i$ in
  $\TPQ$. By \Cref{lemma:mmod-block-orders}, the relative order of
  elements in $M_i$ can be chosen independently from the order of the
  blocks. Thus~\ref{item:ndc4} holds. Finally by
  \Cref{prop:charac-PQtree}, it implies~\ref{item:ndc5}.
\end{proof}

\begin{proposition}\label{lemma:some-d-cut}
  Let $(X,d)$ be a non-connected Robinson space and suppose that the
  diameter of each $\rhos$-mmodule is at most $\rhos$.  Then the root
  of $\TPQ$ is a P-node whose children are the PQ-trees of each
  $\rhos$-mmodule.
\end{proposition}

\begin{proof}
  Let $M$ be a $\rhos$-mmodule. By \Cref{lemma:delta-mmodules},
  $M$ is an mmodule and a block. By \Cref{LEMMA_block_node},
  there exists a node $\beta_M$ in $\TPQ$ such that
  $X(\beta_M) = M$. Moreover, by \Cref{lemma:mmod-block-orders}, in
  any compatible order $<$, reordering the elements of $M$ into an
  order compatible with $\beta_M$ gives a compatible order on $X$.
  Furthermore, since for any $x,y \in X$ that are not in the same
  mmodule, $d(x,y) = \rhos$ holds, any order between the blocks
  corresponding to each $\rhos$-mmodule is compatible. Hence each
  $\rhos$-mmodule is a maximal block, and the result follows by
  \Cref{prop:charac-PQtree}.
\end{proof}

\begin{proposition}\label{lemma:cut-separable}
  Let $(X,d)$ be a non-connected Robinson containing a large
  $\rhos$-mmodule $S_0$. Let $S_1,\ldots,S_k$ be the other
  $\rhos$-mmodules. Then the root of $\TPQ(S_0)$ is a Q-node
  $Q(\beta_1,\ldots,\beta_\ell)$ or a P-node $P(\beta_1,\beta_2)$ (and
  $\ell = 2$), and the root of the PQ-tree $\TPQ(X)$ is a special
  Q-node
  $Q(\beta_1,\ldots,\beta_{i-1},\beta,\beta_{i},\ldots,\beta_\ell)$,
  obtained by adding an apex child $\beta$ to the root of $\TPQ(S_0)$,
  where
  \begin{equation*}
    \beta = \left\{\begin{array}{ll}
                      \TPQ(S_1) & \quad\textrm{ if } k = 1,\\
                      P(\TPQ(S_1),\ldots,\TPQ(S_k)) & \quad\textrm{ if } k \geq 2.
                    \end{array}\right.
  \end{equation*}
\end{proposition}

\begin{proof}
  By \Cref{lemma:large_block}, there is a bipartition
  $S_0 = S_\star \cup S^\star$ of $S_0$ into two blocks, with
  $\diam(S_\star) \leq \deltas$, $\diam(S^\star) \leq \rhos$, and
  $d(x,y) \geq \rhos$ for each $x \in S_\star$, $y \in S^\star$.

  We claim that the root of $\TPQ(S_0)$ is a Q-node or has arity
  two. If ${\Gdelta}[S_0]$ is connected for every $\delta > 0$, then
  by \Cref{lemma:no-d-cut}, the root of $\TPQ(S_0)$ is a
  Q-node. Otherwise, let $\delta > 0$ be such that ${\Gdelta}(S_0)$
  is not connected. By \Cref{lemma:large_block}\ref{item:s04},
  $\delta > \rhos$ and ${\Gdelta}(S_0)$ has two connected components
  $S_\star$ and $S^\star$, each of diameter less than $\delta$. By
  \Cref{lemma:some-d-cut}, the root of $\TPQ(S_0)$ is a P-node
  $P(\TPQ(S_\star),\TPQ(S^\star))$.

  Let $\beta_1, \ldots, \beta_l$ be the children of the root $\alpha$
  of $\TPQ(S_0)$, with $k \geq 2$. By \Cref{lemma:block}, $S_\star$
  and $S^\star$ are induced by consecutive children of $\alpha$, that
  is, up to symmetry, there is $i \in \{2,\ldots,\ell\}$ such that
  $S_\star = \bigcup_{j=1}^{i-1} X(\beta_i)$ and
  $S^\star = \bigcup_{j=i}^\ell X(\beta_i)$. Let $\beta$ be defined as in
  the statement of the result.

  We claim that $\beta = \TPQ(X \setminus S_0)$. If $k=1$ this is
  immediate from the definition of $\beta$. If $k \geq 2$, then
  $\Grho(X \setminus S_0)$ is not connected and its components are
  precisely $S_1,\ldots,S_k$. Thus the claim follow by
  \Cref{lemma:some-d-cut}.

  Then we show that
  $\TPQ(X) = Q(\beta_1,\ldots,\beta_{i-1},\beta,\beta_i,\ldots,\beta_\ell)$.
  Let $\<$ be a compatible order of $(X,d)$.  By
  \Cref{lemma:large_block}, $S_\star$ and $S^\star$ are blocks. Notice
  that the restriction of $\<$ to $S_0$ is represented by
  $\TPQ(S_0)$. By \Cref{lemma:delta-mmodules}, $S_1,\ldots,S_k$ are
  blocks, implying that the restriction of $\<$ to $X \setminus S_0$
  is represented by $\beta$. Moreover as
  $d(S_\star,S_j) = d(S_j,S^\star) = \deltas < \max \{d(x,y) : x \in S_\star, y \in S^\star\} = \diam(S_0)$,
  we have that $S_j$ is between $S_\star$ and $S^\star$ in any
  compatible order. Thus $\<$ is represented by
  $Q(\beta_1,\ldots,\beta_{i-1},\beta,\beta_i,\ldots,\beta_\ell)$.

  Conversely, let $\<$ be an order represented by
  $Q(\beta_1,\ldots,\beta_{i-1},\beta,\beta_i,\ldots,\beta_\ell)$  with
  $S_\star \< S^\star$. We show that $\<$ is a compatible order of
  $(X,d)$. Let $x \< y \< z$ be a triplet of points in $X$. As the
  restriction of $\<$ is represented by $\TPQ(S_0)$, if
  $x,y,z \in S_0$, then $\max \{d(x,y), d(y,z)\} \leq d(x,z)$. If
  $x,y \in S_0$, $z \in X \setminus S_0$, then as $x,y \in S_\star$,
  $d(x,y) \leq \rhos = d(y,z) = d(x,z)$. The case $y,z \in S_0$,
  $x \in X \setminus S_0$ is similar. If $x,z \in S_0$ and
  $y \in X \setminus S_0$, then $x \in S_\star$, $z \in S^\star$,
  hence $d(x,z) \geq \rhos = d(x,y) = d(y,z)$. If
  $x,y \in X \setminus S_0$ and $z \in S_0$, then
  $d(x,y) \leq \rhos = d(y,z) = d(x,z)$. The case $x \in S_0$,
  $y, z \in X \setminus S_0$ is similar. Finally, if
  $x,y,z \in X \setminus S_0$, as the restriction of $\<$ to
  $X \setminus S_0$ is represented by $\beta$, again
  $\max \{d(x,y), d(y,z) \} \leq d(x,z)$. Hence $\<$ is compatible.
\end{proof}

\begin{algorithm}[htbp]
  \caption{Computes the PQ-tree of $(X,d)$ using the $\rhos$-mmodules.}
  \label{algo:deltaToPqt}
  \flushleft{$\underline{\dtopq(S)}$}
  \begin{algorithmic}[1]
  \Require{a Robinson space $(X,d)$ (implicit), a set $S \subseteq X$.}
  \Ensure{the PQ-tree $\TPQ(S)$.}
  \If {$(S,d)$ is connected} \Comment{\mbox{\Cref{lemma:no-d-cut}}}
     \Let $M_1,\ldots,M_l$ be the maximum mmodules of $(S,d)$
     \Let $x_1 \in M_1, \ldots, x_\ell \in M_\ell$\label{line:maxmmod}
     \Let $x_{\sigma(1)} \< \ldots \< x_{\sigma(\ell)}$ be a compatible order of the flat Robinson space $\{x_1,\ldots,x_\ell\}$ \label{line:flat-order}
     \Return $Q(\dtopq(M_{\sigma(1)}), \ldots \dtopq(M_{\sigma(\ell)}))$
  \EndIf
  \State Compute $\rhos$ for $(S,d)$ \Comment{\mbox{\Cref{lemma:find-rho}}}
  \Let $S_0, S_1, \ldots, S_k$ be the connected components of the graph $\Grho$ on vertex set $S$
  \Let $\mathcal{T}_0, \ldots, \mathcal{T}_k = \dtopq(S_0),\ldots,\dtopq(S_k)$
  \If{there is $j \in \{0,\ldots,k\}$ with $\diam(S_j) > \rhos$} \Comment{\mbox{\Cref{lemma:cut-separable}}}
     \State $S_j \leftrightarrow S_0$, $\mathcal{T}_0 \leftrightarrow{\mathcal{T}_j}$ \Comment{\mbox{ensures $S_0$ is the large $\rhos$-mmodule}}
     \Let $Q(\beta_1,\ldots,\beta_\ell) := \mathcal{T}_0$, or $P(\beta_1,\beta_\ell) := \mathcal{T}_0$ and $\ell = 2$
     \Let $i \in \{2,\ldots,\ell\}$ be minimal such that $d(\beta_i,\beta_\ell) \leq \rhos$ and $d(\beta_{i-1}, \beta_{i}) \geq \rhos$,\label{line:def-i}
     \Let $\beta :=  \mathcal{T}_1$ if $k=1$, $\beta := P(\mathcal{T}_1, \ldots, \mathcal{T}_k)$ if $k \geq 2$
     \Return $Q(\beta_1,\ldots,\beta_{i-1},\beta,\beta_{i}, \ldots, \beta_{\ell})$ \Comment{\mbox{conical node with apex child $\beta$}}
  \EndIf
  \Return $P(\mathcal{T}_0,\ldots,\mathcal{T}_k)$ \Comment{\mbox{\Cref{lemma:some-d-cut}}}
  \end{algorithmic}
\end{algorithm}

\Cref{lemma:no-d-cut,lemma:some-d-cut,lemma:cut-separable} lead to
\Cref{algo:deltaToPqt} that builds the PQ-tree of $(X,d)$ using
$\Grho$. It requires an algorithm able to compute a compatible order
of a flat Robinson space on Line~\ref{line:flat-order}, otherwise it
can only compute the structure of the PQ-tree, without the ordering of
children of Q-nodes. We gave such an algorithm in our
paper~\cite[Proposition 6.16]{CaCheNaPre}, that runs in time
$O(|X|^2)$. The value of $\rho$ can be efficiently computed by
\Cref{lemma:find-rho}, as fast as the computation of a minimum
spanning tree of a complete graph. Determining the maximal mmodules on
Line~\ref{line:maxmmod} will be possible using an algorithm presented in
\Cref{SECTION_mmodtree}.

On line~\ref{line:def-i}, the value of $i$ in
\Cref{lemma:cut-separable} is computed as the minimal value $i^\star$
such that ($\ast$) $d(\beta_{i^\star},\beta_\ell) \leq \rhos$ and
$d(\beta_{i^\star-1},\beta_{i^\star}) \geq \rhos$; we now prove that
this computation is correct. We use the notations of
\Cref{lemma:cut-separable}; the root of $\TPQ(X)$ is
$Q(\beta_1,\ldots,\beta_{i-1},\beta,\beta_i,\ldots,\beta_\ell)$ and
$d(\beta,\beta_j) = \deltas$ for each $j \in \{1,\ldots,l\}$.  First
observe that $d(\beta_i,\beta_\ell) \leq d(\beta,\beta_\ell) = \rhos$ and
$d(\beta_{i-1},\beta_i) \geq d(\beta,\beta_i) = \rhos$, hence ($\ast$)
holds for $i$. Now suppose for the sake of contradiction that there
exists $i' \in \{1,\ldots,i-1\}$ for which ($\ast$) holds. Then
for all $j \in \{i',\ldots,i-1\}$ and $j' \in \{i,\ldots,\ell\}$,
$$\rhos = d(\beta,\beta_i) \leq d(\beta_{i-1},\beta_i) \leq d(\beta_j,\beta_{j'}) \leq d(\beta_{i'},\beta_\ell) \leq \rhos,$$
and for all $j \in \{i',\ldots,i-1\}$ and $j' \in \{1,\ldots,i'-1\}$
$$\rhos \leq d(\beta_{i'-1},\beta_{i'}) \leq d(\beta_{j'},\beta_j) \leq d(\beta_1,\beta_{i-1}) \leq d(\beta_1,\beta) = \rhos,$$
proving that all those quantities equal $\rhos$. In particular, for
all $x \in X(\beta_{i'}) \cup \ldots \cup X(\beta_{i-1})$ and
$y \in X(\beta_1) \cup \ldots \cup X(\beta_{i'-1}) \cup X(\beta_i) \cup \ldots \cup X(\beta_\ell)$,
$d(x,y) = \rhos$, contradicting the fact that $S_0$ is connected in
$\Grho$.

%% file: sec_5_translations.tex
In this section, we show how to build the mmodule tree of a Robinson
space $(X,d)$ from its PQ-tree, and \emph{vice-versa}, how to build
the PQ-tree from the mmodule tree.  This cryptomorphism is illustrated
by \Cref{fig:conversion}.

\begin{figure}[htbp]
  \tikzset{leaf/.style={}}
  \tikzset{cupnode/.style={circle,draw=black}}
  \tikzset{capnode/.style={circle,draw=black}}
  \tikzset{qnode/.style={rectangle,minimum height=0.5cm,draw=black}}
  \tikzset{pnode/.style={circle,minimum height=0.5cm,draw=black}}
  \tikzset{pnodellipse/.style={ellipse,minimum height=0.5cm,draw=black}}
  \begin{center}
    \footnotesize
    \begin{tabular}{rcl}
      \begin{tikzpicture}[x=1cm,y=1cm]%leaf
        \draw (0,0) node {Leaf $x$};
      \end{tikzpicture}
      & \raisebox{0.2cm}{\tikz\draw[<->] (0,0) -- (1cm,0) node[midway,above] {(a)};}
      & \begin{tikzpicture}[x=1cm,y=1cm]%leaf
          \draw (0,0) node {Leaf $x$};
        \end{tikzpicture}
      \\[0.4cm]
      \begin{tikzpicture}[x=1cm,y=1cm]%cup-node
        \node[leaf] (b1) at (-1,-1) {$\beta_1$};
        \node[leaf] (bk) at (1,-1) {$\beta_k$};
        \node[leaf] (bi) at (0,-1) {\ldots\phantom{$\beta$}};
        \node[cupnode] (root) {$\cup$};
        \foreach \v in {b1.north,bi.north west,bi.north,bi.north east,bk.north} {
          \draw (\v) -- (root);
        }
      \end{tikzpicture}
      & \raisebox{0.5cm}{\tikz\draw[<->] (0,0) -- (1cm,0) node[midway,above] {(b)};}
      & \begin{tikzpicture}[x=1cm,y=1cm]%Q-node
        \node[leaf] (b1) at (-2,-1) {$\phi(\beta_{\sigma(1)})$};
        \node[leaf] (bk) at (2,-1) {$\phi(\beta_{\sigma(k)})$};
        \draw (0,-1) node {\ldots};
        \node[qnode,minimum width=4.2cm] (root) at (0,0) {};
        \foreach \v in {b1,bk} {
          \draw (\v.north) -- (\v.north |- root.south);
        }
        \foreach \i in {2,3,...,5} {
          \node[leaf] (b\i) at ($(-2.75,-1) + \i*(0.75,0)$) {\phantom{$\beta$}};
          \draw (b\i.north) -- (b\i.north |- root.south);
        }
      \end{tikzpicture}
      \\[0.4cm]
      \begin{tikzpicture}[x=1cm,y=1cm]%cap-node
        \node[leaf] (b1) at (-1,-1) {$\beta_1$};
        \node[leaf] (bk) at (1,-1) {$\beta_k$};
        \node[leaf] (dots) at (0,-1) {\ldots};
        \node[capnode] (root) {$\cap$};
        \foreach \v in {b1,bk} {
          \draw (\v) -- (root);
        }
        \foreach \i in {2,3,4} {
          \node[leaf] (b\i) at ($(-1.5,-1) + 1/2*(\i,0)$) {\phantom{$\beta$}};
          \draw (b\i) -- (root);
        }
      \end{tikzpicture}
      & \raisebox{0.5cm}{\tikz\draw[<->] (0,0) -- (1cm,0) node[midway,above] {(c)};}
      & \begin{tikzpicture}[x=1cm,y=1cm]%P-node
        \node[leaf] (b1) at (-1,-1) {$\phi(\beta_1)$};
        \node[leaf] (bk) at (1,-1) {$\phi(\beta_k)$};
        \node[leaf] (dots) at (0,-1) {\ldots};
        \node[pnode] (root) at (0,0) {};
        \foreach \v in {b1,bk} {
          \draw (\v) -- (root);
        }
        \foreach \i in {2,3,4} {
          \node[leaf] (b\i) at ($(-1.5,-1) + 1/2*(\i,0)$) {\phantom{$\beta$}};
          \draw (b\i) -- (root);
        }
     \end{tikzpicture}
      \\[0.4cm]
      % \begin{tikzpicture}[x=1cm,y=1cm]%cap2
      %   \node[leaf] (b1) at (-0.5,-1) {$\beta_1$};
      %   \node[leaf] (bk) at (0.5,-1) {$\beta_2$};
      %   \node[capnode] (root) at (0,0) {$\cap$};
      %   \foreach \v in {b1,bk} {
      %     \draw (\v) -- (root);
      %   }
      % \end{tikzpicture}
      % & \raisebox{0.5cm}{\tikz\draw[<->] (0,0) -- (1cm,0) node[midway,above] {(d)};}
      % & \begin{tikzpicture}[x=1cm,y=1cm]%Q-node2
      %   \node[leaf] (b1) at (-0.5,-1) {$\phi(\beta_1)$};
      %   \node[leaf] (bk) at (0.5,-1) {$\phi(\beta_2)$};
      %   \node[qnode,minimum width=1.4cm] (root) at (0,-0.2) {};
      %   \foreach \v in {b1.north,bk.north} {
      %     \draw (\v) -- (\v |- root.south);
      %   }
      % \end{tikzpicture}
      % \\[0.4cm]
      \begin{tikzpicture}[x=1cm,y=1cm]
        \node[capnode] (root) at (0,0) {$\cap$};
        \draw (root.east) node[right] {$\delta$-special};
        \node[leaf] (b1) at (-1,-1) {$\beta_1$};
        \node[leaf] (bminus) at (0.5,-1) {$\beta_{k-1}$};
        \node[leaf] (bk) at (1.5,-1) {$\beta_k$};
        \draw (bk.south) node[below=-5pt] {large};
        \node (dots) at (-0.25,-1) {\ldots};
        \foreach \v in {b1,bminus,bk} {
          \draw (root) -- (\v);
        }
        \foreach \i in {2,3} {
          \node[leaf] (b\i) at ($(-1.5,-1) + 1/2*(\i,0)$) {\phantom{$\beta$}};
          \draw (root) -- (b\i);
        }

        \draw (2.5,-0.5) node {and $\phi(\beta_k) =$};
        \begin{scope}[xshift=4.8cm]
        \node[qnode,minimum width=2.6cm] (groot) at (0,-0.2) {};
        \node[leaf] (g1) at (-1.2,-1) {$\gamma_1$};
        \node[leaf] (gk) at (1.2,-1) {$\gamma_l$};
        \draw (0,-1) node {\ldots};
        \foreach \v in {g1,gk} {
          \draw (\v.north) -- (\v.north |- groot.south);
        }
        \foreach \i in {2,3,...,6} {
          \node[leaf] (g\i) at ($(-1.6,-1) + 2/5*(\i,0)$) {\phantom{$\gamma$}};
          \draw (g\i.north) -- (g\i.north |- groot.south);
        }
        \end{scope}
      \end{tikzpicture}
      & \raisebox{1cm}{\tikz\draw[->] (0,0) -- (1cm,0) node[midway,above] {(d)};}
      & \begin{tikzpicture}[x=1cm,y=1cm]%conical
        \node[qnode,minimum width=5.4cm] (root) at (0,-0.2) {$\delta$-conical};
        \node[leaf] (g1) at (-2.5,-1) {$\gamma_1$};
        \draw (-1.75,-1) node {\ldots};
        \node[leaf] (gj) at (-1,-1) {$\gamma_j$};
        \node[leaf] (gplus) at (1,-1) {$\gamma_{j+1}$};
        \draw (1.75,-1) node {\ldots};
        \node[leaf] (gl) at (2.5,-1) {$\gamma_l$};
        \node[pnodellipse] (apex) at (0,-1) {apex};
        \foreach \v in {g1,gj,gplus,gl,apex} {
          \draw (\v.north) -- (\v.north |- root.south);
        }
        \foreach \i in {2,3,9,10} {
          \node[leaf] (g\i) at ($(-3,-1) + 1/2*(\i,0)$) {\phantom{$\gamma$}};
          \draw (g\i.north) -- (g\i.north |- root.south);
        }
        \node (b1) at (-1,-1.8) {$\phi(\beta_1)$};
        \draw (0,-1.8) node {\ldots};
        \node (bk) at (1,-1.8) {$\phi(\beta_{k-1})$};
        \foreach \v in {b1,bk} {
          \draw (\v) -- (apex);
        }
        \foreach \i in {2,3,4} {
          \node (b\i) at ($(-1.5,-1.8) + 1/2*(\i,0)$) {\phantom{$\phi(\beta)$}};
          \draw (b\i) -- (apex);
        }
      \end{tikzpicture}
      \\[0.4cm]
      \begin{tikzpicture}[x=1cm,y=1cm]
        \begin{scope}[xshift=-2.2cm]
          \node[cupnode] (root) at (0,0) {$\cap$};
          \draw (root.east) node[right] {$\delta$-special};
          \node[leaf] (b1) at (-0.5,-1) {$\gamma_j$};
          \node[cupnode] (large) at (0.5,-1) {$\cup$};
          \draw (large.east) node[right] {large};
          \node[leaf] (b2) at (0,-2) {$\gamma_1$};
          \node[leaf] (bk) at (1,-2) {$\gamma_l$};
          \draw (0.5,-2) node {\ldots};
          \foreach \i in {3,4,5} {
            \node[leaf] (b\i) at ($(-0.5,-2) + \i*(0.25,0)$) {\phantom{$\beta$}};
            \draw (b\i) -- (large);
          }
          \foreach \u/\v in {b1/root,large/root,b2/large,bk/large} {
            \draw (\u) -- (\v);
          }
          \draw (0,-2.5) node {$\gamma_j$ is standard};
        \end{scope}
        \draw (0,-1.5) node {or};
        \begin{scope}[xshift=1.7cm]
          \node[cupnode] (root) at (0,0) {$\cap$};
          \draw (root.east) node[right] {$\delta$-special};
          \node[leaf] (g1) at (-1,-1) {$\beta_1$};
          \node[leaf] (gl) at (0.5,-1) {$\beta_k$};
          \foreach \i in {2,3,...,6} {
            \node[leaf] (g\i) at ($(-1.25,-1) + 1/4*(\i,0)$) {\phantom{$\gamma$}};
            \draw (g\i) -- (root);
          }
          \draw (-0.25,-1) node {\ldots};
          \node[cupnode] (large) at (1,-1) {$\cup$};
          \draw (large.east) node[right] {large};
          \node[leaf] (b2) at (0.5,-2) {$\gamma_1$};
          \node[leaf] (bk) at (1.5,-2) {$\gamma_l$};
          \draw (1,-2) node {\ldots};
          \foreach \i in {3,4,5} {
            \node[leaf] (b\i) at ($(0,-2) + \i*(0.25,0)$) {\phantom{$\beta$}};
            \draw (b\i) -- (large);
          }
          \foreach \u/\v in {g1/root,gl/root,large/root,b2/large,bk/large} {
            \draw (\u) -- (\v);
          }
          \draw (0.25,-2.5) node {$\gamma_j$ is split};
        \end{scope}
      \end{tikzpicture}
      & \raisebox{1cm}{\tikz\draw[<-] (0,0) -- (1cm,0) node[midway,above] {(e)};}
      & \begin{tikzpicture}[x=1cm,y=1cm]%conical
        \node[qnode,minimum width=5.4cm] (root) at (0,-0.2) {$\delta$-conical};
        \node[leaf] (g1) at (-2.5,-1) {$\phi(\gamma_1)$};
        \draw (-1.75,-1) node {\ldots};
        \node[leaf] (gj) at (-1,-1) {\phantom{$\phi(\gamma)$}};
        \node[leaf] (gplus) at (1,-1) {\phantom{$\phi(\gamma)$}};
        \draw (1.75,-1) node {\ldots};
        \node[leaf] (gl) at (2.5,-1) {$\phi(\gamma_l)$};
        \node[pnodellipse] (apex) at (0,-1.2) {apex};
        % \draw (apex.south east) node[right] {\scriptsize $\phi(\gamma_j)$};
        \foreach \v in {g1,gj,gplus,gl} {
          \draw (\v.north) -- (\v.north |- root.south);
        }
        \draw (apex.north) -- (apex.north |- root.south) node[midway,right] {$j$};
        \foreach \i in {2,3,9,10} {
          \node[leaf] (g\i) at ($(-3,-1) + 1/2*(\i,0)$) {\phantom{$\phi(\gamma)$}};
          \draw (g\i.north) -- (g\i.north |- root.south);
        }
        \node (b1) at (-1,-2) {$\phi(\beta_1)$};
        \draw (0,-2) node {\ldots};
        \node (bk) at (1,-2) {$\phi(\beta_{k})$};
        \foreach \v in {b1,bk} {
          \draw (\v) -- (apex);
        }
        \foreach \i in {2,3,4} {
          \node (b\i) at ($(-1.5,-2) + 1/2*(\i,0)$) {\phantom{$\phi(\beta)$}};
          \draw (b\i) -- (apex);
        }
        \draw (0,-2.5) node {\phantom{0}};
      \end{tikzpicture}
    \end{tabular}
  \end{center}
  \caption{The translations between mmodule trees and PQ-trees, as
    provided by \Cref{algo:pqtom,algo:mtopq}; here $\phi$ denotes the
    bijection from mmodule trees to PQ-trees. In cases (b), (c), none
    of the nodes is special or conical. In case (d), $\phi(b_k)$ can
    also be $P(\gamma_1,\gamma_2)$ (then $\ell = 2$), and if $k=2$ the
    apex node is simply $\phi(\beta_1)$. In case (f): the large node
    is a $\cap$-node when $l = 2$, and it may happen that the apex
    node is reduced to a leaf (in which case $\gamma_j$ is standard).}
  \label{fig:conversion}
\end{figure}
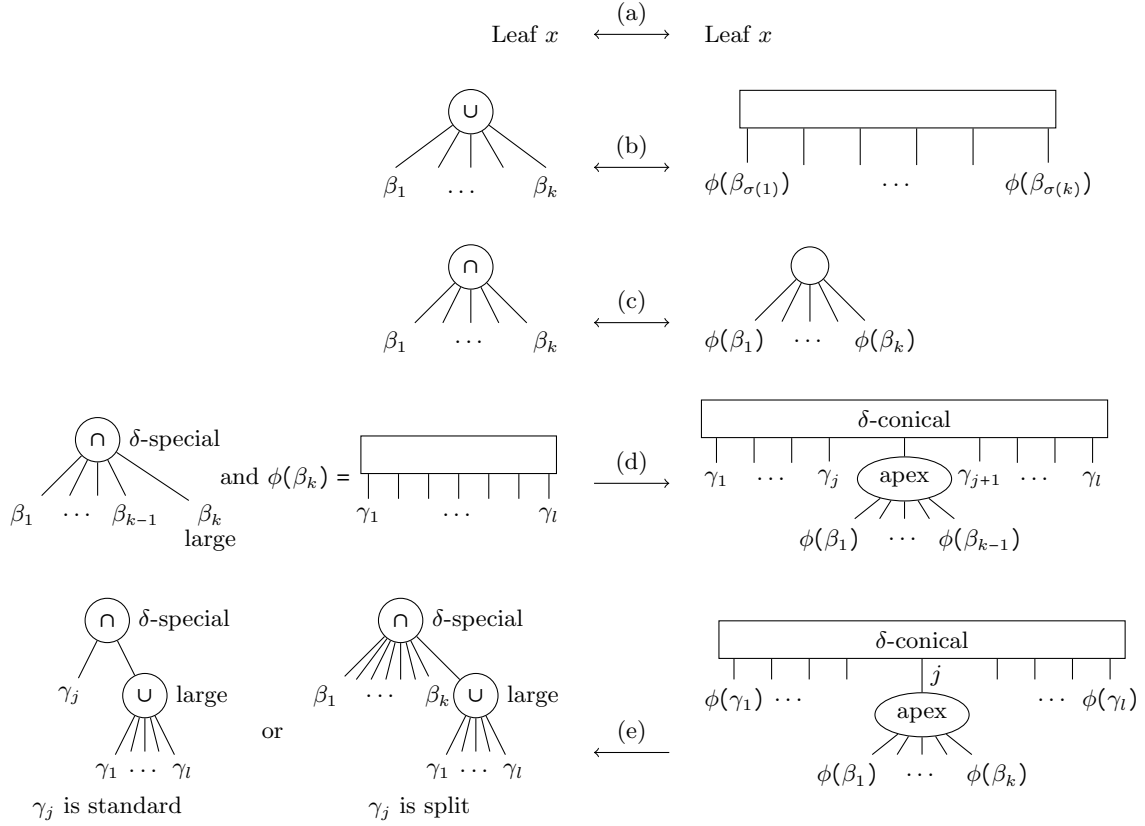

\subsection{The translation between \texorpdfstring{$\TPQ$ and $\TM$}{TPQ and TM}}
The cryptomorphism between the two trees relies on the following
remark: since
\Cref{lemma:no-d-cut,lemma:some-d-cut,lemma:cut-separable} cover all
possible cases, we can ``invert'' their statements depending on the
root of $\TPQ$ and of whether $(X,d)$ is connected or not. More
precisely we have the following result:

\begin{proposition}\label{prop:pqroot-analysis}
  Let $(X,d)$ be a Robinson space with PQ-tree $\TPQ$, then the set of
  maximum mmodules $\Mmax$ is described as follows:
  \begin{enumerate}[label=(\arabic*)]
  \item\label{item:pqr1} If $\TPQ = P(\beta_1,\ldots,\beta_k)$, then
    $(X,d)$ is non-connected and $X(\beta_1),\ldots,X(\beta_k)$ are
    the $\rhos$-mmodules, all of diameter at most $\rhos$, and
    $\Mmax = \{X \setminus X(\beta_1),\ldots, X \setminus X(\beta_k)\}$.
  \item\label{item:pqr2} If $\TPQ = Q(\beta_1,\ldots,\beta_k)$ and
    $(X,d)$ is connected, then
    $\Mmax = \{X(\beta_1),\ldots,X(\beta_k)\}$.
  \item\label{item:pqr4} If $\TPQ = Q(\beta_1,\ldots,\beta_k)$ and
    $(X,d)$ is non-connected, then there is a large $\rhos$-mmodule
    $S_0$, and there exists $i \in \{2,\ldots,k-1\}$ such that
    $S_0 = X \setminus X(\beta_i)$. The root of $\TPQ$ is conical with
    apex child $\beta_i$. Moreover,
    \begin{enumerate}[label=(\arabic{enumi}.\arabic{enumii})]
    \item\label{item:pqr41} if $c_\rhos = 2$, then
      $\Mmax = \{X(\beta_i), X \setminus X(\beta_i)\}$,
    \item\label{item:pqr43} if $c_\rhos \geq 3$, then
      $\beta_i = P(\gamma_1,\ldots,\gamma_{c_\rhos - 1})$ is a split node and
      $\Mmax = \{X \setminus X(\gamma_i) : i \in \{1,\ldots,c_\rhos-1\} \} \cup \{X(\beta_i)\}$.
    \end{enumerate}
  \end{enumerate}
\end{proposition}

\begin{proof}
  For~\ref{item:pqr1}, as the root of $\TPQ$ is a P-node,
  \Cref{lemma:some-d-cut} is the sole proposition allowing this
  conclusion, hence there is $\rhos > 0$ with $c_\rhos \geq 3$. Then
  by \Cref{lemma:copartition-delta},
  $\Mmax = \{X \setminus X(\beta_1),\ldots X \setminus X(\beta_k)\}$.

  Otherwise the root of $\TPQ$ is a Q-node. For~\ref{item:pqr2}, the
  only proposition where $\Gdelta$ is always connected is
  \Cref{lemma:no-d-cut}, from which we get that case.

  For~\ref{item:pqr4}, \Cref{lemma:cut-separable} is the only
  applicable proposition, hence there is $\rhos > 0$ such that $\Grho$
  is not connected, and has a large $\rhos$-mmodule $S_0$. It remains
  to prove what is $\Mmax$ in that case, but it actually easily
  follows from \Cref{lemma:copartition-delta}.
\end{proof}

\subsection{From \texorpdfstring{$\TPQ$ to $\TM$}{TPQ to TM}}
This leads to \Cref{algo:pqtom} to compute the mmodule tree from the
PQ-tree of a Robinson space. Given the PQ-tree $\TPQ$ of a Robinson
space $(X,d)$, we denote by $\rhos(\TPQ)$ the unique value $\delta$
for which $\Gdelta$ is not connected, if it exists, otherwise
$\rhos(\TPQ)$ is undefined. We can compute $\rhos=\rhos(\TPQ)$
efficiently in the following way:
\begin{itemize}[label=-]
\item if the root of $\TPQ$ is a P-node, then
  $\rhos:= d(\alpha, \beta)$, where $\alpha$ and $\beta$ are
  two distinct children of the root. This follows from
  \Cref{prop:pqroot-analysis}~\ref{item:pqr1},
\item otherwise, the root of $\TPQ$ is a Q-node with children
  $\beta_1,\ldots,\beta_k$, with $k \geq 3$. Then there is at most one
  $i \in \{2,\ldots,k-1\}$ such that for all
  $j \in \{1,\ldots,k\} \setminus \{i\}$, and
  $\rhos = d(\beta_i,\beta_j)$. This follows from
  \Cref{prop:pqroot-analysis}~\ref{item:pqr4}. Checking all
  possibilities for $i$ allows to find $\rhos$ (and $i$) in time
  $O(k)$, as it is sufficient to check that $d(\beta_1,\beta_i)$,
  $d(\beta_{i-1},\beta_i)$, $d(\beta_{i},\beta_{i+1})$ and
  $d(\beta_i,\beta_k)$ are all equal (then that value is
  $\rhos$). Then the root of $\TPQ$ is a conical node with apex child
  $\beta_i$. If no such $i$ exists, then the situation is that of
  \Cref{prop:pqroot-analysis}~\ref{item:pqr2}, the root of $\TPQ$ is
  not conical and $\rhos$ is undefined.
\end{itemize}

\begin{algorithm}[htbp]
  \caption{Computes the mmodule tree of a Robinson space $(X,d)$ from its PQ-tree.}
  \label{algo:pqtree2mtree}
  \flushleft{$\underline{\pqtom(\TPQ(S))}$}
  \begin{algorithmic}[1]
  \Require{a Robinson space $(X,d)$ (implicit), a PQ-tree $\TPQ(S)$ for a subspace $S$ of $(X,d)$.}
  \Ensure{the mmodule tree $\TM(S)$ of $S$.}
  \Match{$\TPQ(S)$}
  \Case{$\Leaf\ x$}
    \Return $\Leaf\ x$
  \Case{$P(\beta_1,\ldots,\beta_k)$}
    \Return $\cap(\pqtom(\beta_1),\ldots,\pqtom(\beta_k))$ \Comment{\mbox{\Cref{prop:pqroot-analysis}\ref{item:pqr1}}}
  \Case{$Q(\beta_1,\ldots,\beta_k)$}
    \If {$\TPQ(S)$ is not conical}\label{line:not-conical}
      \Return $\cup(\pqtom(\beta_1),\ldots,\pqtom(\beta_k))$ \Comment{\mbox{\Cref{prop:pqroot-analysis}\ref{item:pqr2}}}
    \EndIf
    \Let $\rhos := \rhos(\TPQ(S))$ and $\beta_i$ the apex child of $\TPQ(S)$\label{line:k3}
    \Let $\mathcal{T}_0 := \cup(\pqtom(\beta_1),\ldots,\pqtom(\beta_{i-1}),\pqtom(\beta_{i+1}),\ldots,\pqtom(\beta_k))$
    \If {$\rhos(\beta_i)$ is undefined or $\rhos(\beta_i) < \rhos$}
      \Return $\cap(\mathcal{T}_0, \pqtom(\beta_i))$ special \Comment{\mbox{\Cref{prop:pqroot-analysis}\ref{item:pqr41}}}\label{line:pqr41}
    \EndIf
    \Let $\gamma_1,\ldots,\gamma_\ell$ be the children of $\beta_i$
    \Return $\cap(\mathcal{T}_0,\pqtom(\gamma_1),\ldots,\pqtom(\gamma_\ell))$ special\Comment{\mbox{\Cref{prop:pqroot-analysis}\ref{item:pqr43}}}\label{line:pqr423}
  \EndMatch  
  \end{algorithmic}
  \label{algo:pqtom}
\end{algorithm}

\begin{theorem}\label{th:pqtom}
  Let $(X,d)$ be a Robinson space and $\TPQ$ be its PQ-tree. Then
  $\pqtom(\TPQ)$ correctly computes the mmodule tree of $(X,d)$ in
  time $O(|X|)$.
\end{theorem}

\begin{proof}
  The correctness mostly follows by induction from
  \Cref{prop:pqroot-analysis}, we only explain the differences between the
  conditions in \Cref{prop:pqroot-analysis} and \Cref{algo:pqtom}.

  % On line~\ref{line:k2}, the algorithm tests whether the root is a
  % Q-node of arity 2. If it is indeed the case, then for
  % $\rhos = d(\beta_1,\beta_2)$, $\Grho$ is not connected and has two
  % connected components
  % $X(\beta_1), X(\beta_2)$. By \label{lemma:d-cut-separable}, each has
  % diameter at most $\rhos$. Hence
  % \Cref{prop:pqroot-analysis}\ref{item:pqr3} applies. Conversely, if
  % \Cref{prop:pqroot-analysis}\ref{item:pqr3} applies, then the test in
  % line~\ref{line:k2} passes and the mmodule tree is correctly
  % computed.

  On line~\ref{line:not-conical}, testing whether $\TPQ(S)$ is conical
  is equivalent to checking whether there is apex child, in which case
  $(S,d)$ is not connected and
  \Cref{prop:pqroot-analysis}\ref{item:pqr4} applies. Otherwise
  $(S,d)$ is connected and \Cref{prop:pqroot-analysis}\ref{item:pqr2}
  applies.
  
  Consider the case when the execution runs through
  line~\ref{line:k3}. The conical child $\beta_i$ is
  $X \setminus S_0$. We must determine $c_\rhos$ to decide between
  \Cref{prop:pqroot-analysis}\ref{item:pqr41} and~\ref{item:pqr43}. As
  $X(\beta_i)$ is the union of the $\rhos$-mmodules other than $S_0$
  we may use $\rhos(\beta_i)$. Indeed, $c_\rhos > 2$ if and only
  if $\Grho(X(\beta_i))$ is not connected. Hence if $c_\rhos = 2$,
  then $S_0,X(\beta_i)$ are the $\rhos$-mmodules of $S$ and the return at
  line~\ref{line:pqr41} is correct. Otherwise, by
  \Cref{lemma:some-d-cut}, the $\rhos$-mmodules other than $S_0$ are
  the sets induced by the children of $\beta_i$, proving that the
  return at line~\ref{line:pqr423} is also correct.

  The complexity follows from the  previous remark that
  $\rhos(\alpha)$ can be computed in time $O(k)$ where $k$ is the
  number of children of the node $\alpha$. Then we use the fact he the
  sum of arities of the nodes in $\TPQ$ is no more than $2|X|$ because
  each inner node has arity at least $2$ and the number of leaves is
  $|X|$, thus the total cost for computing all the $\rhos(\alpha)$ is
  $O(|X|)$. The rest of the algorithm can be computed in time
  proportional to the size of the PQ-tree, that is in $\Theta(|X|)$.
\end{proof}

\subsection{From \texorpdfstring{$\TM$ to $\TPQ$}{TM to TPQ}} 
\Cref{lemma:no-d-cut,lemma:some-d-cut,lemma:cut-separable} also allow
to derive the PQ-tree of a Robinson space from its mmodule tree,
except for the ordering of children of Q-nodes. This is done in
\Cref{algo:mtopq}. We analyse this algorithm.

\begin{algorithm}[htbp]
  \caption{Computes the PQ-tree of a Robinson space $(X,d)$ from its mmodule tree.}
  \label{algo:mtopq}
  \flushleft{$\underline{\mtopq(\TM(S))}$}
  \begin{algorithmic}[1]
  \Require{a Robinson space $(X,d)$ (implicit), an mmodule tree $\TM(S)$ for a subspace $S$ of $(X,d)$.}
  \Ensure{the PQ-tree $\TPQ(S)$ of $S$.}
  \Match{$\TM(S)$}
  \Case{$\Leaf\ x$}
    \Return $\Leaf\ x$
  \Case{$\cup(\beta_1,\ldots,\beta_k)$}
    \Let $\beta_{\sigma(1)} \< \ldots \< \beta_{\sigma(k)}$ be a compatible order of the flat Robinson space $X / \Mmax$ \label{line:pq4}
    \Return $Q(\mtopq(\beta_{\sigma(1)}),\ldots,\mtopq(\beta_{\sigma(k)}))$ \Comment{\mbox{\Cref{lemma:no-d-cut}}} \label{line:pq5}
  \Case{$\cap(\beta_1,\ldots,\beta_k)$}
    \Let $\rhos = d(\beta_1,\beta_k)$ \label{line:pq7}
    \If {$\TM(S)$ is special with large child $\beta_i$}
      \State $\beta_i \leftrightarrow \beta_k$ \Comment{\mbox{ensures $i = k$}}
      \Let $Q(\gamma_1,\ldots,\gamma_\ell) := \mtopq(\beta_k)$ \Comment{\mbox{\Cref{lemma:cut-separable}}}
      \Let $j := \findBip(\rhos,\gamma_1,\ldots,\gamma_\ell)$
      \Let $\beta := \mtopq(\beta_1)$ if $k = 2$, $\beta := P(\mtopq(\beta_1),\ldots,\mtopq(\beta_{k-1}))$ if $k \geq 3$
      \Return $Q(\gamma_1,\ldots,\gamma_j,\beta,\gamma_{j+1}, \ldots, \gamma_\ell)$ conical with apex $\beta$ \Comment{\mbox{\Cref{lemma:cut-separable}}} \label{line:pq13}
    \EndIf
    %\If{$\cap(\beta_1,\beta_2)$}
    %  \Return $Q(\mtopq(\beta_1),\mtopq(\beta_2))$ \label{line:pq14bis} \Comment{\mbox{\Cref{lemma:some-d-cut}}}
    %\EndIf
    \Return $P(\mtopq(\beta_1),\ldots,\mtopq(\beta_k))$ \Comment{\mbox{\Cref{lemma:some-d-cut}}} \label{line:pq14}
  \EndMatch
  \end{algorithmic}\medskip

  \flushleft{$\underline{\findBip(\rhos,\gamma_1,\ldots,\gamma_k)}$}
  \begin{algorithmic}[1]
    \Require{a non-connected Robinson space $(X,d)$ (implicit)
      with a large $\rhos$-mmodule $S_0$ of $\Gdeltas$
      and bipartition $(S_\star,S^\star)$, and
      the children $\gamma_1, \ldots \gamma_\ell$ of the Q-node root of
      $\TPQ(S_0)$ (by \Cref{lemma:cut-separable}).}
    \Ensure{The unique $j \in \{1,\ldots,l-1\}$ such that for
      $S_\star := X(\gamma_1) \cup \ldots \cup X(\gamma_j)$ and
      $S^\star := X(\gamma_{j+1}) \cup \ldots \cup X(\gamma_\ell)$.}
    \Let $i_0 := \max \{i \in \{1,\ldots,\ell-1\} : d(\gamma_i,\gamma_\ell) > \rhos \}$ \Comment{exists as $\diam(S_0) > \rhos$}
    \Return $\min \{i \in \{i_0,\ldots,\ell-1\} : d(\gamma_i,\gamma_{i+1}) \geq \rhos \}$
  \end{algorithmic}
\end{algorithm}

\begin{lemma}\label{lemma:findbip}
  Under the hypothesis of \Cref{lemma:cut-separable},
  $\findBip(\rhos,\gamma_1,\ldots,\gamma_k)$ from \Cref{algo:mtopq}
  correctly computes the bipartition $(S_\star,S^\star)$ of $S_0$ in
  time $O(l)$.
\end{lemma}

\begin{proof}
  We know that $j$ exists by \Cref{lemma:cut-separable}. Let $j'$ be
  the value returned by $\findBip$ from \Cref{algo:mtopq}, and let
  $S := X \setminus S_0$. Denote $S_i = X(\gamma_i)$ for any
  $i \in \{1,\ldots,\ell\}$. Let $<$ be a compatible order with
  $S_1 < \ldots < S_j < S < S_{j+1} < \ldots < S_\ell$.
  As $d(S_{i_0},S_\ell) > \rhos = d(S,S_\ell) \geq d(S_{j+1},S_\ell)$, we have
  that $i_0 < j+1$. This implies that $j'$ is well defined as
  $d(S_j,S_{j+1}) \geq d(S_j,S) = \rhos$, and $j' \leq j$.

  By way of contradiction suppose that $j' < j$. We prove that
  $T := S_{j'+1} \cup \ldots \cup S_j \cup S$ is a union of
  $\rhos$-mmodules, a contradiction to the fact that $S_0$ is a
  maximal mmodule. Notice that
  $S_{j'+1} \< \ldots \< S_{j} \< S$, hence it is sufficient to show
  that $d(S_1,S_{j'+1}) = d(S_{j'+1},S_l) = \rhos$ and
  $d(S_1,S) = d(S,S_{j+1}) = \rhos$. The latter inequalities follow
  from the definition of $S$, we prove the former. By definition of
  $j'$, $d(S_{j'},S_{j'+1}) \geq \rhos$, but
  $S_1 \< S_{j'} \< S_{j'+ 1} \cleq S_{j} \< S$ implies that
  $d(S_{j'},S_{j'+1}) \leq d(S_1,S) = \rhos$, hence
  $d(S_{j'},S_{j'+1}) = \rhos$. By the maximality of $i_0$ and
  since $j' \geq i_0$, $d(S_{j'+1},S_\ell) \leq \rhos$. Then because
  $S_{j'+1} \cleq S_j \< S < S_l$, we get
  $d(S_{j'+1},S_\ell) \geq d(S,S_\ell) = \rhos$, hence the last equality
  follows.
\end{proof}

\begin{theorem}\label{th:mtopq}
  Let $(X,d)$ be a Robinson space and $\TM$ be its mmodule tree. Then
  $\mtopq(\TM)$ correctly computes the PQ-tree $\TPQ$ of $(X,d)$ in time
  $O(|X|)$ without counting the cost of ordering the children of each
  $Q$-node.
\end{theorem}

\begin{proof}
  If $\TM = \cup(\beta_1,\ldots,\beta_k)$, then by
  \Cref{mmodule-tree}, $\Mmax$ is a partition with maximal mmodules
  $X(\beta_1),\ldots,X(\beta_k)$. By \Cref{lemma:copartition-delta},
  $\Gdelta$ is connected for any $\delta > 0$. By
  \Cref{lemma:no-d-cut}, $\TPQ = Q(T_1,\ldots,T_k)$ where
  $T_1,\ldots,T_k$ are the PQ-tree of $X(\beta_1),\ldots,X(\beta_k)$
  given in a compatible order. This proves that the return at
  line~\ref{line:pq5} is correct.

  If $\TM = \cap(\beta_1,\ldots,\beta_k)$, then by
  \Cref{mmodule-tree}, $\Mmax$ is a copartition, the
  $\rhos$-mmodules are $X(\beta_1),\ldots,X(\beta_k)$ and their
  complements are the maximal mmodules of $(X,d)$. By
  \Cref{lemma:copartition-delta},
  $d(X(\beta_i),X(\beta_j))=\rhos$. If $(X,d)$ does not contain
  large $\rhos$-mmodules, then by \Cref{lemma:some-d-cut}, the
  returns at line~\ref{line:pq14} is correct.

  If some $\rhos$-mmodule has diameter greater than $\rhos$, say
  $\diam(X(\beta_k)) > \rhos$, then by \Cref{lemma:cut-separable},
  the root of the PQ-tree of $S_0 := X(\beta_k)$ is a $Q$-node. By
  \Cref{lemma:large_block}, $S_0$ has a bipartition into two blocks
  $S_\star,S^\star$, which must match with consecutive children of the
  root of $\TPQ(S_0)$. By \Cref{lemma:findbip}, $\findBip$
  correctly computes that bipartition, and then by
  \Cref{lemma:cut-separable}, the return at line~\ref{line:pq13} is
  correct.

  The diameter of the sets induced by nodes of the PQ-tree can be
  computed simultaneaously: for a P-node the diameter is the
  distance between any two children, while for a Q-node, the
  diameter is the distance between the two extremal children. To
  compute the diameter of a set induced by an mmodule tree, we first
  compute its PQ-tree (as the algorithm will compute it eventually,
  this does not add any cost), then compute its diameter in constant
  time.

  Then the complexity of the algorithm, without counting the ordering
  at line~\ref{line:pq4} and the recursive calls, is proportional to
  the arity $k$ of the root node of $\TM$, and $\ell$ in case it returns
  at line~\ref{line:pq13}. Notice that the last case happens when
  transforming a Q-node whose subset induces a connected Robinson
  subspace into a Q-node whose subset induces a non-connected Robinson
  subspace.  Hence it happens at most once for any Q-node in the final
  PQ-tree $\TPQ(X)$. Thus, counting the recursive calls (but still not
  line~\ref{line:pq4}), the complexity of the algorithm is
  proportional to the sum of arities in $\TM$ and $\TPQ$, which is
  $O(|X|)$ as each inner node has arity at least 2.
\end{proof}

Another consequence of these translations between PQ-tree and mmodule
tree is that a node-to-node correspondence for most of the nodes of
those trees:

\begin{proposition}\label{prop:node-to-node}
  Let $(X,d)$ be a Robinson space, with PQ-tree $\TPQ$ and mmodule
  tree $\TM$. For any node $\alpha$ in $\TM$, either there exists a
  node $\alpha'$ in $\TPQ$ with $X(\alpha) = X(\alpha')$ or $\alpha$
  is the large child of a special node. For any node $\beta'$ in
  $\TPQ$, either there exists a node $\beta$ in $\TM$ with
  $X(\beta) = X(\beta')$, or $\beta'$ is the split child of a conical
  node $\alpha'$. Moreover, if $\alpha$ is a node of $\TM$ and
  $\alpha'$ a node of $\TPQ$ with $X(\alpha) = X(\alpha')$, then:
  \begin{enumerate}[label=(\roman*)]
  \item $\alpha$ is a non-special $\cap$-node if and only if $\alpha'$
    is a P-node,
  \item $\alpha$ is a $\cup$-node if and only if $\alpha'$ is a
    non-conical Q-node,
  \item $\alpha$ is a special $\cap$-node if and only if $\alpha'$ is
    a conical Q-node.
  \end{enumerate}
\end{proposition}

% \begin{proof}
%   This follows by an easy induction on the tree $\TM$ for $\alpha$ and
%   on the tree $\TPQ$ for $\beta'$, using \Cref{th:pqtom,th:mtopq}.
% \end{proof}

% Then we can also rewrite \Cref{lemma:monotone-diameter}:

% \begin{lemma}\label{lemma:monotone-diam-tm}
%   Let $(X,d)$ be a Robinson space with mmodule tree $\TM$. Let
%   $\alpha$ be an internal node of $\TM$ and $\beta$ a non-leaf child
%   of $\alpha$. Then 
%   \begin{enumerate}[label=(\roman*)]
%   \item\label{item:mono1tm} if $\alpha$ is a non-special $\cap$-node,
%     then $\diam(X(\beta)) \leq \rhos(\alpha) = \diam(X(\alpha))$;
%   \item\label{item:mono2tm} if $\alpha$ is a $\cup$-node, then
%     $\diam(X(\beta)) \leq \min \{d(\beta,\beta') : \beta' \textrm{
%       child of } \alpha, \beta' \neq \beta\} \leq \diam(X(\alpha))$;
%   \item if $\diam(X(\alpha)) = \diam(X(\beta))$, then $\alpha$ is a ...
%   \end{enumerate}
% \end{lemma}

%\begin{example}[Ultrametrics]

\subsection{Mmodules trees of ultrametrics} 
As an application of the results of this section, we show that for
ultrametrics $\TM$ is isomorphic to $\TD$ and $\TPQ$. We characterize
the ultrametric spaces via their mmodule trees in the following way:

  \begin{proposition}\label{prop:ultrametric_mmodule_tree} If $(X,d)$ is an ultrametric space, then the X-trees $\TD$, $\TPQ$, and $\TM$
    are isomorphic. Furthermore, a dissimilarity space $(X,d)$ is ultrametric if and only if its
    mmodule-tree $\TM$ satisfies the following conditions:
    \begin{enumerate}[label=(\roman*)]
    \item\label{item:umt1} the internal nodes of $\TM$ are all
      $\cap$-nodes,
    \item\label{item:umt2} for every node $\alpha$ of $\TM$ and child
      $\beta$ of $\alpha$, we have
      $\diam(X(\beta)) < \diam(X(\alpha))$.
    \end{enumerate}
  \end{proposition}

\begin{proof}
  Suppose first that $(X,d)$ is an ultrametric space. Then, by
  \Cref{prop:ultrametric_pq_tree}, $\TPQ$ contains only P-nodes. By \Cref{prop:pqroot-analysis}\ref{item:pqr1}, all internal nodes of $\TM$ are non-special
  $\cap$-nodes with no large component and are in one-to-one
  correspondance to those of $\TPQ$, proving~\ref{item:umt1}
  and~\ref{item:umt2}, and that $\TPQ$ and $\TM$ are isomorphic. By
  \Cref{prop:dendrogram_pq_tree}, they are also isomorphic to $\TD$.

  Conversely, suppose that $\TM$ satisfies~\ref{item:umt1}
  and~\ref{item:umt2}. Let $x,y,z \in X$ and $\alpha$ be the common
  least ancestor of $x, y ,z$ in $\TM$. If $x,y,z$ are in distinct
  children, then $d(x,y) = d(y,z) = d(x,z) = \rhos(\alpha)$
  by~\ref{item:umt1}. Otherwise, two of them are in a common child,
  say $x$ and $y$. Then $d(x,y) < \rhos(\alpha)$ by~\ref{item:umt2},
  and $d(x,z) = d(y,z) = \rhos(\alpha)$ by \ref{item:umt1}. Thus
  $(X,d)$ is an ultrametric.
\end{proof}

%\end{example}

%% file: sec_6_mmodtree.tex
In \Cref{SECTION_translations} we established a correspondence between
the PQ-tree $\TPQ$ and the mmodule tree $\TM$ of a Robinson space
$(X,d)$, allowing to derive one such tree from another. In
\Cref{SECTION_Gdelta_PQ} we showed how to construct $\TPQ$ from the
maximal mmodules $\Mmax$ of $(X,d)$. In this section, we show how to
build the mmodule tree $\TM$ recursively from top-to-bottom, using a
partition refinement algorithm and the dendrogram $\Td$ of the
ultrametric subdominant $(X,\widehat{d})$. This also allows to
construct $\Mmax$ and thus to get full algorithmic translations
between the trees $\TPQ$ and $\TM$.

\subsection{Stable partitions and partition refinement}

A \emph{partition} of a set $X$ is a family of sets
$\PP=\{ B_1,\ldots,B_m\}$ such that $B_i\cap B_j=\varnothing$ for any
$i\ne j$ and $\bigcup_{i=1}^k B_i=X$. The sets $B_1,\ldots, B_m$ are
called the \emph{classes} of $\PP$.

\begin{definition}[Stable partition]\label{def:stable-partition}
  A partition $\PP=\{ B_1,\ldots,B_m\}$ of a dissimilarity space
  $(X,d)$ is a \emph{stable partition} if for any
  $i \in \{1,\ldots,m\}$, $B_i$ is an mmodule.
\end{definition}

A non-stable partition $\PP$ can be transformed into a stable
partition by applying the classical operation of \emph{partition
  refinement}, which proceeds as follows. \Cref{algo:stable} (see the
Appendix) maintains the current partition $\PP$ and for each class $B$
of $\PP$ maintains the set $Z(B)$ of all points outside $B$ which
still have to be processed to refine $B$. While $\PP$ contains a class
$B$ with nonempty $Z(B)$, the algorithm pick any point $z$ of $Z(B)$
and partition $B$ into maximal classes that are not distinguishable
from $z$, {\it i.e.} for any such new class $B'\subseteq B$ and any
$x,x' \in B'$ we have $d(x,z) = d(x',z)$. Finally, the algorithm
removes $B$ from $\PP$ and insert each new class $B'$ in $\PP$ and
sets $Z(B'):= (B \setminus B')\cup (Z(B) \setminus \{z\})$. Notice
that each class $B$ is partitioned into subclasses by comparing the
distances of points of $B$ to the point $z\notin B$ and such distance
items never occur in later comparisons. Also, if the final stable
partition has classes $B'_1,\ldots, B'_t$, then the distances between
points in the same class $B'_i$ are never compared to other distances.
This algorithm is formalized in \Cref{algo:stable}, where one would
call $\stablePart(\PP)$ to get a stable partition from $\PP$. The
complexity of $\partRefine$ in \Cref{algo:stable} is
$O(\sum_{i = 1}^k |B_i| \times |B \setminus B_i|)$, the complexity of
refining $\PP$ into $\PP'$ with \Cref{algo:stable} is
$O(\sum_{P \in \PP'} |P| \times |X \setminus P|)$. The copoint
partition $\C_p$ of $(X,d)$ can be constructed by applying
\Cref{algo:stable} to the partition $\{ \{ p\}, X\setminus\{ p\})$.

The effect of \Cref{algo:stable} applied to a partition of a subset
$S \subseteq X$ of a dissimilarity space $(X,d)$ is described in
\Cref{lemma:stable-correct}, which asserts that $\stablePart(\PP)$ is
the less refined refinement of $\PP$ into mmodules of $(S,d)$.

We now present \Cref{algo:dendro-refine}, a variant of the stable
partition algorithm, that uses $S$-trees to represent sets. This will
allow us to use the dendrogram $\Td$ of the ultrametric subdominant to
represent $X$, providing us with extra information that will be needed to
obtain an efficient algorithm to build the mmodule tree.

\begin{algorithm}[htbp]
  \caption{A refinement algorithm where sets are represented by $X$-trees}
  \label{algo:dendro-refine}

  \flushleft{$\underline{\stableDendro(\Tt(S_1),\ldots,\Tt(S_k))}$}
  \begin{algorithmic}[1]
    \Require{a dissimilarity space $(X,d)$ (implicit), $\Tt(S_i)$ an $S_i$-tree for $i \in \{1,\ldots,k\}$, where
      $\{S_1,\ldots,S_k\}$ is a partition of a subset $S \subseteq X$.}
    \Ensure{$\Tt(M_i)$ an $M_i$-tree for $i \in \{1,\ldots,\ell\}$, where $\{M_1,\ldots,M_\ell\}$ is
      a partition of $S$ into mmodules of $(S,d)$}.
    \For{$i \in \{1,\ldots,k\}$}
      \YieldAll $\refineDendro(\Tt(S_i), S_1 \cup \ldots \cup S_{i-1} \cup S_{i+1} \cup \ldots \cup S_k)$\label{line:sdcalls}
    \EndFor
  \end{algorithmic}\medskip

  \flushleft{$\underline{\refineDendro(\Tt(S),Z(S))}$}
  \begin{algorithmic}[1]\setcounter{ALG@line}{1}
    \If{$Z(S) = \varnothing$}
      \Return $\Tt(S)$\label{line:rdreturn}
    \EndIf
    \Let $q \in Z(S)$
    \Let $\Tt(S_1),\ldots,\Tt(S_k) := \drefine(q, \Tt(S))$\label{line:calldref}
    \For{$i \in \{1,\ldots,k\}$}
      \YieldAll $\refineDendro(\Tt(S_i), S_1 \cup \ldots \cup S_{i-1} \cup S_{i+1} \cup \ldots \cup S_k \cup Z(S) \setminus \{q\})$\label{line:rdrecall}
    \EndFor
  \end{algorithmic}\medskip

  \flushleft{$\underline{\drefine(q, \Tt(S))}$}
  \begin{algorithmic}[1]\setcounter{ALG@line}{8}
    \Require{a dissimilarity space $(X,d)$ (implicit), an $S$-tree $\Tt(S)$ for some
      subset $S \subseteq X$, and $q \in X \setminus S$.}
    \Ensure{$\Tt(S_i)$ an $S_i$-tree for $i \in \{1,\ldots,n\}$, where
      $S_1,\ldots,S_n$ is a partition of $S$.}
    \If{$d(q,x)$ is constant for $x \in \Tt(S)$}\label{line:rd-uniform}
      \Return $\Tt(S)$\label{line:rd-earlyexit}
    \EndIf
    \Let $\Node(\beta_1,\ldots,\beta_k) := \Tt(S)$
    \Let $I := \{i \in \{1,\ldots,k\} : \textrm{$d(q,x)$ is constant for $x \in X(\beta_i)$}\}$ and $J := \{1,\ldots,k\} \setminus I$\label{line:rd-defI}
    \For {$j \in J$}
        \YieldAll $\drefine(q,\beta_j)$ \label{line:rd-recurse}
    \EndFor
    \Let $\{d_1,\ldots,d_\ell\} := \{ d(q,\beta_i) : i \in I\}$\label{line:rd-sort}
    \For {$j \in \{1,\ldots,\ell\}$}
      \Yield $\join(\{ \beta_i : i \in I, d(q,\beta_i) = d_j \})$\label{line:rd-join}
    \EndFor
  \end{algorithmic}\medskip

  \flushleft{$\underline{\join(\{\alpha_1,\ldots,\alpha_k\})}$}
  \begin{algorithmic}[1]\setcounter{ALG@line}{19}
    \Return $\left\{\begin{array}{ll}
                      \alpha_1 &\quad\textrm{ if $k = 1$,}\\
                      \Node(\alpha_1,\ldots,\alpha_k) &\quad\textrm{ otherwise.}
                    \end{array}\right.$
  \end{algorithmic}
\end{algorithm}

\begin{lemma}\label{lemma:struct-preserve}
  Let $(X,d)$ be a dissimilarity space and let $\Tt(S)$ be an $S$-tree
  for some subset $S \subseteq X$. Let $q \in X \setminus S$ and
  $\Tt(S_1),\ldots,\Tt(S_n) := \drefine(q,\Tt(S))$. Then for all
  $i \in \{1,\ldots, n\}$, $\Tt(S_i)$ is either a subtree of $\Tt(S)$
  or the join of some children of a node in $\Tt(S)$.
\end{lemma}

\begin{proof}
  We proceed by induction on the size of $\Tt(S)$. If $d(q,x)$ is constant for
  $x \in S$, then $n = 1$ and $\Tt(S_1) = \Tt(S)$
  (line~\ref{line:rd-uniform}). Let $i \in \{1,\ldots,n\}$. If
  $\Tt(S_i)$ is yielded at line~\ref{line:rd-recurse}, then the result
  holds by induction. Otherwise, it is yielded at
  line~\ref{line:rd-join}, hence $\Tt(S_i)$ is the join of some
  children of the root.
\end{proof}

\begin{lemma}\label{lemma:drefine-uniform}
  Let $(X,d)$ be a dissimilarity space and let $\Tt(S)$ be an $S$-tree
  for some subset $S \subseteq X$. Let $q \in X \setminus S$ and
  $\Tt(S_1),\ldots,\Tt(S_n) := \drefine(q,\Tt(S))$. Then for all
  $i \in \{1,\ldots, n\}$, $d(q,x)$ is constant for $x \in \Tt(S_i)$.
\end{lemma}

\begin{proof}
This easily follows by induction on the size of $\Tt(S)$.
\end{proof}

The fundamental property of the stable partition algorithm is that it
outputs the maximal mmodules included in some parts of the initial partition.
This property will only be preserved when working with $S$-trees whose
structure respects the mmodules of the dissimilarity space.

\begin{definition}
  Let $(X,d)$ be a dissimilarity space and $S \subseteq X$. An
  $S$-tree $\Tt$ is \emph{coherent} if for each mmodule $M$ of $X$
  with $M \subseteq S$, for each child $\beta$ of the $M$-pertinent
  node in $\Tt$, either $X(\beta) \subset M$ or
  $X(\beta) \cap M = \varnothing$.
\end{definition}

Notice that any subtree of a coherent tree is itself coherent.

\begin{lemma}\label{lemma:no-mmod-cut}
  Let $(X,d)$ be a dissimilarity space and let $\Tt(S)$ be a coherent
  $S$-tree for some $S \subseteq X$. Let $q \in X \setminus S$ and
  $\Tt(S_1),\ldots,\Tt(S_n) := \drefine(q,\Tt(S))$. Then
  $\Tt(S_1),\ldots,\Tt(S_n)$ are coherent. Moreover, let
  $M \subsetneq S$ be an mmodule of $S$, then there is
  $i \in \{1,\ldots,n\}$ such that $M \subseteq S_i$.
\end{lemma}

\begin{proof}
  Let $M$ be a proper mmodule of $S$. By induction on the size of
  $\Tt(S)$, first we prove the claim that there exist
  $i \in \{1,\ldots,n\}$ and a node $\alpha$ of $\Tt(S_i)$ such that
  $M = X(\alpha)$ or $M = X(\beta_1) \cup \ldots X(\beta_k)$ for some
  children $\beta_1,\ldots,\beta_k$ of $\alpha$. Then we obtain the
  assertion of the lemma by applying this claim to all proper mmodules
  of $S$.

  If $d(q,x)$ is constant for $x \in S$, then $n = 1$,
  $\Tt(S_1) = \Tt(S)$ and the claim follows from the coherence of
  $\Tt(S)$. If the root of $\Tt(S)$ is not $M$-pertinent, then as
  $\Tt(S)$ is coherent, $M \subseteq X(\beta_i)$ for some child
  $\beta_i$ of the root. If $i \in I$, $\beta_i$ is in the join of
  some $\Tt(S_j)$ yielded at line~\ref{line:rd-join} and the claim
  holds as $\beta_i$ is coherent, else the claim holds by induction,
  from line~\ref{line:rd-recurse} because $\beta_i$ is coherent.

  Otherwise the root is $M$-pertinent. Since $M$ is an mmodule and
  $\Tt(S)$ is coherent, there is a subset $I' \subseteq I$ such that
  $M = \bigcup_{i \in I'} X(\beta_i)$, and there is
  $j \in \{1,\ldots,\ell\}$ such that $d_j = d(q,M)$. Then $M$ is a
  subset of leaves of the join yielded on line~\ref{line:rd-join}
  at iteration $j$, establishing the claim.
\end{proof}

The next proposition establishes that $\stableDendro$ is semantically
equivalent to $\stablePart$.

\begin{proposition}\label{lemma:correction-stableDendro}
  Let $(X,d)$ be a dissimilarity space and
  $\Tt(S_1),\ldots,\Tt(S_k)$ be coherent trees, where $S_1 , \ldots, S_k$
  is a partition of $X$. Let
  $\Tt(R_1),\ldots,\Tt(R_\ell) := \stableDendro(\Tt(S_1),\ldots,\Tt(S_k))$.
  Then $R_1,\ldots,R_\ell$ are the maximal by inclusion mmodules of
  $X$ contained in $S_1,\ldots,S_k$.
\end{proposition}

\begin{proof}
  By \Cref{lemma:no-mmod-cut}, for each such maximal mmodule $M$,
  there is $i \in \{1,\ldots,\ell\}$ such that $M \subseteq R_i$. It
  remains to prove that, for $i \in \{1,\ldots,\ell\}$, $R_i$ is an
  mmodule.

  We check the following invariant: for each call
  $\refineDendro(T(S),Z(S))$, for each $u \notin S \cup Z(S)$,
  $d(x,u)$ is constant for $x \in S$. This is trivial
  in line~\ref{line:sdcalls}. Then consider a call
  $\refineDendro(T(S),Z(S))$ for which the invariant holds, and let us
  prove it for the recursive calls happening at
  line~\ref{line:rdrecall}. For iteration $i$, let
  $u \in X \setminus (Z(S_i) \cup S_i) = \{q\} \cup (X \setminus Z(S) \setminus S)$.
  If $u \in X \setminus Z(S) \setminus S$, then $d(x,u)$ is constant
  for $x \in S$ hence for $x \in S_i$. Otherwise for $u = q$, $S_i$ is
  a set yielded by $\drefine(q,\Tt(S))$, and the invariant follows by
  \Cref{lemma:drefine-uniform}.

  Consequently, for any $i \in \{1,\ldots,\ell\}$, $\Tt(R_i)$ is returned at
  line~\ref{line:rdreturn}, for some call
  $\refineDendro(\Tt(R_i),\varnothing)$. Thus by the invariant,
  $\Tt(R_i)$ is an mmodule.
\end{proof}

The complexity of \Cref{algo:stable,algo:dendro-refine} are
asymptotically equivalent:

\begin{lemma}\label{lemma:complexity-stableDendro}
  Let $(X,d)$ be a dissimilarity space and $\Tt(S_1),\ldots,\Tt(S_k)$
  be coherent trees, where $S_1 , \ldots, S_k$ is a partition of $X$.
  Then
  $\Tt(R_1),\ldots,\Tt(R_\ell) := \stableDendro(\Tt(S_1),\ldots,\Tt(S_k))$
  takes $O(\sum_{i=1}^\ell |R_i| |X \setminus R_i|)$ time.
\end{lemma}

\begin{proof}
  Notice that $\refineDendro(\Td(S), Z(S))$ is called at most
  $|X \setminus S|$ times, because at each call, some element $q$ is
  chosen and removed from $Z(S)$. For a given $q \in X \setminus S$,
  we must evaluate whether $d(q,x)$ is constant for $x \in S'$. In
  time $O(|S|)$, we can decide that property for each node in
  $\Td(S)$. We charge a cost of $O(1)$ on each pair $(q,x)$ with
  $x \in S$. This allows to solve lines~\ref{line:rd-uniform}
  and~\ref{line:rd-defI} in all calls to $\drefine(q,\Td(S))$ for
  all $S$ with $q \notin S$ in total time $O(|X \setminus R_i|)$, where
  $q \in R_i$. Summing over all $q \in X$, we get
  $O(\sum_{i=1}^{n} |R_i| |X \setminus R_i|)$.

  It remains to evaluate the cost of line~\ref{line:rd-sort}. Using a
  binary search tree, associating to each distance the children at that
  distance, this has  cost $O(|I| \log \ell)$, which we can amortized
  by charging a cost of $\log \ell$ to an element of each child
  $\beta_i$, for $i \in I$. Notice that this element is, from this
  step, split from each of the other $\ell-1$ parts, hence from at least
  $\log \ell$ elements. This implies that the total charge accumulated by
  an element during the main call to $\stableDendro$ is less than the
  number of trees returned. Thus the total cost of
  line~\ref{line:rd-sort} is at most
  $O(\sum_{i=1}^n |R_i| |X \setminus R_i|)$.
\end{proof}

\subsection{\texorpdfstring{$\rhos$}{rho}-Components and the maximal mmodules}
In this subsection, we study the relationships between the maximal
mmodules and the $\rhos$-components. They will lead us to an efficient
algorithm to find the maximal mmodules of a Robinson dissimilarity and
to build its mmodule tree. First, we define the $\rho$-components of a
dissimilarity space (recall that $\rhos$ is the minimum value for
which the graph $G_{\leq \rhos}$ is connected):

\begin{definition}[$\rhos$-components]
   The \emph{$\rhos$-components} of a dissimilarity space
  $(X,d)$ are the connected components of the graph $G_{< \rhos}$.
\end{definition}

\begin{lemma}\label{lemma:giant-or-dwarf}
  Let $(X,d)$ be a connected Robinson space with
  $\TM(X) = \cup(\beta_1,\ldots,\beta_k)$. Let $C$ be a
  $\rhos$-component of $(X,d)$. Then
  \begin{enumerate}[label=(\roman*)]
  \item\label{item:cc1} either there exists $I \subseteq \{1,\ldots,k\}$
    such that $C = \bigcup_{i \in I} X(\beta_i)$,
  \item\label{item:cc2} or there exists $i \in \{1,\ldots,k\}$ such that
    $\diam(X(\beta_i)) = \rhos$,
    $\beta_i = \cap(\gamma_1,\ldots,\gamma_\ell)$, and there is
    $j \in \{1,\ldots,\ell\}$ such that $C = X(\gamma_j)$.
  \end{enumerate}
\end{lemma}

\begin{proof}
  Suppose that~\ref{item:cc1} does not happen, i.e., there is some
  $i \in \{1,\ldots,k\}$ such that
  $C \cap X(\beta_i) \neq \varnothing$ and
  $X(\beta_i) \setminus C \neq \varnothing$. By
  \Cref{lemma:monotone-diameter,prop:node-to-node},
  $\diam(X(\beta_i)) \leq \min \{d(\beta_i,\beta_j) : j \in \{1,\ldots,k\} \setminus \{i\}\} \leq \rhos$,
  and thus for all $x \in C \cap X(\beta_i)$ and
  $y \in X(\beta_i) \setminus C$, we have $d(x,y) = \rho$. This
  implies that the Robinson space $(X(\beta_i),d)$ is non-connected,
  with $\rhos(X(\beta_i)) = \rhos = \diam(X(\beta_i))$, hence
  $\beta_i = \cap(\gamma_1,\ldots,\gamma_\ell)$ is non-special. Then
  for all $x \in X(\beta_i)$ and $z \in X \setminus X(\beta_i)$, we
  have $d(x,z) \geq \rhos$. Thus each $X(\gamma_j)$ is a
  $\rhos$-component, establishing  \ref{item:cc2}.
\end{proof}

In case~\ref{item:cc1}, the $\rhos$-component $C$ is said to be
\emph{giant}, while in case~\ref{item:cc2}, it is said to be
\emph{dwarf}. Making that distinction, the next theorem gives a way to
find most of maximal mmodules of a connected Robinson space using the
stable partition algorithm.

\begin{proposition}\label{thm:find-mmax}
  Let $(X,d)$ be a connected Robinson space with
  $\TM = \cup(\beta_1,\ldots,\beta_k)$. Let $i \in \{1\ldots,k\}$ and
  $Y$ be a $\rhos$-component intersecting $X(\beta_i)$. Let $\PP$ be
  the stable partition obtained by calling
  $\stablePart(\{Y,X \setminus Y\})$.
  \begin{enumerate}[label=(\arabic*)]
  \item\label{item:giant} If $Y$ is a giant $\rho$-component, then
    $\PP = \Mmax$.
  \item\label{item:dwarf} If $Y$ is a dwarf $\rho$-component with
    $Y = X(\gamma)$ for some child $\gamma$ of $\beta_i$, then
    $\PP = \Mmax \setminus \{X(\beta_i)\} \cup \{Y, X(\beta_i) \setminus Y\}$.
  \end{enumerate}
\end{proposition}

\begin{proof}
  By \Cref{mmodule-tree}, the maximal mmodules $\Mmax$ of
  $(X,d)$ are $\{X(\beta_1),\ldots,X(\beta_k)\}$.

  In case~\ref{item:giant}, by \Cref{lemma:giant-or-dwarf},
  $Y = \bigcup_{i \in I} X(\beta_i)$ for some
  $I \subset \{1,\ldots,k\}$, hence $Y$ and $X \setminus Y$ are unions
  of maximal mmodules.  Therefore the maximal mmodules contained in
  $Y$ and $X \setminus Y$ are exactly the maximal mmodules of $X$,
  thus by \Cref{lemma:stable-correct},  $\PP = \Mmax$ holds.

  In case~\ref{item:dwarf}, for each $\beta_j$ with
  $j \in \{1,\ldots,k\} \setminus \{i\}$,
  $X(\beta_j) \subseteq X \setminus Y$ holds, hence
  $X(\beta_j) \in \PP$. Moreover, since $(X(\beta_i),d)$ is not
  connected, the children of $\beta_i$ are its $\rho$-mmodules, thus
  $Y$ and $X(\beta_i) \setminus Y$ are mmodules in $X(\beta_i)$. By
  \Cref{mmodules}\ref{item:mmod2}, $Y$ and $X(\beta_i) \setminus Y$
  are also mmodules in $(X,d)$. Hence by \Cref{lemma:stable-correct},
  $\PP := \Mmax \setminus \{X(\beta_i)\} \cup \{Y, X(\beta_j) \setminus Y\}$.
\end{proof}

In the case of dwarf $\rhos$-components, we still need to retrieve the
maximal mmodule $X(\beta_i)$:

\begin{lemma}\label{lemma:dwarf-comp}
  Let $(X,d)$ be a connected Robinson space and
  $\{X_1,\ldots,X_k\} := \stablePart(\{Y,X \setminus Y\})$. Then a
  $\rhos$-component $Y$ of $(X,d)$ is dwarf if and only if there are
  distinct and uniquely defined indices $i,j \in \{1,\ldots,k\}$ such
  that $Y = X_i$, $d(Y,X_j) = \rho$ and for all
  $h \in \{1,\ldots,k\} \setminus \{j\}$, we have
  $d(Y,X_h) = d(X_j,X_h)$. In that case, $Y \cup X_j$ is a maximal
  mmodule of $(X,d)$.
\end{lemma}

\begin{proof}
  Let $Y$ be a dwarf $\rho$-component. From \Cref{thm:find-mmax}, and
  using the same notations, $Y = X(\gamma)$ and
  $X_j = X(\gamma) \setminus Y$ for some $j$. Then $j$ has the
  required properties, since $X(\gamma)$ is an mmodule. Conversely,
  for any index $j$ with those properties, $Y \cup X_j$ is an mmodule.
  As $X(\gamma)$ is the maximum mmodule containing $Y$, for any
  $j' \neq j$, $Y \cup X_{j'}$ is not an mmodule, proving that $j$ is
  unique. If $Y$ is a giant $\rho$-component, then either $Y$ is the
  union of at least 2 maximal mmodules, in which case $i$ does not
  exist, or $Y$ is a maximal mmodule, and if $j$ exists, then
  $Y \cup X_j$ is an mmodule, a contradiction.
\end{proof}

Now, we consider the case of non-connected spaces.

\begin{lemma}\label{lemma:mmod-component}
  Let $(X,d)$ be a non-connected Robinson space. Then any
  $\rhos$-mmodule $S$ of diameter at most $\rhos$  is also a
  $\rhos$-component.
\end{lemma}

\begin{proof}
  Clearly $S$ is the union of $\rhos$-components, because for any $x \in S$
  and $y \in X \setminus S$, $d(x,y) = \rhos$ by definition of
  $\rhos$-mmodules. Suppose there is a $\rhos$-component $C$ with
  $C \subsetneq S$. By definition of $\rhos$-components, for each
  $x \in C$, $y \in S \setminus C$, we have $d(x,y) \geq \rhos$.
  Since $\diam(S) \leq \rhos$, we get $d(x,y) = \rhos$, whence $C$ is
  a $\rho$-component, a contradiction.
\end{proof}

\begin{proposition}\label{thm:find-mmax-nc}
  Let $(X,d)$ be a non-connected Robinson space and let
  $C_1,\ldots,C_k$ be its $\rhos$-components. Let
  $I := \{ i \in \{1,\ldots,k\} : C_i \textrm{ is not a
    $\rhos$-mmodule}\}$. Then exactly one of the following assertions
  holds:
  \begin{enumerate}[label=(\arabic*)]
  \item\label{item:ncmmax2} $|I| = 0$; in this case, $\Mmax =
    \{\co{C_1},\ldots,\co{C_k}\}$ and there are no large
    $\rhos$-mmodules.
  \item\label{item:ncmmax1} $|I| \geq 2$; in this case,
    $S_0 := \bigcup_{i \in I} C_i$ is a large $\rhos$-mmodule, and
    $\Mmax = \{\co{C_j} : j \in \{1,\ldots,k\} \setminus I\} \cup \{\co{S_0}\}$.
    Moreover, the graph $H:=(I, \{ii': \textrm{ there are } x \in C_i,
      y \in C_{i'} \textrm{ such that } d(x,y) > \rhos\})$ is connected and
    bipartite, with bipartition $I = I_\star \cup I^\star$, and the
    maximal mmodules of $S_0$ are given by
    $\stablePart(\{S_\star,S^\star\})$, where
    $S_\star := \bigcup_{i \in I_\star} C_i$ and
    $S^\star := \bigcup_{i \in I^\star} C_i$.
  \end{enumerate}
\end{proposition}

\begin{proof}
  Since $(X,d)$ is non-connected, by \Cref{lemma:delta-mmodules}, each
  $\rhos$-mmodule is an mmodule, and at most one $\rhos$-mmodule $S_0$
  has diameter greater than $\rhos$. All the $\rhos$-mmodules except
  $S_0$ are $\rhos$-components by \Cref{lemma:mmod-component}. If
  there is no such $S_0$, then this immediately leads to
  case~\ref{item:ncmmax2}.

  We may now assume that $S_0$ exists. Then by
  \Cref{lemma:large_block}, there is a bipartition
  $S_0 = S_\star \cup S^\star$. By
  \Cref{lemma:large_block}\ref{item:s03}, $S_\star$ and $S^\star$ are
  unions of $\rhos$-components, and $S_0 = \bigcup_{i \in I} C_i$. By
  \Cref{lemma:copartition-delta},
  $\Mmax = \{\co{C_j} : j \in \{1,\ldots,k\} \setminus I\} \cup \{\co{S_0}\}$.
  Furthermore, the graph $H$ is connected, because any connected
  component is an $\rhos$-mmodule by definition and $S_0$ does not
  contain a proper $\rhos$-mmodule. By
  \Cref{lemma:large_block}\ref{item:s02}, $I_\star \cup I^\star$ with
  $I_\star := \{i \in I : C_i \subseteq S_\star\}$ and
  $I^\star := \{i \in I : C_i \subseteq S^\star\}$ is the (unique)
  bipartition of $H$.

  Suppose that there is a non-trivial mmodule $M$ of $S_0$ with
  $x \in M \cap S_\star \neq \varnothing$ and
  $y \in M \cap S^\star \neq \varnothing$. Let
  $z \in S_0 \setminus M$, say $z \in S_\star \setminus M$. Then
  $\rhos \leq d(z,y) = d(z,x) \leq \rhos$ by
  \Cref{lemma:large_block}\ref{item:s02} and~\ref{item:s03} and the
  fact that $M$ is an mmodule. Also for $w \in X \setminus S_0$, we have
  $d(w,x) = d(w,y) = \rhos$. Hence $M$ is a $\rhos$-mmodule of $X$, a
  contradiction. Thus, for any maximal mmodule $M$ of $S_0$, either
  $M \subseteq S_\star$ or $M \subseteq S^\star$, hence by
  \Cref{lemma:stable-correct}, the maximum mmodules of $S_0$ are provided
  by $\stablePart(\{S_\star,S^\star\}))$.
\end{proof}

%\subsection{Computation of mmodule tree and maximal mmodules}

\Cref{thm:find-mmax,thm:find-mmax-nc} can be used to compute the
mmodule tree of a Robinson dissimilarity in optimal $O(|X|^2)$ time in
a top-down way, that is root first then recursively on each child. The
alternative way to compute the mmodule tree of an arbitrary
dissimilarity space is to compute the copoints of a point and to sort
them, leading to a root-to-leaf path. Then recurse on all subtrees
attached to that path. This gives an optimal $O(|X|^2)$-time
algorithm~\cite{EhGaMcCSu}. Our \Cref{algo:mmax} is more complicated,
 %and thus is unlikely to be practical. N
 nevertheless, we think that it
shed complimentary light on the links between mmodule trees and the
dendrogram of the ultrametric subdominant.

\begin{algorithm}[htpb]
  \caption{Computes the mmodule tree of $(X,d)$}
  \label{algo:mmax}

  \flushleft{$\underline{\mmodTree(S)}$}
  \begin{algorithmic}[1]
    \Require{a Robinson space $(X,d)$, an mmodule $S \subseteq X$.}
    \Ensure{The mmodule tree of $S$}
    \If {$|S| = 1$}
      \Return $\Leaf\ x$, where $\{x\} = S$ \label{line:mmtleaf}
    \EndIf
    \Let $C_1,\ldots,C_k$ be the $\rhos$-component of $S$ \Comment{assume $|C_1| \leq |C_i|$ for any $i \in \{1,\ldots,k\}$}\label{line:def-Ci}
    \Let $J := \{ i \in \{1,\ldots,k\} : C_i \textrm{ is a $\rhos$-mmodule}\}$ and $I := \{1,\ldots,k\} \setminus I$\label{line:def-IJ}
    \If{$I = \varnothing$} \Comment{\mbox{\Cref{thm:find-mmax-nc}\ref{item:ncmmax2}}}
      \Return $\cap(\mmodTree(C_1),\ldots,\mmodTree(C_k))$ \label{line:mmtcap}
    \EndIf
    \If{$J = \varnothing$}
      \Let $S_1,\ldots,S_\ell := \stablePart(\{C_1, S \setminus C_1\})$ with $S_1 \subseteq C_1$\label{line:apply-stable-1}
      \For{$i \in \{1,\ldots,\ell\}$}\label{line:decide-dwarf-start}
        \If{for all $j \in \{2,\ldots,\ell\} \setminus \{i\}$, $d(S_1,S_j) = d(S_i,S_j)$}
          \Let $\{M_1,\ldots,M_{\ell-2}\} = \{\mmodTree(S_j) : j \in \{2,\ldots,\ell\} \setminus \{i\}\}$ \label{line:mmod-hard}
          \Return $\cup(\mmodTree(S_1 \cup S_i), M_1,\ldots,M_{\ell-2})$\label{line:mmtcupdwarf}\Comment{\mbox{\Cref{thm:find-mmax}\ref{item:dwarf}}}
        \EndIf
      \EndFor
      \Return $\cup(\mmodTree(S_1),\ldots,\mmodTree(S_\ell))$ \label{line:mmtcupgiant}\Comment{\mbox{\Cref{thm:find-mmax}\ref{item:giant}}}
    \EndIf
    \Let $\{\beta_1,\ldots,\beta_{k'}\} := \{ \mmodTree(C_j) : j \in J\}$\label{line:mmt-def-beta}
    \Let $H = (I, \{ij : \textrm{ there is $x \in C_i$, $y \in C_j$, with $d(x,y) > \rhos$}\})$
    \Let $I = I_\star \cup I^\star$ be the unique bipartition of $H$, $S_\star := \bigcup_{i \in I_\star} C_i$ and $S^\star := \bigcup_{i \in I^\star} C_i$
    \Let $S_1,\ldots,S_\ell := \stablePart(\{S_\star,S^\star\})$\label{line:apply-stable-2}
    \If{$\ell=2$}\Comment{\Cref{thm:find-mmax-nc}\ref{item:ncmmax1}, with $|I| = 2$ and $d(x,y)$ constant for $x \in S_\star$, $y \in S^\star$}
      \Return $\cap(\beta_1,\ldots,\beta_{k'},\cap(\mmodTree(S_\star),\mmodTree(S^\star))$\label{line:mmtsuperspecial}
    \EndIf
    \Return $\cap(\beta_1,\ldots,\beta_{k'},\cup(\mmodTree(S_1),\ldots,\mmodTree(S_\ell)))$\label{line:mmtspecial}\Comment{\mbox{\Cref{thm:find-mmax-nc}\ref{item:ncmmax1}}}
  \end{algorithmic}
\end{algorithm}

\begin{proposition}
  Let $(X,d)$ be a Robinson space, $\mmodTree(X)$ correctly computes
  the mmodule tree of $X$.
\end{proposition}

\begin{proof}
  Obviously the return at line~\ref{line:mmtleaf} is correct. The
  space $(S,d)$ is not connected precisely when there is no
  $\rhos$-mmodule, that is when $J = \varnothing$. Thus, the return at
  line~\ref{line:mmtcap} is a direct consequence
  of \Cref{thm:find-mmax-nc}\ref{item:ncmmax2} as $J$ is non-empty.

  If $J = \varnothing$, then $(S,d)$ is connected
  and \Cref{thm:find-mmax} applies. Deciding whether the
  $\rhos$-component $C_1$ is dwarf or giant is done in
  lines~\ref{line:decide-dwarf-start} to~\ref{line:mmtcupdwarf},
  following \Cref{lemma:dwarf-comp}. Consequently, the returns at
  line~\ref{line:mmtcupdwarf} and~\ref{line:mmtcupgiant} are correct
  by \Cref{thm:find-mmax}.

  Otherwise, $J \neq \varnothing$ implies that $(S,d)$ is not connected
  and \Cref{thm:find-mmax-nc} applies. As $I \neq \varnothing$ we are
  in case~\ref{item:ncmmax2} of \Cref{thm:find-mmax-nc}. If $\ell =
  2$, notice that we must have $\{S_1,S_2\} = \{S_\star, S^\star\}$:
  $S_0 = S_\star \cup S^\star$ has only two maximal mmodules and thus
  the root of its mmodule tree is a $\cap$-node. This justifies the
  return at line~\ref{line:mmtsuperspecial}. If $\ell > 2$, we
  follow \Cref{thm:find-mmax-nc}\ref{item:ncmmax1}, proving that
  line~\ref{line:mmtspecial} is also correct.
\end{proof}

\subsection{From the dendrogram \texorpdfstring{$\Td$}{Td} to the mmodule tree \texorpdfstring{$\TM$}{TM}}

To make \Cref{algo:mmax} efficient, one need to efficiently
compute the $\rhos$-components at each recursive step. To this end, we
use the dendrogram $\Td$ of the ultrametric
subdominant $(X,\hd)$. This is possible since the maximal clusters of
$\Td$ correspond to $\rhos$-components.

\begin{lemma}\label{lemma:comp-in-Td}
  Let $(X,d)$ be a dissimilarity space, $\alpha$ be an internal node
  in the dendrogram $\Td$, and $\beta$ be a child of $\alpha$. Then
  $X(\beta)$ is a $p(\alpha)$-component of $X(\alpha)$.
\end{lemma}

\begin{proof}
  %Let $x \in X(\beta)$ and $y \in X(\alpha) \setminus X(\beta)$.
  By the definition and construction of $\Td$, for any $x \in X(\beta)$ and $y \in X(\alpha) \setminus X(\beta)$
  we have $d(x,y) \geq \hd(x,y) = p(\alpha)$. Hence
  $X(\beta)$ is an union of components $C_1,\ldots,C_\ell$ of
  $G_{< p(\alpha)}(X(\alpha))$. If $X(\beta)$ is not a
  $p(\alpha)$-component, then $\ell > 1$. Let
  $x \in C_1$ and $y \in X(\beta) \setminus C_1$. Then there is a path
  from $x$ to $y$ in $X(\beta)$ with maximum weight at most
  $p(\beta) < p(\alpha)$. This path contains an edge $x'y'$
  with $x \in C_1$, $y \in X(\beta) \setminus C_1$. But then
  $d(x',y') < p(\alpha)$, which is impossible. Hence  $X(\beta)$ is a
  $p(\alpha)$-component.
\end{proof}

In order to apply $\stableDendro$ (\Cref{algo:dendro-refine}) to the dendrogram $\Td$, we need:

\begin{lemma}\label{lemma:mmclust}
  For any dissimilarity space $(X,d)$, the dendrogram $\Td$ is coherent.
\end{lemma}

\begin{proof}
  Let $M$ be an mmodule, and $\alpha$ be the $M$-pertinent node in
  $\Td$, and suppose $X(\alpha) \neq M$. Let $\beta_1,\ldots,\beta_k$
  be the children of $\alpha$ intersecting $M$ ($k \geq 2$). By way of
  contradiction, suppose that there is $i \in \{1,\ldots,k\}$ with
  $X(\beta_i) \cap M \neq \varnothing$ and
  $X(\beta_i) \setminus M \neq \varnothing$. Since $X(\beta_i)$ is a
  $p(\alpha)$-component, there exists $x,y \in X(\beta_i)$ with
  $x \in M$, $y \notin M$ and $d(x,y) < p(\alpha)$. Let
  $z \in X(\beta_j)$ for some $j \neq i$. Then
  $d(y,z) \geq p(\alpha) > d(y,x)$, hence $z \notin M$ and
  $M \subseteq X(\beta_i)$. This contradicts the fact that $\alpha$ is
  $M$-pertinent.
\end{proof}

\begin{theorem} \label{mmodTreeSubdominant}
  Using $\stableDendro$ as a stable partition algorithm, one can
  implement $\mmodTree$ to run in time $O(|X|^2)$.
\end{theorem}

\begin{proof}
  To prove the theorem, we need to establish  that (1) at each recursive call, we can build the
  dendrogram of the ultrametric subdominant for that subset; (2)  all these computations can be
  performed efficiently. Consider a call
  to $\mmodTree(\Td(S))$. Suppose that
  $\Td(S) = \Node(\beta_1,\ldots,\beta_k)$, and notice that
  $C_i = X(\beta_i)$ and $\Td(C_i) = \beta_i$ for all
  $i \in \{1,\ldots,k\}$ by \Cref{lemma:comp-in-Td}.

  We start by proving (1). First, we justify the call $\stablePart(\{C_1,S \setminus C_1\})$ in
  line~\ref{line:apply-stable-1}. $\Td(C_1)$ is available as a child
  in $\Td(S)$. Then removing that child, we get a coherent tree
  $\Tt(S \setminus C_1)$ (notice that it may not be a dendrogram).  Next we prove, through several claims, that all calls to
  $\stablePart$ return dendrograms of the ultrametric subdominants of
  their respective subsets.

  \begin{claim}\label{claim:mmt-analysis-1}
    Let $\Tt(S_i)$ be an $S_i$-tree obtained at
    lines~\ref{line:apply-stable-1} or~\ref{line:apply-stable-2}. If
    $\diam(S_i) \leq \rhos$, then $\Td(S_i) = \Tt(S_i)$.
  \end{claim}

  \begin{proof}
    If $\Tt(S_i)$ is a subtree of $\Td(S)$, then
    $\Td(S_i) = \Tt(S_i)$. Otherwise let
    $\Tt(S_i) = \Node(\beta_1,\ldots,\beta_n)$ where
    $\beta_1,\ldots,\beta_n$ define a proper subset of children of a node
    $\alpha$ in $\Td(S)$. By the properties of dendrograms, there
    exist $x \in S_i$ and $y \in X(\alpha) \setminus S_i$ with
    $d(x,y) = p(\alpha)$. As $S_i$ is a $\rhos$-component, we deduce
    $p(\alpha) \geq \rhos$. Thus $\alpha$ is the root of $\Td(S)$. Now
    let $x,y \in S_i$ be in distinct children of $\alpha$. Then
    $d(x,y) \geq \rhos$, but also $d(x,y) \leq \rhos$ since
    $\diam(S_i) \leq \rhos$. Consequently, $d(x,y) = \rhos$. Hence any two
    children of $\Tt(S_i)$ are at the same distance $\rhos$, proving
    that $\Td(S_i) = \Tt(S_i)$.
  \end{proof}

  \begin{claim}\label{claim:mmt-analysis-2}
    For $j \in \{1,\ldots,\ell\}$, let $\Tt(S_j)$ be an $S_j$-tree
    obtained at line~\ref{line:apply-stable-1}. Then
    $\Td(S_j) = \Tt(S_j)$.
  \end{claim}

  \begin{proof}
    If there is a child $\gamma$ of the root of $\TM(S)$ with
    $X(\gamma) = S_j$, then
    $\diam(S_j) \leq \min_{\gamma'} d(S_j,\gamma')$ where $\gamma'$
    ranges over the other children of the root of $\TM$. Hence
    $\diam(S_j) \leq \rhos$ and the claim follows by
    \Cref{claim:mmt-analysis-1}.

    Otherwise, $S_1$ is a dwarf component and $S_j$ is either $S_1$ or
    $S_i$ in line~\ref{line:mmtcupdwarf}. In either case,
    $\diam(S_j) \leq \rhos$ and again \Cref{claim:mmt-analysis-1}
    applies.
  \end{proof}

  \begin{claim}\label{claim:mmt-analysis-3}
    For $j \in \{1,\ldots,\ell\}$, let $\Tt(S_j)$ be an $S_j$-tree
    obtained at line~\ref{line:apply-stable-2}. Then
    $\Td(S_j) = \Tt(S_j)$.
  \end{claim}

  \begin{proof}
    By \Cref{lemma:large_block}, $\diam(S_\star) \leq \rhos$ and
    $\diam(S^\star) \leq \rhos$. Hence $\diam(S_j) \leq \rhos$ and the
    claim follows by \Cref{claim:mmt-analysis-1}.
  \end{proof}

  \Cref{claim:mmt-analysis-2,claim:mmt-analysis-3} imply that
  almost all recursive calls to $\mmodTree$ are valid since we know the
  dendrogram for each set. It remains to consider the case of $S_1 \cup S_i$ in
  line~\ref{line:mmtcupdwarf}, which happens when $C_1$ is a dwarf
  component. Then $S_1 = C_1$. If $S_i$ is a single $\rhos$-component,
  then $\Tt(S_i)$ is a single child of the root of $\Td(S)$ and
  $\Td(S_1 \cup S_i) = \Node(\Tt(S_1),\Tt(S_i))$. Otherwise, $S_i$ is a
  union of $\rhos$-components and $\Tt(S_i)$ is a join of several
  children of the root of $\Td(S)$, and $\Td(S_1 \cup S_i)$ is
  obtained by adding $\Tt(S_1)$ as a child to $\Tt(S_i)$. Thus, the
  recursive call $\mmodTree(S_1 \cup S_i)$ is also valid. 
  We also have to justify the call $\stablePart(\{S_\star,S^\star\})$ in
  line~\ref{line:apply-stable-2}. Consider $S_\star$. If
  $|I_\star| = 1$, then $S_\star$ is a $\rho$-component and thus
  $\Td(S_\star) = \Tt(C_i)$ for $i \in I_\star$. Otherwise, by
  definition of $H$, for each $i, i' \in I_\star$, $x \in C_i$ and
  $y \in C_{i'}$, we have $d(x,y) \geq \rhos$ (since $C_i$ and $C_{i'}$ are
  $\rhos$-components) and $d(x,y) \leq \rhos$ (as, in $H$, $I_\star$ is an
  independent set of vertices). This implies that $\Td(S_\star)$ is the join of
  the $\Tt(C_i)$ for $i \in I_\star$. Similarly, $\Td(S^\star)$ is the
  join of $\Tt(C_i)$ for $i \in I^\star$. This concludes the proof of
  assertion (1).

  We now prove the assertion (2) that $\mmodTree(X)$ can be implemented in $O(|X|^2)$.
  First by \Cref{prop:td-analysis} the dendrogram $\Td(X)$ can be
  computed in $O(|X|^2)$. Then we analyse the cost of
  $\mmodTree(\Td(S))$ without counting the recursive calls, depending
  on the line of return. Observe that the cost of
  lines~\ref{line:def-Ci} and~\ref{line:def-IJ} is
  $O(\sum_{i=1}^k |C_i| |S \setminus C_i|)$. We will relate in each case  the cost
  to the number $N$ of pairs of elements in distinct children
  of the root of $\TM(S)$.

  If $\mmodTree(\Td(S))$ returns on line~\ref{line:mmtcap} (the case
  of a non-special $\cap$-node), then the total cost is
  $O(\sum_{i=1}^k |C_i| |S \setminus C_i|)$, which is proportional to
  $N$. As we will see in the next two cases, this case can be charged
  with an additional cost, still proportional to $N$.

  If $\mmodTree(\Td(S))$ returns on line~\ref{line:mmtcupgiant} (the
  case of a $\cup$-node discovered through a giant component $C_1$),
  then the total cost is
  $$O\left(\sum_{i=1}^k|C_i| |S \setminus C_i| + \sum_{i=1}^\ell |S_i| |S \setminus S_i|\right)$$
  (counting line~\ref{line:apply-stable-1}). The left-hand term is the
  number of pairs of elements in distinct children of the root of
  $\TM(S)$. By \Cref{lemma:giant-or-dwarf}, each component $C_i$ is
  either dwarf or giant. If $C_i$ is giant, $x \in C_i$ and
  $y \in S \setminus C_i$, then $x$ and $y$ are not in the same
  $\rhos$-component, hence those pairs $x,y$ are counted in
  $\sum_{i =1}^\ell |S_i| |S \setminus S_i|$. Thus, without counting
  the additional contribution from dwarf components, we get a cost
  proportional to $N$. It remains to count the pairs $x,y$ where $x$
  and $y$ are in two distinct dwarf components that are part of the
  same $S_i$. In that case
  $\TM(S_i) = \cup(\gamma_1,\ldots,\gamma_m)$ where
  $X(\gamma_1),\ldots,X(\gamma_m)$ are the dwarf components from that
  $S_i$. We thus charge the uncounted cost induced by those dwarf
  components to the call $\mmodTree(\Td(S_i))$ (the charged cost is
  proportional to its $N$ value).

  If $\mmodTree(\Td(S))$ returns on line~\ref{line:mmtcupdwarf} (the
  case of a $\cup$-node discovered from a dwarf component), the
  analysis is similar to the previous case. Notice that there is an
  additional cost of $O(|S_1| |S_i|)$ induced in
  line~\ref{line:apply-stable-1} by the fact that $S_1$ and $S_i$ are
  not mmodules. We charge it to $\mmodTree(S_1 \cup S_i)$ as $S_1$
  is a dwarf component of the mmodule $S_1 \cup S_i$.

  If $\mmodTree(\Td(S))$ returns on lines~\ref{line:mmtsuperspecial}
  or~\ref{line:mmtspecial} (the case of a special $\cap$-node), then
  the total cost is
  $$O\left(\sum_{i=1}^k |C_i| |S \setminus C_i| + \sum_{i \in I} |C_i| |(S_\star \cup S^\star) \setminus C_i| + \sum_{i=1}^l |S_i| |(S_\star \cup S^\star) \setminus S_i|\right).$$
  Let $i \in \{1,\ldots,k\}$. If $i \in J$, then $C_i = X(\beta_j)$
  for some $j \in \{1,\ldots,k'\}$, that is $C_i$ is the set induced
  by some child of $\TM(S)$. Otherwise, $i \in I$. If
  $C_i = \bigcup_{i \in L} S_i$ for some $L \subset \{1,\ldots,\ell\}$,
  then
  $|C_i| |(S_\star \cup S^\star) \setminus C_i| \leq \sum_{j \in L} |S_j| |(S_\star \cup S^\star) \setminus S_j|$.
  Otherwise there is some $j \in \{1,\ldots,\ell\}$ such that
  $C_i \cap S_j \neq \varnothing$ and
  $S_j \setminus C_i \neq \varnothing$. We may assume that
  $S_j \subseteq S_\star$ (the case $S_j \subseteq S^\star$ is
  similar). Since $\diam(S_\star) \leq \rhos$, we obtain
  $\diam(S_j) \leq \rhos$. Thus for any $x \in S_j \cap C_i$ and
  $y \in S_j \setminus C_i$, $d(x,y) = \rhos$. Thus
  $\TM(S_j) = \cap(\gamma_1,\ldots,\gamma_m)$ with
  $\rhos(S_j) = \rhos$ and $X(\gamma_1),\ldots,X(\gamma_m)$ are
  $\rhos$-components of $S$, one of them being $C_i$. That is, $C_i$
  is a dwarf component of $S_\star \cup S^\star$. Similarly to the two
  previous case, we may charge the cost related to the dwarf
  components to the corresponding recursive call $\mmodTree(S_j)$. To
  conclude that case, we get that the total cost is composed of one
  term that is proportional to the number of pairs of elements
  separated by the current call:
  $$ \sum_{j \in J} |C_j| |S \setminus C_j| + \sum_{i=1}^{l} |S_i| |S \setminus S_i|,$$
  and another term corresponding to dwarfs components in
  $S_\star \cup S^\star$, that is charged one their respective
  recursive calls.

  Finally, each call incurs a cost proportional to the number of pairs
  of distinct elements of $S$ that are not together in a deeper
  recursive call. Hence each pair of $\binom{X}{2}$ is counted exactly
  once, thus the total complexity is $O(|X^2|)$.
\end{proof}

\subsection{Nodes of \texorpdfstring{$\Td$}{Td} and nodes of \texorpdfstring{$\TM$}{TM} and \texorpdfstring{$\TPQ$}{TPQ}}

We conclude this section with a characterization of nodes of $\TM$ and
$\TPQ$ that correspond to nodes of $\Td$.

\begin{lemma}\label{lemma:cluster-pqt}
  Let $(X,d)$ be a Robinson space with PQ-tree $\TPQ$ and let $\beta$
  be an internal node of $\TPQ$. Then one of the following assertions
  holds:
  \begin{enumerate}[label=(\roman*)]
  \item\label{item:cluster1} $X(\beta)$ is a cluster of $\Td$,
  \item\label{item:cluster2} $\beta$ is P-node, $\beta$ is the child of a
    Q-node $\alpha$, and for each child $\gamma$ of $\beta$,
    $X(\gamma)$ is a cluster of $\Td$.
  \end{enumerate}
\end{lemma}

\begin{proof}
  Suppose that $X(\beta)$ is not a cluster, in particular $\beta$ is
  not the root of $\TPQ$, and let $\alpha$ be its parent node. Let
  $\rhos := \rhos(X(\beta))$ be the minimum value such that the graph
  $G_{\leq \rhos}(X(\beta))$ is connected. Let also $C$ be the
  minimal cluster containing $X(\beta)$, and let $\rhos'$ be the
  weight of the node corresponding to $C$ in the vertex representation of $\Td$. In particular,
  $\min \{d(x,z): x \in X(\beta), z \in C \setminus X(\beta)\} \leq \rhos'$. By
  \Cref{lemma:monotone-diameter} this implies
  $\rhos \leq \diam(X(\beta)) \leq \rhos'$.

  If $\rhos < \rhos'$, then there would be a cluster with value at
  most $\rhos$ containing $X(\beta)$, contradicting the minimality of
  $C$. Thus $\rhos = \diam(X(\beta)) = \rhos'$, which implies that
  $\Grho(X(\beta))$ is not connected, and each of its
  $\rhos$-components has diameter at most $\rhos$. By
  \Cref{lemma:some-d-cut}, $\beta$ is a P-node, and
  any  children $\gamma$ is a $\rhos$-component of $X(\beta)$.

  By way of contradiction, suppose that $\alpha$ is a P-node. By
  \Cref{lemma:monotone-diameter}\ref{item:mono3},
  $\rhos(\alpha) > \rhos$. As $X(\alpha)$ is a block, for all
  $x \in X(\alpha)$ and $y \in X \setminus X(\alpha)$,
  $d(x,y) \geq \rhos(\alpha) > \rhos$. Also, as $\alpha$ is a P-node,
  for each $x \in X(\beta)$ and $y \in X(\alpha) \setminus X(\beta)$,
  $d(x,y) = \rhos(\alpha) > \rhos$. But then, $X(\beta)$ is a cluster
  of value $\rhos$, contradiction. Thus $\alpha$ is a Q-node.

  Since we proved that any node of the PQ-tree $\TPQ$ that does not induce
  a cluster of $\Td$ has a Q-node parent, we deduce that each child of the P-node $\beta$
  induces a cluster. Thus~\ref{item:cluster2} is proved.
\end{proof}

Leveraging the translation between PQ-trees and mmodule trees, we get:

\begin{lemma}\label{lemma:cluster-mt}
  Let $(X,d)$ be a Robinson space with mmodule tree $\TM$ and $\beta$
  be an internal node of $\TM$. Then one of the following assertions
  holds:
  \begin{enumerate}[label=(\roman*)]
  \item\label{item:clust1} $X(\beta)$ is a cluster of $\Td$,
  \item\label{item:clust2} $\beta$ is the large child of a $\cap$-node,
  \item\label{item:clust3} $\beta$ is a $\cap$-node and is the child
    of a $\cup$-node, and for each child $\gamma$ of $\beta$,
    $X(\gamma)$ is a cluster of $\Td$.
  \end{enumerate}
\end{lemma}

\begin{proof}
  By \Cref{prop:node-to-node}, either $\beta$ is the large child of a
  special node $\alpha$ (that is~\ref{item:clust2} holds), or there is
  a node $\beta'$ in the PQ-tree $\TPQ$ with
  $X(\beta') = X(\beta)$. In the later case, by
  \Cref{lemma:cluster-pqt}, either $X(\beta')$ is a cluster (and so is
  $X(\beta)$, thus~\ref{item:clust1} holds), or $\beta'$ is a P-node child
  of a Q-node $\alpha'$. Again, we assume the later. Then there is a
  node $\alpha$ in $\TM$ with $X(\alpha) = X(\alpha')$ by
  \Cref{prop:node-to-node}.

  Since $X(\alpha) = X(\alpha')$, $X(\beta) = X(\beta')$, $\beta'$ is a
  P-node and $\alpha'$ a Q-node, by applying \Cref{prop:pqroot-analysis} to
  $\alpha'$, we conclude that $\beta$ is a $\cap$-node and one of three possibilities
  happens:
  \begin{enumerate}[wide,label=Case \arabic*:]
  \item $\alpha'$ is non-conical, then $\alpha$ is a
    $\cup$-node, $\beta$ is a child of $\alpha$. Then, as $X(\beta)$
    is not a cluster, there is another child $\beta'$ of $\alpha$ with
    $d(\beta,\beta') = \rhos(\beta)$, hence by
    \Cref{lemma:monotone-diameter},
    $\diam(X(\beta)) = \rhos(\beta)$. This implies that $\beta$ has no
    large child, and each of its child is a cluster, thus 
    \ref{item:clust3}~holds.
  \item $\alpha'$ is conical, but its apex is not
    $\beta'$, then $\alpha$ has a large child $\gamma$, $\gamma$ is a
    $\cup$-node or have arity two, and $\beta$ is a child of
    $\gamma$. As in the previous case, as $X(\beta)$ is not a cluster, its
    diameter is $\rhos(\beta)$ and each of its child is a cluster, whence 
    \ref{item:clust3}~holds.
  \item $\alpha'$ is conical with apex $\beta'$, then
    $\alpha$ is a special $\cap$-node. Also $\beta'$ cannot be split
    (otherwise there would not be a node with leaf set $X(\beta')$
    in $\TM$), hence $\beta$ is a child of $\alpha$. Then for each
    $x \in X(\beta)$, $y \notin X(\beta)$,
    $d(x,y) \geq \rhos(\alpha) > \rhos(\beta)$ by
    \Cref{lemma:monotone-diameter}. Thus $X(\beta)$ is a cluster, whence
    \ref{item:clust1}~holds.\qedhere
  \end{enumerate}
\end{proof}

%% file: sec_7_copoints_to_pq.tex
In our previous paper~\cite{CaCheNaPre}, the copoints of a point $p$
were at the heart of our divide-and-conquer recognition algorithm.
This is due to the fact that the copoints of $p$ together with $\{
p\}$ define a partition $\C_p$ of $X\setminus \{ p\}$. In this
section, we provide a correspondence between the copoint partition
$\C_p$ and the trees $\TM$ and $\TPQ$ of a Robinson space $(X,d)$.
First, for any dissimilarity space $(X,d)$ we characterize the
copoints of $\C_p$ in terms of subtrees of the mmodule tree $\TM$
rooted at the nodes of the unique path $\Psi(p)$ between $\{ p\}$ and
the root of $\TM$. This allows us to establish that the total number
of copoints of $(X,d)$ is at most $2|X| - 1$. Second, we characterize
the nodes on the unique path $\Upsilon (p)$ of the PQ-tree $\TPQ$ of a
Robinson space $(X,d)$ between the leaf $\{ p\}$ and the root of
$\TPQ$. We also locate the copoints of $\C_p$ with respect to the
nodes of this path $\Upsilon (p)$ and the $p$-proximity order,
introduced in~\cite{CaCheNaPre}.

\subsection{Translation between \texorpdfstring{$\C_p$ and $\Psi(p)$}{Cp and Psi(p)}}
The next result characterizes the copoints attached at $p$ in terms of
subtrees of $\TM$.

\begin{proposition}\label{lemma:copoint-in-mtree}
  Let $(X,d)$ be a dissimilarity space with mmodule tree $\TM$. For
  any $p \in X$, the $p$-copoints of $X$ are:
  \begin{enumerate}[label=(\roman*)]
  \item\label{item:cpm1} $X(\alpha) \setminus X(\beta)$ for each
    $\cap$-node $\alpha\in \Psi(p)$ with $\beta$ be the child of
    $\alpha$ in  $\Psi(p)$,
  \item\label{item:cpm2} $X(\beta)$ for each $\cup$-node
    $\alpha\in \Psi(p)$ and each child $\beta$ of $\alpha$ such
    that $\beta \notin \Psi(p)$).
  \end{enumerate}
\end{proposition}

\begin{proof}
  Pick any $\alpha\in \Psi(p)$. Let
  $S := X(\alpha) \setminus X(\beta)$ if $\alpha$ is a $\cap$-node
  with child $\beta\in \Psi(p)$ or $S := X(\beta)$ for a child
  $\beta\notin \Psi(p)$ of $\alpha$ if $\alpha$ is a $\cup$-node. By
  \Cref{mmodule-tree}, $S$ is an mmodule. Moreover, also by
  \Cref{mmodule-tree}, any mmodule properly containing $S$ contains
  $X(\alpha)$ and $p$, hence $S$ is a $p$-copoint.

  Conversely, let $S$ be a $p$-copoint, in particular $S$ is an
  mmodule. By \Cref{mmodule-tree}, the following holds:
  \begin{itemize}[label=-]
  \item either $S$ is the leaf-set of the union of some children of a
    $\cap$-node $\alpha$. Then by maximality of $S$, as $X(\alpha)$ is
    an mmodule, $p \in X(\alpha)$. Let $S'$ be
    $\bigcup \{ X(\beta) : \beta\textrm{ child of
    }\alpha, p \notin X(\beta)\}$. Again by \Cref{mmodule-tree}, $S'$
    is an mmodule, and $S \subseteq S'$, $p \notin S'$. By maximality
    of $S$, $S = S'$.
  \item or $S$ is the leaf-set of a child $\beta$ of a $\cup$-node
    $\alpha$. Then by maximality of $S$, as $X(\alpha)$ is an mmodule,
    $p \in X(\alpha)$.
  \end{itemize}
  This concludes the proof.
\end{proof}

From \Cref{copoint-partition} it follows that any dissimilarity space
$(X,d)$ contains at most $|X|(|X|-1)$ copoints.  In fact, by
\Cref{lemma:copoint-in-mtree} the number of copoints is
always linear in the size of $X$:

\begin{corollary}\label{lemma:number-copoints}
  The number of copoints of a dissimilarity space $(X,d)$ is at most
  $2|X| - 1$.
\end{corollary}

\begin{proof}
  By \Cref{lemma:copoint-in-mtree}, each node of the mmodule tree
  induces at most as many copoints as its arity. As the sum of arities
  equals the number of nodes minus one, and each inner node as arity
  at least 2, we get the result.
\end{proof}

Since by \Cref{prop:ultrametric_mmodule_tree} all inner nodes of the
mmodule tree of an ultrametric space are $\cap$-nodes, from
\Cref{lemma:copoint-in-mtree} we obtain the following observation:

\begin{corollary}\label{copoints-ultrametric}
  The $p$-copoints of an ultrametric space $(X,d)$ are the sets of the
  form $X(\alpha) \setminus X(\beta)$, where $\alpha\in \Psi(p)$ and
  $\beta$ is the child of $\alpha$ in $\Psi(p)$.
\end{corollary}

\subsection{Translation between  \texorpdfstring{$\C_p$ and $\Upsilon(p)$}{Cp and Upsilon(p)}}

We already know by \Cref{LEMMA_block_node} that for each node $\alpha$
of the PQ-tree, $X(\alpha)$ is an mmodule. Next we determine which
mmodules do not correspond to nodes of $\TPQ$.

\begin{lemma}\label{lemma:mmod-in-pqt}
  Let $(X,d)$ be a Robinson space and $\TPQ$ its PQ-tree. Let
  $M \subseteq X$. Then $M$ is an mmodule of $(X,d)$ if and only if
  \begin{enumerate}[label=(\roman*)]
  \item\label{item:minpq1} either there is a node $\alpha \in \TPQ$
    such that $M = X(\alpha)$,
  \item\label{item:minpq2} or there is a P-node
    $P(\beta_1,\ldots,\beta_k) \in \TPQ$ and
    $I \subset \{1,\ldots,k\}$ non-empty such that
    $M = \bigcup_{i \in I} X(\beta_i)$,
  \item\label{item:minpq3} or there is a $\delta$-conical Q-node
    $\alpha = Q(\gamma_1,\ldots,\gamma_l) \in \TPQ$, with apex child
    $\gamma_j$, and $M = X(\alpha) \setminus X(\gamma_j)$,
  \item\label{item:minpq4} or there is a $\delta$-conical Q-node
    $\alpha = Q(\gamma_1,\ldots,\gamma_l) \in \TPQ$, with split child
    $\gamma_j$, and a subset $B$ of children of $\gamma_j$, such that
    $M = X(\alpha) \setminus \bigcup_{\beta \in B} X(\beta)$.
  \end{enumerate}
\end{lemma}

\begin{proof}
  By \Cref{mmodule-tree}, either (a) there is a node $\beta$ in $\TM$
  with $M = X(\beta)$, or (b) there is a $\cap$-node
  $\alpha = \cap(\beta_1,\ldots,\beta_k)$ and
  $I \subsetneq \{1,\ldots,k\}$, $|I| \geq 2$ such that
  $M = \bigcup_{i \in I} X(\beta_i)$.

  In case (a), suppose~\ref{item:minpq1} does not hold, then by
  \Cref{prop:node-to-node} $\beta$ is a large child of a special
  Q-node $\alpha$. By \Cref{lemma:monotone-diameter}, $\alpha$ is not
  a large child itself, hence there is a P-node $\alpha'$ in $\TPQ$
  with $X(\alpha) = X(\alpha')$. By \Cref{lemma:cut-separable},
  $\alpha' = Q(\gamma_1,\ldots,\gamma_{i-1},\beta',\gamma_i,\ldots,\gamma_\ell)$,
  with $X(\alpha) \setminus X(\beta') = X(\beta)$,
  hence~\ref{item:minpq3} holds.

  In case (b), $\alpha$ is not a large child because $k = 2$, hence by
  \Cref{prop:node-to-node} there is a P-node $\alpha'$ in $\TPQ$ with
  $X(\alpha) = X(\alpha')$. If $\alpha$ has no split child, then by
  \Cref{lemma:some-d-cut}, for each $\beta_i$, there is a child
  $\beta'_i$ of $\alpha$ with $X(\beta_i) = X(\beta'_i)$,
  and~\ref{item:minpq2} follows. Thus we may assume that there is a
  split child, say $\beta_k$. By \Cref{lemma:cut-separable},
  $\alpha' = Q(\gamma_1,\ldots,\gamma_{j-1},\beta',\gamma_j,\ldots,\gamma_\ell)$
  with $X(\alpha') \setminus X(\beta') = X(\beta_k)$ and
  $\beta' = P(\gamma_1,\ldots,\gamma_{k-1})$ such that for each
  $i \in \{1,\ldots,k-1\}$, $X(\gamma_i) = X(\beta_i)$. If $k \in I$,
  \ref{item:minpq4}~holds, while if $k \notin I$, \ref{item:minpq3}
  holds on $\beta'$.
\end{proof}

Our next result establishes a correspondence between the copoints of
$\C_p$ and the nodes of $\Upsilon(p)$.

\begin{theorem}\label{thm:copoint-in-pqtree}
  Let $(X,d)$ be a Robinson space. For any $p \in X$, the copoints of
  $\C_p$ are:
  \begin{enumerate}[label=(\arabic*)]
  \item\label{item:cpp1} $X(\alpha) \setminus X(\beta)$ for each P-node
    $\alpha \in \Upsilon(p)$ with child $\beta \in \Upsilon(p)$,
  \item\label{item:cpp2} $X(\beta)$ for each Q-node
    $\alpha \in \Upsilon(P)$ and each child
    $\beta \notin \Upsilon(p)$, when, if $\alpha$ is conical, $p$ is not in the apex child of
    $\alpha$,
  \item\label{item:cpp3} $X(\alpha) \setminus X(\beta)$ for each
    $\delta$-conical Q-node $\alpha \in \Upsilon(p)$ with standard
    apex child $\beta \in \Upsilon(p)$, and $\Gdelta(X(\beta))$ is
    connected,
  \item\label{item:cpp4} $X(\alpha) \setminus X(\gamma) $ for each
    $\delta$-conical Q-node $\alpha \in \Upsilon(p)$ with split child
    $\beta \in \Upsilon(p)$, and $\gamma \in \Upsilon(p)$ child of
    $\beta$.
  \end{enumerate}
\end{theorem}

\begin{proof}
  Let $M$ be a $p$-copoint. Then $M$ is an mmodule and by
  \Cref{lemma:mmod-in-pqt} one of the four following cases occurs.

  \begin{itemize}
  \item[\ref{item:minpq1}] There is a node $\alpha$ in $\TPQ$ such
    that $M = X(\alpha)$. As $p \notin X(\alpha)$, $\alpha$ is not the
    root of $\TPQ$, and let $\gamma$ be its parent. By maximality of $M$,
    $p \in X(\gamma)$. By \Cref{lemma:mmod-in-pqt}\ref{item:minpq2}
    $\gamma$ is not a P-node, hence is a Q-node. Suppose that it is
    conical and $p$ is in an apex child $\alpha'$ of $\gamma$. Then
    $X(\gamma) \setminus X(\alpha')$ is an mmodule containing
    $X(\alpha)$ and not $p$, contradicting the maximality of $M$. Thus
    case~\ref{item:cpp2} applies.
  \item[\ref{item:minpq2}] There is a node
    $\alpha = P(\beta_1,\ldots,\beta_k)$ in $\TPQ$ and
    $I \subsetneq \{1,\ldots,k\}$ such that
    $M = \bigcup_{i \in I} X(\beta_i)$. By maximality of $M$,
    $p \in X(\alpha)$ and $|I| = k-1$, hence
    $M = X(\alpha) \setminus X(\beta_i)$ for the unique
    $i \in \{1,\ldots,k\} \setminus I$, $p \in X(\beta_i)$ and
    case~\ref{item:cpp1} applies.
  \item[\ref{item:minpq3}] There is a $\delta$-conical Q-node
    $\alpha = Q(\beta_1,\ldots,\beta_k)$, with $\beta_i$ apex, and
    $M = X(\alpha) \setminus X(\beta)$. Then as $X(\alpha)$ is an
    mmodule, by maximality of $M$, $p \in X(\alpha)$ and thus
    $p \in X(\beta_i)$. Suppose that $\Gdelta(X(\beta_i))$ is not
    connected, let $S$ be one of its $\delta$-mmodule not containing
    $p$. Then $M \cup S$ is an mmodule not containing $p$,
    contradicting the maximality of $p$. Thus case~\ref{item:cpp3}
    applies.
  \item[\ref{item:minpq4}] There is a $\delta$-conical Q-node
    $\alpha = Q(\beta_1,\ldots,\beta_k)$ with an apex child $\beta_i$
    such that the graph $\Gdelta(X(\beta_i))$ is not connected, and there is
   a subset  $\Gamma$ of children of $\beta_i$ such that
    $M = X(\alpha) \setminus \bigcup_{\gamma \in \Gamma} X(\gamma)$. As
    $X(\alpha)$ is an mmodule, by maximality of $M$, $p \in X(\alpha)$
    hence $p \in X(\gamma^\star)$ for some
    $\gamma^\star \in \Gamma$. Then by maximality of $M$,
    $\Gamma = \{\gamma^\star\}$, and case~\ref{item:cpp4} applies.
  \end{itemize}

  Conversely, let $M$ be a set as in one of the cases~\ref{item:cpp1}
  to~\ref{item:cpp4}. Then by \Cref{lemma:mmod-in-pqt}, $M$ is an
  mmodule. Moreover, any mmodule properly containing $M$ must contain
  $X(\alpha)$, hence contains $p$, proving that $M$ is a $p$-copoint.
\end{proof}

\subsection{\texorpdfstring{$\Upsilon(p)$}{Upsilon(p)} and the \texorpdfstring{$p$}{p}-proximity order}

We recall the definition of $p$-proximity orders of Robinson spaces,
which was one of the ingredients of our recognition algorithm
in~\cite{CaCheNaPre}:

\begin{definition}[p-Proximity order~\cite{CaCheNaPre}]
  Let $(X,d)$ be a Robinson space with a compatible order $\<$ and let
  $p$ be a point of $X$.  A {\em $p$-proximity order} (relatively to
  $\<$) is a total order $\prec$ on $\C_p$ such that if $C,C'\in \C_p$
  and $C \prec C'$, then:
  \begin{enumerate}[label=(PO\arabic*),nosep]
  \item\label{item:proxorder1} $d(C, p) \leq d(C', p)$;
  \item\label{item:proxorder2} if $X$ is sorted according to $\<$,
    then no point of $C'$ is located between $p$ and a point of $C$.
  \end{enumerate}
\end{definition}

A $p$-proximity order exists for every compatible order and can be
efficiently computed, even without the knowledge of the compatible
order $\<$~\cite{CaCheNaPre}. Actually,
in~\cite[Algorithm~6.1]{CaCheNaPre} we constructed a \emph{universal}
compatible order: an order $\prec$ on $\C_p$ which is a $p$-proximity
order relatively to any compatible order $<\in \Pi(X,d)$; in what
follows we will consider only universal $p$-proximity orders. Given a
universal $p$-proximity order $\prec$, suppose that the copoints of
$\C_p$ are ordered in the following way:
$$\{ p\}:=C_0\prec C_1\prec\ldots\prec C_k.$$
For any $C_i$ we denote by $[C_0,C_i]$ the points in the copoints
between $C_0$ and $C_i$ in the ordered set $(\C_p,\prec)$:
$[C_0,C_i]: =C_0\cup \ldots \cup C_i$. We call $[C_0,C_i]$ an
\emph{initial interval} of $\prec$.  Since $\prec$ is a $p$-proximity
order for all compatible orders, any initial interval $[C_0,C_i]$ is a
block.

Next, we will characterize the copoints $C \in \C_p$ for which the
initial interval $[C_0,C]$ is also an mmodule. Since the nodes of
$\Upsilon(p)$ are exactly the subsets of $X$ which are simultaneously
mmodules and blocks and also contain $p$, these intervals will
correspond to nodes of $\Upsilon(p)$. Let
$\Upsilon(p) = (\alpha_0,\alpha_1,\ldots,\alpha_{m})$, where
$\alpha_0$ correspond to $p$ and $\alpha_m$ corresponds to the root of
$\TPQ$. Then
$\{ p\} =X(\alpha_0) \subsetneq X(\alpha_1) \subsetneq \ldots \subsetneq X(\alpha_m) = X$
is a chain of subsets of $X$ which are mmodules and blocks containing
$p$. However, in view of \Cref{thm:copoint-in-pqtree}, only those
$X(\alpha_i)$ where $\alpha_i$ is standard correspond to initial
intervals of $\prec$. We now formalize these ideas, starting with a
definition motivated by the previous discussion.

\begin{definition}[frontier]
  A copoint $C$ is a {\em frontier} if $[C_0,C]$ is an mmodule.
\end{definition}

% with a simple metric
% property of nodes of $\Upsilon(p)$:

% \begin{lemma}\label{lemma:nodes_upsilon_metric}
%   If $\alpha$ is a node of $\Upsilon(p)$, then for any
%   $x,y \in X(\alpha)$ and $z \in X \setminus X(\alpha)$, we have
%   $d(x,y) \le d(z,p) = d(z,x) = d(z,y)$.
% \end{lemma}

% \begin{proof}
%   By \Cref{LEMMA_block_node}, $X(\alpha)$ is a block and an mmodule.
%   Since $x,y,p \in X(\alpha),$ $z \notin X(\alpha)$, and $X(\alpha)$
%   is an mmodule, we have $d(z,p) = d(z,x) = d(z,y)$. Since $X(\alpha)$
%   is a block, for any compatible order $<$, $z$ is not between $x$ and
%   $y$, yielding
%   $d(x,y) \le \max \{d(z,x), d(z,y)\} = d(z,x) = d(z,y)$.
% \end{proof}

% \begin{proposition}
%   If $C$ is a frontier of a Robinson space $(X, d)$, then or all
%   $x, y \in [C_0,C]$ and $ z \in X \setminus [C_0, C]$, we have
%   $d(x, y) \leq d(x, z) = d(y, z)$.
% \end{proposition}

% \begin{proof}
%   Since every initial interval is a block, there exists, up to a
%   symmetry, a compatible order with $x < y < z$. Hence
%   $d(x, y) \leq d(x, z) = d(y,z)$ where the equality follows from
%   the fact that $[C_0,C]$ is an mmodule.
% \end{proof}

We establish the correspondence between frontiers and standard nodes
of $\Upsilon(p)$.

\begin{theorem}\label{lemma:node_upsilon_frontiere}
  Let $(X, d)$ be a Robinson space, $p$ a point of $X$, and
  $S \subseteq X$ with $p \in S$. Let $\prec$ be a universal
  $p$-proximity order. Then there is a standard node
  $\alpha \in \Upsilon(p)$ with $X(\alpha) = S$ if and only if there
  exists a frontier $C \in \C_p$ with $[C_0,C] = S$.
  % Frontiers in $\C_p$ and standard nodes in $\Upsilon_p$ are in
  % one-to-one correspondance. More precisely: a node
  % $\alpha \in \Upsilon_p$ is a standard node if and only there exists
  % a frontier $C \in \C_p$ such that $X(\alpha)= [C_0, C]$. A copoint
  % $C \in \C_p$ is a frontier if and only if there exists a standard
  % node $\alpha \in \Upsilon_p$ such that $X(\alpha) = [C_0, C]$.
\end{theorem}

\begin{proof}
  Let $\alpha \in \Upsilon(p)$ be a standard node. By
  \Cref{thm:copoint-in-pqtree}, each copoint intersecting $X(\alpha)$
  is contained in $X(\alpha)$. Let $\C_p(\alpha)$ be the set of
  copoints contained in $X(\alpha)$. Let $C' \in \C_p(\alpha)$ and
  $C'' \in C_p \setminus C_p(\alpha)$, and suppose by way of
  contradiction that $C'' \prec C'$. Consider a compatible order $<$.
  Let $x' \in C'$ and $x'' \in C''$. We may assume (up to reversal of
  $<$) that $p < x''$.

  As $X(\alpha)$ is a block by \Cref{LEMMA_block_node}, $x''$ is not
  between $p$ and $x'$ in $<$, hence $x' < x''$. Moreover, $\prec$
  being a $p$-proximity order, by~\ref{item:proxorder2} and
  $C'' \prec C'$, $x'$ is not between $p$ and $x''$ in $<$, hence
  $x' < p < x''$. Then because $X(\alpha)$ is a block, by
  \Cref{lemma:mmodule_interval}, $\rev{X(\alpha)}$ is a compatible
  order, with $p \rev{X(\alpha)} x' \rev{X(\alpha)} x''$, in
  contradiction with the fact that $\prec$ is a universal
  $p$-proximity order. Thus, if $C$ is the maximum element in
  $\C_p(\alpha)$ for $\prec$, then $X(\alpha) = [C_0,C]$ holds.

  Conversely, let $C \in \C_p$ be a frontier, meaning by definition
  that $[C,C_p]$ is an mmodule and a block. By
  \Cref{LEMMA_block_node}, there is a node $\alpha$ in $\TPQ$ with
  $X(\alpha) = [C_0,C]$. Clearly $\alpha \in \Upsilon(p)$. Suppose by
  contradiction that $\alpha$ is not standard, \emph{i.e.} $\alpha$ is
  a split node, child of a $\delta$-conical node $\beta$. Let $\gamma$
  be the child of $\alpha$ with $p \in X(\gamma)$. Then by
  \Cref{thm:copoint-in-pqtree}, $X(\beta) \setminus X(\gamma)$ is a
  $p$-copoint intersecting $C$, but the $p$-copoints are disjoint,
  hence this is a contradiction.
\end{proof}

The algorithm from~\cite{EhGaMcCSu} can be rephrased and applied to
build the mmodule tree over a dissimilarity space. Recall that it is
based on the partition into $p$-copoints and then building the path of
the mmodule tree from $p$ to the root. Our understanding of mmodules
and of PQ-trees together with the algorithm from~\cite{CaCheNaPre}
lead to an alternative approach to the construction of the PQ-tree,
which also provides the orders of the children of Q-nodes. Toward this
goal, we first characterize the frontiers corresponding to each kind
of nodes of the PQ-tree: P-node, non-conical Q-node, or conical Q-node
with or without a split child.

\begin{algorithm}[htbp]
  \caption{Computes the PQ-tree of a Robinson space $(X,d)$ using the
    copoints attached at some point $p$.}
  \label{algo:copoints2pqtree}

  \flushleft{$\underline{\cptopqt(S)}$}
  \begin{algorithmic}[1]
    \Require{A Robinson space $(X,d)$ (implicit), a set $S \subseteq X$.}
    \Ensure{$\TPQ(S)$ the PQ-tree of $(S,d)$.}
    \Let $p \in S$
    \If {$|S| = 1$ }
      \Return $\Leaf\ p$
    \EndIf
    \Let $C_0 = \{p\}, C_1,\ldots,C_k = \stablePart(\{p\}, S \setminus \{p\})$ \Comment{with $C_0 \prec C_1 \prec \ldots \prec C_k$}
    \Return $\cptopqtC(p,C_1,\ldots,C_k)$
  \end{algorithmic}\medskip

  \flushleft{$\underline{\cptopqtC(p, C_1,\ldots,C_k)}$}
  \begin{algorithmic}[1]\setcounter{ALG@line}{5}
    \Let $(L, i, R) := \sortSiblings(p, C_1,\ldots,C_k)$
    \Let $T_p := \cptopqtC(p,C_1,\ldots,C_i)$ if $i > 0$, $T_p = \Leaf\ p$ if $i = 0$
    \If {$i < k-1$}
      \Let $[\beta_1,\ldots,\beta_l] = \map(\cptopqt, L) \append [T_p] \append \map(\cptopqt, R)$
      \Return $Q(\beta_1,\ldots,\beta_l)$ \Comment{Case~\ref{item:cptpqt1}}
    \EndIf
    \Let $\alpha := \cptopqt(C_k)$, $\delta := d(p,C_k)$, and $D = \diam(C_k)$ \Comment{$C_{k-1}$ is a frontier}
    \Match {$\alpha$}
    \Case {$P(\beta_1,\ldots,\beta_l)$ with $D = \delta$} \Comment{Case~\ref{item:cptpqt211}}
      \Return $P(\beta_1,\ldots,\beta_l,T_p)$
    \Case {$P(\beta_1,\ldots,\beta_l)$ with $D < \delta$ or $Q(\beta_1,\ldots,\beta_l)$ with $D \leq \delta$} \Comment{Case~\ref{item:cptpqt212}}
      \Return $P(\alpha, T_p)$
    \Case {$P(\beta_1,\beta_2)$ with $D > \delta$} \Comment{Case~\ref{item:cptpqt221}}
      \Return  $Q(\beta_1, T_p, \beta_2)$
    \Case {$Q(\beta_1,\ldots,\beta_l)$ with $D > \delta$}
      \If {$\alpha$ is $\delta$-conical with split child $\beta_j = P(\gamma_1,\ldots,\gamma_n)$} \Comment{Case~\ref{item:cptpqt231}}
        \Return $Q(\beta_1,\ldots,\beta_{j-1},P(\gamma_1,\ldots,\gamma_n,T_p),\beta_{j+1},\ldots,\beta_l)$
      \ElsIf {$\alpha$ is $\delta$-conical with standard apex child $\beta_j$}  \Comment{Case~\ref{item:cptpqt232}}
        \Return $Q(\beta_1,\ldots,\beta_{j-1},P(\beta_j,T_p),\beta_{j+1},\ldots,\beta_l)$ \Comment{New split child}
      \Else \Comment{Case~\ref{item:cptpqt222}}
        \Let $(j,j+1)$ be an admissible hole in $\alpha$ \label{line:find-adm-hole}
        \Return $Q(\beta_1,\ldots,\beta_{j},T_p,\beta_{j+1},\ldots,\beta_l)$ \Comment{New apex child}
      \EndIf
    \EndMatch
  \end{algorithmic}\medskip

  \flushleft{$\underline{\sortSiblings(p,C_1,\ldots,C_k)}$}
  \begin{algorithmic}[1]\setcounter{ALG@line}{26}
    \Require{A Robinson space $(X,d)$ (implicit), $p \in X$ and
      $C_1 \prec C_2 \prec \ldots \prec C_k$ the $p$-copoints with
      $\prec$ a universal $p$-proximity order.}
    \Ensure{$(L,i,R)$ where $i < k$ is maximum such that $C_i$ is a
      frontier (or $i = 0$ if there is no such frontier),
      $L \append R$ contains $C_{i+1}, \ldots, C_k$ and there is a
      compatible order such that $L \append [\{p\}] \append R$ is
      increasing.}
    \Let $L := []$, $R := [C_k]$, $i := k$ and $l := k$
    \While {$l \geq i$}
      \For {$j := i-1$ to $1$}
        \If {($d(C_j,C_l) < d(p,C_l)$ and $C_l \in L$) or ($d(C_j,C_l) > d(p,C_l)$ and $C_l \in R$)}
          \State $L \gets C_j \cons L$,\quad $R \gets [C_{j+1},\ldots,C_{i-1}] \append R$, and  $i \gets j$
        \EndIf
        \If {($d(C_j,C_l) < d(p,C_l)$ and $C_l \in R$) or ($d(C_j,C_l) > d(p,C_l)$ and $C_l \in L$)}
          \State $R \gets C_j \cons R$,\quad $L \gets [C_{j+1},\ldots,C_{i-1}] \append L$, and $i \gets j$
        \EndIf
      \EndFor
      \State $l \gets l - 1$
    \EndWhile
    \Return $(\call{reverse}(L), i-1, R)$
  \end{algorithmic}
\end{algorithm}

Let $(X,d)$ be a Robinson space, $p \in X$, and $C_1,\ldots,C_k$ be
the $p$-copoints, and let $\prec$ be a $p$-proximity order on $\C_p$.
We assume that $\{p\} = C_0 \prec C_1 \prec \ldots \prec C_k$. Let
$i \in \{1,\ldots,k\}$. By definition, $C_i$ is a frontier if and only
if for each $j \in \{1,\ldots,i\}$ and each $j' \in \{i+1,\ldots,n\}$, we have 
$d(p,C_{j'}) = d(C_j,C_{j'})$. This allows to determine all the
frontiers in time $O(k^2)$: for each $j' \in \{1,\ldots,k\}$, let
$j \in \{0,\ldots,j'\}$ be maximum such that
$d(p,C_{j'}) = d(C_j,C_{j'})$ and mark $C_{j+1},\ldots,C_{j'-1}$ as
non-frontier. Then all non-marked copoints are frontiers.

Once we know that $C_i$ is a frontier, let $\alpha_i \in \Upsilon(p)$
be the node with $X(\alpha_i) = [C_0,C_i]$, and we can determine the
root of $\alpha_i$, using \Cref{thm:copoint-in-pqtree}.
\begin{enumerate}[label=(\arabic*),labelindent=0pt,labelwidth=50pt,itemindent=20pt,leftmargin=0pt]
\item\label{item:cptpqt1} If $i > 2$ and $C_{i-1}$ is not a frontier,
  then $\alpha_i$ is a $Q$-node. Let $j \in \{1,\ldots,i-1\}$ be
  maximum with $C_j$ a frontier (or $j = 0$ if $C_i$ is the first
  frontier). Then $\alpha_j, \TPQ(C_{j+1}),\ldots,\TPQ(C_{i})$ are the
  children of $\alpha_i$. Indeed those non-frontier copoints
  correspond to the case~\ref{item:cpp2} in
  \Cref{thm:copoint-in-pqtree}. If $\alpha_i$ is conical, then $p$ is
  not in its apex child.
\item\label{item:cptpqt2} Else, $\alpha_i$ is either a P-node, a
  conical Q-node with $\alpha_{i-1}$ apex and not split, or a conical
  Q-node with split child $\beta$ and $\alpha_{i-1}$ a child of
  $\beta$. Let $\delta = d(p,C_i)$, in any of these cases
  $\Gdelta([C_0,C_i])$ is not connected, let $K_1,\ldots,K_m$ be its
  connected component, with $K_1$ having the largest diameter.
 \begin{enumerate}[label=(\arabic{enumi}.\arabic*),labelindent=0pt,labelwidth=50pt,itemindent=20pt,leftmargin=20pt]
 \item\label{item:cptpqt21} If $\diam(K_1) \leq \delta$, there is no
   large component, hence $\alpha_i$ is a P-node
   $P(\beta_1,\ldots,\beta_l)$, and each component corresponds to a
   child of $\alpha_i$ by \Cref{lemma:some-d-cut}. In that case,
   $\alpha_{i-1}$ is the child $\beta_j$ of $\alpha_i$ containing $p$,
   hence exactly one of the following cases hold:
   \begin{enumerate}[label=(\arabic{enumi}.\arabic{enumii}.\arabic*),labelindent=0pt,labelwidth=50pt,itemindent=20pt,leftmargin=20pt]
   \item\label{item:cptpqt211} either $l > 2$ and $\TPQ(C_k) = P(\beta_1,\ldots,\beta_{j-1},\beta_{j+1},\ldots,\beta_l)$;
   \item\label{item:cptpqt212} or $l = 2$ and $\TPQ(C_k) = \beta_{3-j}$.
   \end{enumerate}
   In both cases, $\diam(C_k) \leq \delta$.
 \item\label{item:cptpqt22} Else if $m=2$, then, $\alpha_i$ is a
   $\delta$-conical Q-node whose apex child $\alpha_{i-1}$ is not
   split by \Cref{lemma:cut-separable}. In this case,
   $C_k = X(\alpha_i) \setminus X(\alpha_{i-1})$, and one of the
   following cases occurs:
   \begin{enumerate}[label=(\arabic{enumi}.\arabic{enumii}.\arabic*),labelindent=0pt,labelwidth=50pt,itemindent=20pt,leftmargin=20pt]
   \item\label{item:cptpqt221} either $\TPQ(C_k) = P(\beta_1,\beta_2)$
     and then $\alpha_i = Q(\beta_1,\alpha_{i-1},\beta_2)$;
   \item\label{item:cptpqt222} or
     $\TPQ(C_k) = Q(\beta_1,\ldots,\beta_l)$, and there is
     $j \in \{2,\ldots,l\}$ with \\
     $\alpha_i = Q(\beta_1,\ldots,\beta_{j-1},\alpha_{i-1},\beta_j,\ldots,\beta_k)$;
   \end{enumerate}
   In both cases, $\diam(C_k) > \delta$.
 \item\label{item:cptpqt23} Else, also by \Cref{lemma:cut-separable}, $m > 2$ and
   $\alpha_i$ is a $\delta$-conical Q-node with a split child $\beta$.
   The children of $\beta$ are in one-to-one correspondence with the
   components $K_2,\ldots,K_m$, one of them being $\alpha_{i-1}$. In
   that case, $\TPQ(C_k) = Q(\beta_1,\ldots,\beta_l)$ with apex child
   $\beta_j$, $\diam(C_k) > \delta$, and one of the following cases
   hold:
    \begin{enumerate}[label=(\arabic{enumi}.\arabic{enumii}.\arabic*),labelindent=0pt,labelwidth=50pt,itemindent=20pt,leftmargin=20pt]
    \item\label{item:cptpqt231} either $\beta_j$ is a split child, that is
      $\beta_j = P(\gamma_1,\ldots,\gamma_n)$ with
      $d(X(\gamma_1),X(\gamma_n)) = \delta$, then
      $\alpha_i = Q(\beta_1,\ldots,\beta_{j-1},Q(\gamma_1,\ldots,\gamma_n,\alpha_{i-1}),\beta_{j+1},\ldots,\beta_k)$;
    \item\label{item:cptpqt232} or $\beta_j$ is standard, then
      $\alpha_i = Q(\beta_1,\ldots,\beta_{j-1},P(\beta_j,\alpha_{i-1}),\beta_{j+1},\ldots,\beta_l)$.
    \end{enumerate}
  \end{enumerate}
\end{enumerate}

This leads to \Cref{algo:copoints2pqtree}, that builds the PQ-tree
from the $p$-copoints. We use the notation $[e_1,\ldots,e_n]$ for a
list or sequence of $n$ elements; $[]$ denotes the empty list. The
\emph{append} operation on list is denoted $\append$, while $\cons$ is
the \emph{insert first} operation; thus
$[e_1,\ldots,e_n] \append [f_1,\ldots, f_p] = e_1 \cons ([e_2,\ldots,e_n] \append [f_1,\ldots,f_p])$.
We also use the $\call{map}$ operation, defined by
$\call{map}(f,[e_1,\ldots,e_n]) = [f(e_1),\ldots,f(e_n)]$. Applying
the operation $\cons$ takes constant-time on single-linked list, and
we assume that $\append$ and $\map$ are implemented as linear-time
operations.

\Cref{algo:copoints2pqtree} builds the $p$-copoints using the stable
partition algorithm from~\cite[Algorithm 6.1]{CaCheNaPre}, which
returns the copoints in a universal $p$-proximity order. Then it
builds recursively the PQ-tree $\TPQ(C_k)$, and deduce $\TPQ(X)$ from
the previous analysis. The case analysis presented above justifies
procedure $\call{copointsToPqTree}$. \Cref{algo:copoints2pqtree} also
relies on the procedure $\sortSiblings$ which finds the last frontier
$C_i$ before $C_k$ and sorts the children of the root when this root
is a Q-node. Those children are given by a permutation of the sequence
provided by the $p$-proximity order $[C_0,C_i], C_{i+1},\ldots,C_k$.
The procedure $\sortSiblings$ builds a compatible pre-order from this
$p$-proximity order, by using the inner instructions of the main loop
of~\cite[Algorithm 6.2]{CaCheNaPre}. Consequently the correctness of
$\sortSiblings$ can be derived from the correctness of the
construction in~\cite{CaCheNaPre} of the compatible order from the
$p$-proximity order.

\Cref{algo:copoints2pqtree} can be implemented in $O(|X|^2)$
using~\cite[Algorithm 6.1]{CaCheNaPre} to compute the stable
partition, and~\cite[Algorithm 4.1]{CaCheNaPre} to find the admissible
hole at line~\ref{line:find-adm-hole}. The analysis follows the same
arguments as in~\cite[Theorem 7.2]{CaCheNaPre}. Therefore we can state:

\begin{theorem}
  Given a Robinson space $(X,d)$, \Cref{algo:copoints2pqtree} computes
  $\TPQ(X)$ in time $O(|X|^2)$.
\end{theorem}

\subsection{Copoints and frontiers in ultrametrics}
We end this section with a characterization of ultrametric spaces in
terms of copoints and frontiers:

\begin{proposition}\label{prop:ultrametric_copoints}
  A dissimilarity space $(X,d)$ is ultrametric if and only if for any
  $p \in X$ and any $p$-copoint $C$, $C$ is a frontier and
  $\diam(C) \leq d(p,C)$.
\end{proposition}

\begin{proof}
  Let $p \in X$ and $C_o := \{p\} \prec C_1 \prec \ldots \prec C_k$ be
  the $p$-copoints in increasing universal $p$-proximity order. Let
  $\Upsilon(p) = \{\alpha_0, \alpha_1,\ldots,\alpha_{k'}\}$, sorted by
  decreasing depth; $\alpha_0 = \Leaf\ p$ and $\alpha_{k'} = \TPQ(X)$.

  Suppose that $(X,d)$ is ultrametric. By
  \Cref{prop:ultrametric_pq_tree}, $\TPQ(X)$ contains only P-nodes;
  thus by \Cref{thm:copoint-in-pqtree}, $k = k'$ and
  $C_i = X(\alpha_i) \setminus X(\alpha_{i-1})$ for all
  $i \in \{1,\ldots,k\}$. As $[C_0,C_i] = X(\alpha_i)$, by
  \Cref{LEMMA_block_node} it is an mmodule, hence $C_i$ is a frontier.
  Moreover, for each child $\beta$ of $\alpha_i$ distinct from
  $\alpha_{i-1}$, $X(\beta)$ is a block and
  $d(p,\beta) = d(p,C_i) = \rho(\alpha_i)$, hence
  $\diam(\beta) \leq \rho(\alpha_i)$ and thus
  $\diam(C_i) \leq d(p,C_i)$.

  Conversely, suppose that for each $i \in \{1,\ldots,k\}$, $C_i$ is a
  frontier with $\diam(C_i) \leq d(p,C_i)$. By
  \Cref{lemma:node_upsilon_frontiere}, $k = k'$ and
  $X(\alpha_i) = [C_0,C_i]$ for each $i \in \{1,\ldots,k\}$. Moreover,
  as $C_{i-1}$ is a frontier and $\diam(C_i) \leq d(p,C_i)$, one of
  Cases~\ref{item:cptpqt211} and~\ref{item:cptpqt212} applies. In any
  case, this implies that $\alpha_i$ is a P-node. Therefore all nodes
  in the path from $p$ to the root are P-nodes. As this is true for
  any $p \in X$, all internal nodes of $\TPQ(X)$ are P-nodes, thus by
  \Cref{prop:ultrametric_pq_tree} $(X,d)$ is ultrametric.
\end{proof}

%\end{example}

%% file: sec_conclusion.tex
Our paper establishes a cryptomorphism between the PQ-tree $\TPQ$ and
the mmodule-tree $\TM$ of a Robinson space $(X,d)$. We show how to
derive $\TM$ from $\TPQ$ and, vice-versa, how to construct $\TPQ$ from
$\TM$. We also present optimal $O(|X|^2)$ time algorithms for
constructing the $\TPQ$ (without ordering the children of Q-nodes) and
$\TM$. Our proofs and algorithms use two technical ingredients: the
$\delta$-graph $\Gdelta$ of $(X,d)$ with the properties of its
connected components and the dendrogram $\Td$ of the ultrametric
subdominant $\hd$ of $(X,d)$. We also show how to construct in
$O(|X|^2)$ the PQ-tree $\TPQ$ with the correct ordering of the
children of Q-nodes from the copoint partitions and frontiers of
$(X,d)$; this algorithm is based on the concept of universal
$p$-proximity order and the construction from this order of a
compatibility order presented in our previous paper~\cite{CaCheNaPre}.

Our $O(|X|^2)$ time algorithm for constructing the PQ-tree $\TPQ$ of a
Robinson space $(X,d)$ is simpler than the algorithm of \cite{Prea}
which also constructs $\TPQ$ in $O(|X|^2)$ time in order to enumerate
all compatible orders of $(X,d)$. At the difference of \cite{Prea},
our algorithm does not use the quite complex algorithm of Booth and
Lueker \cite{BoothLueker} as a subroutine. On the other hand, our
top-to-bottom $O(|X|^2)$ time algorithm for constructing the mmodule
tree $\TM$ is more involved than the algorithm of~\cite{EhGaMcCSu}.
Nevertheless, it uses some surprising links between the trees $\TPQ$
and $\TM$ on the one hand and the dendrogram $\Td$ on the other hand.

\section*{Acknowledgements}

{This research was supported in part by ANR project
    DISTANCIA (ANR-17-CE40-0015) and has received funding from
    Excellence Initiative of Aix-Marseille - A*MIDEX (Archimedes
    Institute AMX-19-IET-009), a French ``Investissements d’Avenir”
    Programme.}
%
%
%To establish this cryptomorphism between PQ-trees and mmodules trees, additionally
%to the previous notions, we introduce the notion of  $\delta$-graph $\Gdelta$ of
%a Robinson space $(X,d)$. We prove that either $\Gdelta$ is connected for all $\delta>0$ or there exists
%a unique value of $\delta$ for which $\Gdelta$ is not connected. In the later case,
%the connected components of $\Gdelta$ are called $\delta$-mmodules. The dichotomy between
%the connectivity for all $\delta$ and  the non-connectivity for some $\delta$ of $\Gdelta$ and
%the $\delta$-mmodules in the second case are crucial in the construction of the PQ-tree and the translations between PQ-trees and mmodule trees.
%Since this paper is a follow-up of~\cite{CaCheNaPre}, we refer
%to the paper~\cite{CaCheNaPre} for a complete bibliography on Robinson spaces
%and on their recognition algorithms (in this paper, we cite only the recognition algorithms using PQ-trees).

%% file: appendix.tex
\section{Appendices}

\subsection{Construction of the subdominant ultrametrics}

In this subsection we show how to construct the dendrogram $\Td$ of
the ultrametric space $(X,\widehat{d})$; recall that $\widehat{d}$ is
the subdominant ultrametric, i.e., the largest ultrametric such that
$\widehat{d}(x,y)\le d(x,y)$ for any $x,y\in X$. We can construct
$\Td$ in the following iterative way (which seems us to be new): when
using Prim's algorithm to compute the minimum spanning tree $T$ of
$(X,d)$, one can build $\Td$ by inserting each vertex in $\Td$ at the
moment when it is visited, leading to \Cref{algo:dendo}.

\begin{algorithm}[h]
  \caption{Computes the dendrogram $\Td$ of the subdominant ultrametric of a dissimilarity space $(X,d)$.}
  \label{algo:dendo}

  \flushleft{$\underline{\dendo(X,d)}$}
  \begin{algorithmic}[1]
    \Require{A dissimilarity space $(X,d)$}
    \Ensure{The dendrogram $\Td$ of the subdominant ultrametric $\widehat{d}$ of $(X,d)$}
    \Let $s \in X$
    \Let $\Td:= \Leaf\ s$
    \Let $p(u) := +\infty$ for all $u \in X \setminus \{s\}$
    \State $\propagate(p, s)$
    \While {there is an unvisited vertex}
      \Let $u$ be an unvisited vertex with $p(u)$ minimum
      \State $u$ is now visited
      \State $\Td \gets \insertDendo(u,p(u),\Td)$
      \State $\propagate(p,u)$
    \EndWhile
    \Return $\Td$
  \end{algorithmic}\bigskip

  \flushleft{$\underline{\propagate(p,u)}$}
  \begin{algorithmic}[1]\setcounter{ALG@line}{10}
    \For {$v \in X$}
      \If {$v$ is unvisited}
        \State $p(v) \gets \min \{ p(v), d(u,v) \}$
      \EndIf
    \EndFor
  \end{algorithmic}\bigskip

  \flushleft{$\underline{\insertDendo(u,\rho,T)}$}%% x--> u
  \begin{algorithmic}[1]\setcounter{ALG@line}{13}
    \Require{A dissimilarity space $(X,d)$ (implicit), a dendrogram $T$ %% $\Td$
      on $S \subseteq X$ obtained at some step of $\dendo(X,d)$, $u \in X \setminus S$
      minimizing $\min_{x \in S} d(u,x)$, and
      $\rho := \min \{d(u,x) : x \in S\}$}
    \Ensure{A dendrogram on $S \cup \{u\}$}
    \Match{T}
    \Case{$\Leaf\ x$}
    \Return $\Node(\rho,[\Leaf\ u, \Leaf\ x])$
    \Case{$\Node(\rho', \variable{children})$ \When\ $\rho' < \rho$}
    \Return $\Node(\rho, [\Leaf\ u, T])$
    \Case{$\Node(\rho', \variable{children})$ \When\ $\rho' = \rho$}
    \Return $\Node(\rho, \Leaf\ u : \variable{children})$
    \Case{$\Node(\rho', \variable{child}:\variable{children})$ \When\
      $\rho' > \rho$} \label{line:dendo_rho_petit} \Return
    $\Node(\rho', \insertDendo(u,\rho,\variable{child}) : \variable{children})$
    \EndMatch
  \end{algorithmic}
\end{algorithm}

This algorithm has a pretty visualization that we explain before
giving a formal proof. It builds a diagram where each point is a
column, sorted left-to-right in visiting order (in the diagram, the
order of the leaves is mirrored from that of the tree). Each point $u$
casts a vertical line upward, that bends by a right angle at height
$p(u)$ to join the rest of the diagram (the vertical line from the
source $s$ does not bend). This is illustrated in
\Cref{fig:dendrogram-sketch}, where the diagram for the dissimilarity
from \Cref{TABLE_gros_example} is given. This diagram is topologically
equivalent to the dendrogram, given in \Cref{fig:dendrogram}. The
$\insertDendo$ procedure traverses down the leftmost branch of the
tree, until finding the correct height at which the new leaf is added.

\begin{figure}[htbp]
  \begin{tabular}{lp{1cm}r}
    \begin{minipage}{0.4\textwidth}
      \footnotesize
      $$\begin{array}{ccccccccccccc}
        \vspace{0.1cm} D & 10 & 8 & 6 & 3 & 7 & 1 & 9 & 4 & 12 & 2 & 11 & 5 \\
             10 & 0 & 2 & 6 & 8 & 6 & 8 & 2 & 8 & 2 & 8 & 2 & 6 \\
             8  &   & 0 & 6 & 8 & 6 & 8 & 1 & 8 & 3 & 8 & 2 & 6 \\
             6  &   &   & 0 & 5 & 1 & 5 & 6 & 5 & 6 & 5 & 6 & 1 \\
             3  &   &   &   & 0 & 5 & 2 & 8 & 2 & 8 & 1 & 8 & 5 \\
             7  &   &   &   &   & 0 & 5 & 6 & 5 & 6 & 5 & 6 & 1 \\
             1  &   &   &   &   &   & 0 & 8 & 3 & 8 & 2 & 8 & 5 \\
             9  &   &   &   &   &   &   & 0 & 8 & 2 & 8 & 2 & 6 \\
             4  &   &   &   &   &   &   &   & 0 & 8 & 2 & 8 & 5 \\
             12 &   &   &   &   &   &   &   &   & 0 & 8 & 2 & 6 \\
             2  &   &   &   &   &   &   &   &   &   & 0 & 8 & 5 \\
             11 &   &   &   &   &   &   &   &   &   &   & 0 & 6 \\
             5  &   &   &   &   &   &   &   &   &   &   &   & 0
      \end{array}$$
    \end{minipage} & &
    \begin{minipage}{0.5\textwidth}
      \begin{tikzpicture}[x=0.6cm,y=0.6cm]
        \foreach \i/\y/\w in {1/8/1,2/2/1,3/1/2,4/2/1,5/5/1,6/1/5,7/1/5,8/6/1,9/1/8,10/2/8,11/2/8,12/2/8} {
          \draw (\i,0) node[anchor=north] {\i};
          \draw (\i,0) |- (\w,\y);
        }
        \foreach \y in {1,2,...,7} {
          \draw[dotted] (0.5,\y) -- (12.5,\y);
        }
        \foreach \y in {1,2,5,6} {
          \draw (12.5,\y) node[right] {\tiny \y};
        }
      \end{tikzpicture}
    \end{minipage}
  \end{tabular}
  \caption{The dissimilarity matrix from \Cref{TABLE_gros_example} on
    the left, the diagram built to illustrate \Cref{algo:dendo} on the
    right. The rows and columns of the matrix are shuffled to avoid
    masking some of the work of the algorithm.}
  \label{fig:dendrogram-sketch}
\end{figure}
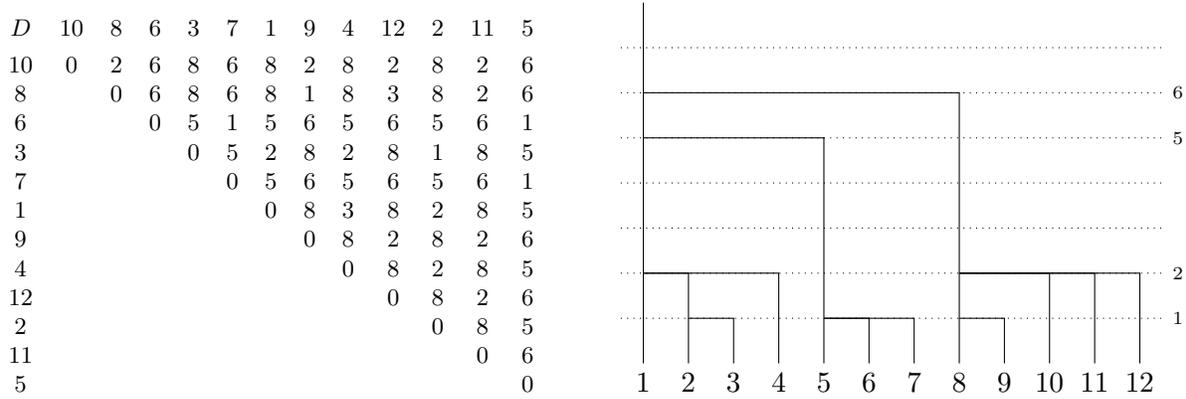

In order to analyse the algorithm, recall that $\hd(x,y)$ is the
minimum over all $(x,y)$-paths of the maximum 
weight of the edges of that path. This is known as the \emph{bottleneck
shortest path}, for which a good characterization is well-known
(see~\cite[Theorem 8.17]{Schrijver} for instance). In our context,
that good characterization is:

\begin{theorem}[\cite{fulkerson1966flow}]\label{thm:bottleneck}
  Let $(X,d)$ be a dissimilarity space and $\hd$ be its ultrametric
  subdominant. Then for all distinct $x, y \in X$,
  \begin{displaymath}
    \hd(x,y) = \min \{ \max_{uv \in P} d(u,v) : \textrm{$P$ is an
    $(x,y)$-path}\} = \max \{ \min_{u \in S, v \in X \setminus S}
  d(u,v) : S \subset X, x \in S, y \notin S\}.
  \end{displaymath}
\end{theorem}

Let the points be sorted in visiting order, $x < y$ means that $x$ is
visited before $y$. We let $p^\star : X \to \mathbb{R}_{\geq 0}$
denote the final values of $p$, taking $p^\star(s) = +\infty$. We
prove by induction that at the end of each iteration, for all visited
$x,y$, $\hd(x,y)$ is the weight of the lowest common ancestor of $x$
and $y$ in $\Td$. Also let $S_v := \{ x \in X : v < x\}$ be the set of
points unvisited after the iteration in which $u = v$. 

\begin{claim}\label{claim:Sv-cut}
  Let $v \in X$. For each $x \in S_v$ and $y \in X \setminus S_v$,
  $d(x,y) \geq p^\star(v)$.
\end{claim}

\begin{proof}
  Consider the iteration for which $u = v$. Then by choice of $v$,
  $p(v)$ is minimum, in particular $p(v) \leq p(y) \leq d(x,y)$. As
  $v$ is now visited, $p(v) = p^\star(v)$, and the claim follows.
\end{proof}

We consider an iteration of the main loop in $\dendo$, when vertex
$u$ is visited; let $u \in X \setminus \{s\}$. Let $w$ be the maximum
vertex (in visiting order) with $w < u$ and $p(w) > p^\star(u)$ (the
column for $w$ is the one to which $u$ is attached).

\begin{claim}\label{claim:exists-v}
   There exists $v \in X$ with $w \leq v < u$ and $d(u,v) = p^\star(u)$.
\end{claim}

\begin{proof}
  There exists $v \in X$ such that $d(u,v) = p^\star(u)$ by definition
  of $p$ and $v < u$. Suppose by contradiction that $v \leq w$, then
  during the iteration when $w$ is visited, $p(u) = p^\star(u) <
  p^\star(w)$, contradicting the choice of $w$ during that iteration.
\end{proof}

\begin{claim}\label{claim:after-w}
  For all $x \in X$ with $w \leq x < u$, $\hd(x,u) = p^\star(u)$.
\end{claim}

\begin{proof}
  By \Cref{claim:Sv-cut,thm:bottleneck}, $\hd(u,x) \geq p^\star(u)$.
  Let $v$ be the vertex from \Cref{claim:exists-v}. Then $\hd(u,x)
  \leq \max \{\hd(u,v),\hd(v,x)\} = \max \{p^\star(u),\hd(v,x)\}$. By
  induction, $\hd(v,x)$ is the height of the lowest common ancestor of
  $v$ and $x$, that cannot be more than $p^\star(u)$ as by
  construction the part of the diagram between $w$ and $u$ is below
  $p^\star(u)$. The claim follows.
\end{proof}

\begin{claim}\label{claim:before-w}
  If $w \neq s$, then for all $x \in X$ with $x < w$, we have $\hd(x,u) = \hd(x,w)$.
\end{claim}

\begin{proof}
  There exists $z \leq w$ with $p^\star(z) = \hd(x,w)$ by construction
  of $p$. Then by \Cref{claim:Sv-cut}, $\min_{t < z \leq y} d(t,y) =
  p^\star(z)$, thus by \cref{thm:bottleneck}, $\hd(x,u) \geq p^\star(z)
  = \hd(x,w)$. On the other hand, $\hd(x,u) \leq
  \max\{\hd(x,w),\hd(w,u)\} = \hd(x,w)$, proving the claim.
\end{proof}

Together, \Cref{claim:after-w,claim:before-w} proves that $\Td$ is
correctly computed. As for the complexity, observe that $\dendo$, when
removing the lines about $\Td$, is exactly Prim's algorithm to compute
a minimum spanning tree (except that we compute neither the tree nor
the predecessor of each vertex in the tree). Moreover the time spent
to build $\Td$ is $O(|X|^2)$, as each of the $|X|$ insertions requires
at most $O(|X|)$ operations (that part could be optimized to $O(|X|)$
by using a zipper~\cite{huet1997zipper}, thus avoiding going back to
the root at each insertion). Therefore we have proved:

\begin{proposition}\label{prop:td-analysis}
  Let $(X,d)$ be a dissimilarity space. Then $\dendo(X,d)$
  computes $\Td(X)$ in time $O(|X|^2)$. \qedhere
\end{proposition}

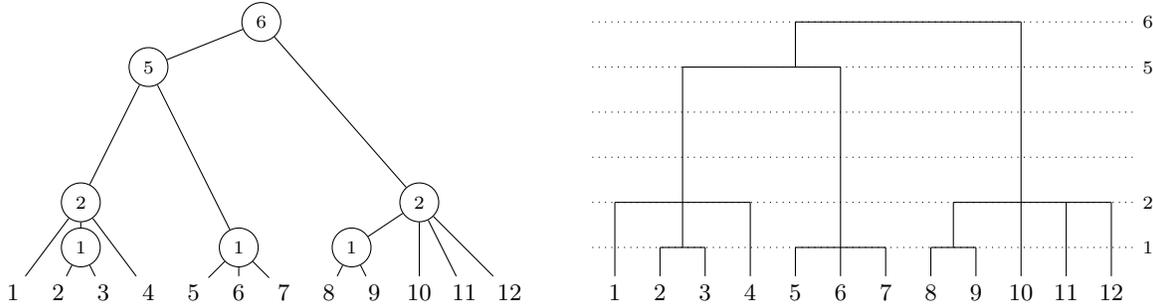
\begin{figure}[htbp]
  \begin{center}\footnotesize
    \begin{tikzpicture}[x=0.6cm,y=0.6cm]
    \begin{scope}
      \foreach \i in {1,2,...,12} {
        \node (p\i) at (\i,0) {\i};
      }
      \foreach \i/\x/\y in {a/2.5/1,b/6/1,c/8.5/1,d/2.5/2,e/10/2,f/4/5,g/6.5/6} {
        \node[circle,draw=black] (\i) at (\x,\y) {\tiny $\y$};
      }
      \foreach \u/\v in {p1/d,p2/a,p3/a,p4/d,p5/b,p6/b,p7/b,p8/c,p9/c,p10/e,p11/e,p12/e,a/d,b/f,c/e,d/f,e/g,f/g} {
        \draw (\u) -- (\v);
      }
      \end{scope}
      \begin{scope}[xshift=8cm]
        \foreach \x/\y in {1/2,2/1,3/1,4/2,5/1,6/1,7/1,8/1,9/1,10/2,11/2,12/2} {
          \draw (\x,0) node (u\x) {\x};
          \draw (u\x.north) -- (\x,\y);
        }
        \foreach \u/\v/\y in {1/4/2,2/3/1,5/7/1,2.5/6/5,8/9/1,8.5/12/2,5/10/6} {
          \draw (\u,\y) -- (\v,\y);
        }
        \draw (2.5,1) -- (2.5,5);
        \draw (6,1) -- (6,5);
        \draw (8.5,1) -- (8.5,2);
        \draw (10,2) -- (10,6);
        \draw (5,5) -- (5,6);
        \foreach \y in {1,2,...,6} {
          \draw[dotted] (0.5,\y) -- (12.5,\y);
        }
        \foreach \l in {1,2,5,6} {
          \draw (12.5,\l) node[right] {\tiny \l};
        }
      \end{scope}

    \end{tikzpicture}
  \end{center}
  \caption{The dendrogram associated to \Cref{fig:dendrogram-sketch},
    with the tree representation on the left (actually the mirror of
    the tree returned by \Cref{algo:dendo}), and the more traditional
    dendrogram representation on the right.}
  \label{fig:dendrogram}
\end{figure}

%% \begin{proof}
%%   Let the points be sorted in visiting order, $u < v$ means that $u$
%%   is visited before $v$. We let $p : X \to \mathbb{R}_{\geq 0}$ denote
%%   the final values of $p$, taking $p(s) = +\infty$. The effect of
%%   $\insertDendo(u,\rho,T)$ is to recursively go down on the leftmost
%%   child, until reaching a node whose value is at most $\rho$. Then $u$
%%   is inserted around this position. We have the following invariant
%%   that is easily checked:

%%   \begin{claim}[Invariant]
%%     The leaves of $\Td$ are the visited points sorted from left to
%%     right in decreasing order.%\qedhere
%%   \end{claim}

%%   Thus for any subtree, its rightmost (resp. leftmost) point is its
%%   first (resp. last) visited point. For a node $\alpha$ in $\Td$, we
%%   denote $\rho(\alpha)$ the value $\rho$ stored in that node. The next
%%   two invariants are also immediate by an induction on the calls to
%%   $\insertDendo$:

%%   \begin{claim}[Invariant]
%%     $\rho$ is decreasing from the root to any leaf. %\qedhere
%%   \end{claim}

%%   \begin{claim}[Invariant]
%%     For any node $\alpha = Node(\rho,[\beta_k,\ldots,\beta_1])$, let
%%     $x$ be the rightmost leaf in $\beta_2$. Then $\rho = p(x)$. \qedhere
%%   \end{claim}

%% \end{proof}

\subsection{Stable partition algorithm} In this subsection, we present the algorithm which refines any partition to a stable partition.

\begin{algorithm}[h]
  \caption{Computes the refinement of a partition of $(X,d)$ into a stable partition.}
  \label{algo:stable}

  \flushleft{$\underline{\stablePart(\PP)}$}
  \begin{algorithmic}[1]
    \Require{a dissimilarity space $(X,d)$ (implicit), a partition
      $\PP := \{S_1,\ldots,S_k\}$ of $X$.}
    \Ensure{a partition $\PP'$ of $X$ that refines $\PP$ and such that
      each part $M \in \PP'$ is an mmodule.}
    \For{$i \in \{1,\ldots,k\}$}
       \YieldAll $\partRefine(S_i, X \setminus S_i)$
    \EndFor
  \end{algorithmic}\bigskip

  \flushleft{$\underline{\partRefine(B,Z(B))}$}
  \begin{algorithmic}[1]\setcounter{ALG@line}{1}
  \Require{a dissimilarity space $(X,d)$ (implicit),
    a class $B \subseteq X$ and a set $Z(B) \subseteq X \setminus B$.}
  \Ensure{a partition $\{B_1,B_2,\ldots,B_k\}$ of $B$.}
  \If{$Z(B) = \emptyset$}
      \Return $\{ B \}$
  \EndIf
  \Let $q \in Z(B)$,    \Comment{choose $q$ to be the first element of $Z(B)$}
  \Let $\{B_1,\ldots,B_m\} = \refine(q,S)$ \Comment{ignore the order of the $B_i$s}
  \For {$i \in \{1,\ldots,m\}$}
      \YieldAll $\partRefine(B_i, \concatenate(B_1, \ldots, B_{i-1}, B_{i+1}, \ldots, B_m, Z(B) \setminus \{q\}))$
  \EndFor
  \end{algorithmic}\bigskip

  \flushleft{$\underline{\refine(q,S)}$}
  \begin{algorithmic}[1]\setcounter{ALG@line}{8}
    \Require{a dissimilarity space $(X,d)$ (implicit),
      a point $q \in X$, a subset $S \subseteq X$.}
    \Ensure{an ordered partition of $S$, by increasing distances from $q$.}
    \Let $T$ be an empty balanced binary tree, with keys in $\mathbb{N}$
    \For {$x \in S$}
      \If {$\lnot \containsKey(T,d(q,x))$}
        \State $\insertDico(T,d(q,x),[])$
      \EndIf
      \State $\insertDico(T,d(q,x), x \cons \getDico(T,d(q,x)))$
    \EndFor
    \Return $\values(T)$
  \end{algorithmic}
\end{algorithm}

\begin{lemma}\label{lemma:stable-correct}
  Let $(X,d)$ be a dissimilarity space, $S \subseteq X$, and $\PP$ a
  partition of $S$. Then the partition $\stablePart(\PP)$ consists of
  maximal by inclusion mmodules of $S$ contained in a class of $\PP$.
\end{lemma}

\begin{proof}
  Let $M \in \stablePart(\PP)$. By correction of \Cref{algo:stable},
  $M$ is an mmodule, and since $\stablePart(\PP)$ is a refinement of
  $\PP$, there is $P \in \PP$ such that $M \subseteq P$.
  Conversely, let $M$ be a maximal mmodule of $S$ contained in some
  class of $\PP$. Since for any $z \in S \setminus M$ and any
  $x,y \in M$, we have $d(x,z) = d(y,z)$, one can check that
  \Cref{algo:stable} cannot separate $x$ from $y$, that is there is a
  set $P \in \stablePart(\PP)$ such that $M \subseteq P$. By
  maximality of $M$, we get $M = P$, concluding the proof of the lemma.
\end{proof}

%% file: Modules_PQ-tree.bbl
\begin{thebibliography}{10}

\bibitem{AtBoHe}
{\sc J.~E. Atkins, E.~G. Boman, and B.~Hendrickson}, {\em A spectral algorithm
  for seriation and the consecutive ones problem}, SIAM J. Comput., 28 (1998),
  pp.~297--310.

\bibitem{BoothLueker}
{\sc K.~S. Booth and G.~S. Lueker}, {\em Testing for the consecutive ones
  property, interval graphs, and graph planarity using {P}{Q}-tree algorithms},
  J. Comput. Syst. Sci., 13 (1976), pp.~335--379.

\bibitem{CaCheNaPre}
{\sc M.~Carmona, V.~Chepoi, G.~Naves, and P.~Pr{\'e}a}, {\em Modules in
  robinson spaces}, 2023.
\newblock submitted.

\bibitem{CrFi}
{\sc F.~Critchley and B.~Fichet}, {\em The partial order by inclusion of the
  principal classes of dissimilarity on a finite set, and some of their basic
  properties}, in Classification and Dissimilarity Analysis, Springer, 1994,
  pp.~5--65.

\bibitem{Di}
{\sc E.~Diday}, {\em Orders and overlapping clusters in pyramids.}, in
  Multidimensional Data Analysis, J.~e.~a. de~Leeuw, ed., DSWO Press, Leiden,
  1986.

\bibitem{DuFi}
{\sc C.~Durand and B.~Fichet}, {\em One-to-one correspondences in pyramidal
  representation: A unified approach.}, in Conf. of IFCS, 1987, pp.~85--90.

\bibitem{EhGaMcCSu}
{\sc A.~Ehrenfeucht, H.~N. Gabow, R.~M. McConnell, and S.~Sullivan}, {\em An
  ${O}(n^2)$ divide-and-conquer algorithm for the prime tree decomposition of
  two-structures and modular decomposition of graphs}, J. Algorithms, 16
  (1994), pp.~283--294.

\bibitem{Ev}
{\sc B.~Everitt}, {\em Cluster analysis}, Wiley, 2011.

\bibitem{fulkerson1966flow}
{\sc D.~R. Fulkerson}, {\em Flow networks and combinatorial operations
  research}, The American Mathematical Monthly, 73 (1966), pp.~115--138.

\bibitem{GoRo}
{\sc J.~Gower and G.~Ross}, {\em Minimum spanning trees and single linkage
  cluster analysis}, Journal of the Royal Statistical Society, Series C, 18
  (1969), pp.~54--64.

\bibitem{huet1997zipper}
{\sc G.~P. Huet}, {\em The zipper}, J. Funct. Program., 7 (1997), pp.~549--554.

\bibitem{MiRo}
{\sc B.~Mirkin and S.~Rodin}, {\em Graphs and {G}enes}, Springer, 1984.

\bibitem{Prea}
{\sc P.~Pr{\'e}a and D.~Fortin}, {\em An optimal algorithm to recognize
  {R}obinsonian dissimilarities}, J. Classif., 31 (2014), pp.~351--385.

\bibitem{Robinson}
{\sc W.~S. Robinson}, {\em A method for chronologically ordering archaeological
  deposits}, Amer. Antiq., 16 (1951), pp.~293--301.

\bibitem{Schrijver}
{\sc A.~Schrijver}, {\em Combinatorial optimization: polyhedra and efficiency},
  vol.~24, Springer, 2003.

\bibitem{SeSt}
{\sc C.~Semple and M.~Steel}, {\em Phylogenetics}, Oxford University Press,
  Oxford, 2003.

\bibitem{SestonThese}
{\sc M.~Seston}, {\em Dissimilarit{\'e}s de {R}obinson: algorithmes de
  reconnaissance et d'approximation}, PhD thesis, Universit{\'e} de la
  M{\'e}diterran{\'e}e, Aix-Marseille 2, 2008.

\end{thebibliography}
